%% file: towards.tex
\documentclass[review,11pt]{elsarticle}
\input{preamble}

\usepackage{lineno,hyperref}

\journal{Journal of Computer and System Sciences}

\begin{document}

\begin{frontmatter}


\title{Incompressibility of $H$-free edge modification problems: \\Towards a dichotomy\tnoteref{mytitlenote}}
\tnotetext[mytitlenote]{A preliminary version of this paper has appeared in the proceedings of European Symposium on Algorithms 2020.}

\author{D\'{a}niel Marx\fnref{myfootnote1}\corref{mycorrespondingauthor}}
\address{CISPA Helmholtz Center for Information Security, Germany}
\ead{marx@cispa.saarland}

\author{R.\,B.\, Sandeep\fnref{myfootnote1,myfootnote2}\corref{mycorrespondingauthor}}
\cortext[mycorrespondingauthor]{Corresponding author}
\address{Department of Computer Science and Engineering, Indian Institute of Technology Dharwad, India}
\ead{sandeeprb@iitdh.ac.in}
\fntext[myfootnote1]{Supported by the European Research Council (ERC) grant SYSTEMATICGRAPH: ``Systematic mapping of the complexity landscape of hard algorithmic graph problems", reference 725978}

\fntext[myfootnote2]{Supported by Science and Engineering Research Board (SERB) India grant SRG/2019/002276: ``Complexity dichotomies for graph modification problems"}



\input{abstract}


\end{frontmatter}


\input{introduction}
\input{preliminaries}

\input{larger-gang}
\input{reductions}
\input{cai}
\input{conclusion}




\end{document}

%% file: preamble.tex
\usepackage{scalerel}

\usepackage{tikz,amsmath, amssymb,bm,color, amsthm}
\usetikzlibrary{positioning, calc}
\usepackage{enumerate}
\usepackage{comment}
\usepackage{tikz}
\usepackage{graphics}
\usepackage{longtable}
\usepackage{mdframed}
\usepackage{caption}
\usepackage{subcaption}
\usepackage{slashbox}
\usepackage{url}
\usepackage{framed}
\usepackage{array}
\usepackage{tabu}
\usepackage{lscape}
\usepackage{multirow}
\usepackage{ulem}
\usepackage{multicol}
\usepackage{placeins}
\usepackage{setspace}
\mdfsetup{skipabove=2pt,skipbelow=2pt}
\usepackage{geometry}
\newtheorem{theorem}{Theorem}[section]
\newtheorem{lemma}[theorem]{Lemma}
\newtheorem{observation}[theorem]{Observation}
\newtheorem{corollary}[theorem]{Corollary}
\newtheorem{proposition}[theorem]{Proposition}

\newtheorem{construction}{Construction}
\newtheorem{conjecture}{Conjecture}

\newcommand{\pname}[1]{\textnormal{\textsc{#1}}}
\newcommand{\cclass}[1]{\textnormal{\textsf{#1}}}
\newcommand{\nog}{nine} 
\newcommand{\nogd}{nineteen} 
\newcommand{\simulates}{simulates} %
\newcommand{\baseset}{base} %
\newcommand{\issimulatedby}{is simulated by} %

\newcommand{\HED}{\pname{${H}$-free Edge Deletion}}

\newcommand{\HEE}{\pname{${H}$-free Edge Editing}}
\newcommand{\HEC}{\pname{${H}$-free Edge Completion}}
\newcommand{\HDEE}{\pname{${H'}$-free Edge Editing}}
\newcommand{\HDDEE}{\pname{${H''}$-free Edge Editing}}
\newcommand{\HDED}{\pname{${H'}$-free Edge Deletion}}
\newcommand{\HDEC}{\pname{${H'}$-free Edge Completion}}
\newcommand{\HBEE}{\pname{${\overline{H}}$-free Edge Editing}}
\newcommand{\HBED}{\pname{${\overline{H}}$-free Edge Deletion}}
\newcommand{\HBEC}{\pname{${\overline{H}}$-free Edge Completion}}

\newcommand{\PFS}{\pname{Propagational-$f$ Satisfiability}}
\newcommand{\RHED}{\pname{Restricted ${H}$-free Edge Deletion}}
\newcommand{\RHEC}{\pname{Restricted ${H}$-free Edge Completion}}
\newcommand{\RHDED}{\pname{Restricted ${H'}$-free Edge Deletion}}
\newcommand{\RHDEC}{\pname{Restricted ${H'}$-free Edge Completion}}

\newcommand{\NOPH}{$\cclass{NP} \not\subseteq \cclass{coNP/poly}$}
\newcommand{\LG}{\mathcal{W}}
\newcommand{\LGD}{\mathcal{W}'}

\renewcommand\vee{\boxtimes}

\newcounter{rowcntr}[table]
\renewcommand{\therowcntr}{\thetable.\arabic{rowcntr}}

\newcolumntype{N}{>{\refstepcounter{rowcntr}\therowcntr}c}

\AtBeginEnvironment{longtabu}{\setcounter{rowcntr}{0}}

\newcounter{rowcntra}[table]
\renewcommand{\therowcntra}{\arabic{rowcntra}}

\newcolumntype{M}{>{\refstepcounter{rowcntra}\therowcntra}c}

\AtBeginEnvironment{tabular}{\setcounter{rowcntra}{0}}

%% file: abstract.tex
\begin{abstract}
 Given a graph $G$ and an integer $k$, the \HEE\ problem is to find whether there exist at most $k$ pairs 
  of vertices in $G$ such that changing the adjacency of the pairs in~$G$ results in a graph without any
  induced copy of $H$. The existence of polynomial kernels for \HEE\ (that is, whether it is possible to reduce the size of the instance  to $k^{O(1)}$ in polynomial time) received significant attention in the parameterized complexity literature. Nontrivial polynomial kernels are known to exist for some graphs~$H$ with at most 4 vertices (e.g., path on 3 or 4 vertices, diamond, paw), but starting from 5 vertices, polynomial kernels are known only if $H$ is either complete or empty. This suggests the conjecture that there is no other $H$ with at least 5 vertices where \HEE\ admits a polynomial kernel. Towards this goal, 
  we obtain a set $\mathcal{H}$
  of \nog\ 5-vertex graphs 
  such that if for every $H\in\mathcal{H}$, \HEE\ is incompressible and the complexity assumption \NOPH\ holds, then \HEE\ is incompressible for every graph $H$
  with at least five vertices that is neither complete nor empty. That is, proving incompressibility for these \nog\ graphs would give a complete classification of the kernelization complexity of \HEE\ for every $H$ with at least 5 vertices.
  
  We obtain similar result also for \HED. Here the picture is more complicated due to the existence of another infinite family of graphs $H$ where the problem is trivial (graphs with exactly one edge). We obtain a larger set $\mathcal{H}$ of \nogd\ graphs whose incompressibility would give a complete classification of the kernelization complexity of \HED\ for every graph $H$ with at least 5 vertices. Analogous results follow also for the \HEC\ problem by simple complementation.
\end{abstract}

%% file: introduction.tex
\section{Introduction}
\label{sec:intro}

In a typical graph modification problem, the input is a graph $G$ and an integer $k$, and the task is to perform at most $k$ allowed editing operations on $G$ to make it belong to a certain graph class or satisfy a certain property. 
For example, \pname{Vertex Cover} (remove~$k$ vertices to make the graph edgeless), 
\pname{Feedback Vertex Set} (remove~$k$ vertices to make the graph acyclic), 
\pname{Odd Cycle Transversal} (remove $k$ edges/vertices to make the graph bipartite), 
\pname{Minimum Fill-in} (add $k$ edges to make the graph chordal), and
\pname{Cluster Editing} (add/remove $k$ edges to make the graph a disjoint union of cliques)
are particularly well-studied members of this problem family. 
Most natural graph modification problems are known to be NP-hard, 
in fact, there are general complexity results showing that large families of problems are hard
\cite{DBLP:journals/jcss/LewisY80,DBLP:journals/siamcomp/Yannakakis81,DBLP:journals/siamcomp/Yannakakis81a}. On the other hand, most of these problems are fixed-parameter tractable (FPT) parameterized by $k$: it can be solved in time $f(k)n^{O(1)}$, where $f$ is a computable function depending only on $k$ \cite{DBLP:journals/ipl/Cai96,DBLP:journals/algorithmica/CaoM16,DBLP:journals/jcss/GuoGHNW06,DBLP:conf/stoc/KawarabayashiR07}. 
Looking at the parameterized complexity literature, one can observe that, even though there are certain recurring approaches and techniques, these FPT results are highly problem specific, and often rely on a very detailed understanding of the graph classes at hand.   

A class of problems that can be treated somewhat more uniformly is \HEE. This is a separate problem for every fixed graph~$H$: given a graph~$G$ and an integer $k$, the task is to find whether there exist at most $k$ pairs of vertices in $G$ such that changing the adjacency of the pairs in $G$ results in a graph without any induced copy of~$H$. Aravind et al.~\cite{AravindSS17} proved that \HEE\ is NP-hard for every graph $H$ with at least 3 vertices. However, a simple application of the technique of bounded-depth search trees shows that \HEE\ is FPT parameterized by $k$ for every fixed~$H$ \cite{DBLP:journals/ipl/Cai96}. 

Graph modification problems were explored also from the viewpoint of polynomial kernelization: is there a polynomial-time preprocessing algorithm that does not necessarily solve the problem, but at least reduces the size of the instance to be bounded by a polynomial of $k$? The existence of a polynomial kernelization immediately implies that the problem is FPT (after the preprocessing, one can solve the reduced instance by brute force or any exact method). Therefore, one can view polynomial kernelization as a special type of FPT result that tries to formalize the question whether the problem can be efficiently preprocessed in a way that helps exhaustive search methods.
There is a wide literature on algorithms for kernelization (see, e.g., \cite{fomin_lokshtanov_saurabh_zehavi_2019}). Conversely, incompressibility results can show, typically under the complexity assumption \NOPH, that a parameterized problem has no polynomial kernelization. 

Most of the highly nontrivial FPT algorithms for graph modification problems do not give kernelization results and, in many cases, it required significant amount of additional work to obtain kernelization algorithms. In particular, the FPT algorithm for \HEE\ based on the technique of bounded-depth search trees does not give polynomial kernels. For the specific case when $H=K_r$ is a complete graph, it is easy to see that there is a solution using only deletions. Now the problem essentially becomes a \pname{Hitting Set} problem with sets of bounded size: we have to select at least one edge from the edge set of each copy of $K_r$. Therefore, known kernelization results for \pname{Hitting Set} can be used to show that \pname{$K_r$-free Edge Editing} has a polynomial kernel for every fixed~$r$. A similar argument works if $H$ is an empty graph on $r$ vertices.

Besides cliques and empty graphs, it is known for certain graphs $H$ of at most 4 vertices (diamond \cite{cai2012polynomial, CRSY18polynomial}, path \cite{CaoC12, GrammGHN03,GuillemotHPP13}, paw \cite{CaoKY20, EibenLS15}, and their complements) that \HEE\ has a polynomial kernel, but these algorithms use very specific arguments exploiting the structure of $H$-free graphs. As there is a very deep known structure theory of claw-free (i.e., $K_{1,3}$-free) graphs~\cite{DBLP:journals/jct/ChudnovskyS07a,DBLP:journals/jct/ChudnovskyS08a,DBLP:journals/jct/ChudnovskyS08b,DBLP:journals/jct/ChudnovskyS08c,DBLP:journals/jct/ChudnovskyS08e,DBLP:journals/jct/ChudnovskyS10,DBLP:journals/jct/ChudnovskyS12a}, it might be possible to obtain a polynomial kernel for \pname{Claw-free Edge Editing}, but this is currently a major open question \cite{CaiC15incompressibility,crespelle2020survey,CyganPPLW17}. However, besides cliques and empty graphs, no $H$ with at least 5 vertices is known where \HEE\ has a polynomial kernel and there is no obvious candidate $H$ for which one would expect a kernel. This suggests the following conjecture:
\begin{conjecture}
\label{conj:editing}
If $H$ is a graph with at least 5 vertices, then \HEE\ has a polynomial kernel if and only if $H$ is a complete or empty graph.
\end{conjecture}
We are not able to resolve this conjecture, but make substantial progress towards it by showing that only a finite number of key cases needs to be 
understood. Our main result for \HEE\ is the following.

\begin{theorem}
\label{thm:main-editing}
There exists a set $\mathcal{H}^E$ of \nog\ graphs, each with five vertices, such that if \HEE\ is incompressible for every $H\in \mathcal{H}^E$, then for a graph $H$ with at least five vertices \HEE\ is incompressible if and only if $H$ is neither complete nor empty, where the incompressibility assumes \NOPH. 
\end{theorem}
The set $\mathcal{H}^E$ of \nog\ graphs are shown in Figure~\ref{table:tough-gang}. Note that a simple reduction by complementation shows that \HEE\ and \HBEE\ have the same complexity. Therefore, for each of these \nog\ graphs, we could put either it or its complement into the set $\mathcal{H}^E$. As it will be apparent later, we made significant efforts to reduce the size of $\mathcal{H}^E$ as much as possible. However, the known techniques for proving incompressibility do not seem to work for these graphs. Let us observe that most of these graphs are very close to the known cases that admit a polynomial kernel: for example, they can be seen as a path, paw, or diamond with an extra isolated vertex or with an extra degree-1 vertex attached. Thus resolving the kernelization complexity of \HEE\ for any of these remaining graphs seems to be a particularly good research question: either one needs to extend in a nontrivial way the known kernelization results, or significant new ideas are needed for proving hardness.

The reader might not be convinced of the validity of Conjecture~\ref{conj:editing} and may wonder about the value of Theorem~\ref{thm:main-editing} when the conjecture is false. However, we can argue that Theorem~\ref{thm:main-editing} is meaningful even in this case. It shows that if there is any $H$ violating Conjecture~\ref{conj:editing}, then one of the 9 graphs in $\mathcal{H}^E$ also violates it. That is, if we believe that there are kernelization results violating the conjecture, then we should focus on the 9 graphs in $\mathcal{H}^E$, as these are the easiest cases where we may have a kernelization result. In other words, Theorem~\ref{thm:main-editing} precisely shows the frontier where new algorithmic results are most likely to exist.

\input{tables/gang-table}

\HED\ is the variant of \HEE\ where only edge removal is allowed. For the same fixed graph $H$, it seems that \HED\ should be a simpler problem than \HEE, but we want to emphasize that \HED\ is \textit{not} a special case of \HEE. There is no known general reduction from the former to the latter, although the technique of completion enforcers (see Section~\ref{sec:cai} and \cite{CaiC15incompressibility}) can be used for many specific graphs $H$.\footnote{Interestingly, going beyond H-free graphs we have evidence that such reductions are impossible: \pname{Split Editing} (Split graphs are the class of $\{2K_2,C_4,C_5\}$-free graphs) is known to be polynomial-time solvable~\cite{DBLP:journals/combinatorica/HammerS81}, while \pname{Split Deletion} is NP-hard~\cite{DBLP:journals/dam/NatanzonSS01}.} 
There is a known case where \HED\ seems to be strictly easier: if $H$ has at most one edge, then there is only one way of destroying a copy of an induced $H$ by edge removal, making the problem polynomial-time solvable. Aravind et al.~\cite{AravindSS17} showed that having at most one edge is the only condition that makes \HED\ polynomial-time solvable: if $H$ has at least two edges, then the problem is NP-hard. Therefore, the counterpart of Conjecture~\ref{conj:editing} for \HED\ should take this case also into account.

\begin{conjecture}
\label{conj:deletion}
If $H$ is a graph with at least 5 vertices, then \HED\ has a polynomial kernel if and only if $H$ is a complete graph or has at most one edge.
\end{conjecture}
Working toward this conjecture, we show that only a finite number of cases needs to be shown incompressible. 
\begin{theorem}
\label{thm:main-deletion}
There exists a set $\mathcal{H}^D$ of \nogd\ graphs, each with either five or six vertices such that if \HED\ is incompressible for every $H\in \mathcal{H}^D$ then for a graph $H$ with at least five vertices, \HED\ is incompressible if and only if $H$ is a graph with at least two edges but not complete, where the incompressibility assumes \NOPH. 
\end{theorem}
The set $\mathcal{H}^D$ contains the graphs in set $\mathcal{H}^E$, as well as their complements. This seems reasonable and hard to avoid: if we do not have an incompressibility result for \HEE\ for some $H\in \mathcal{H}^E$, then it is unlikely that we can find such a result for \HED\ (even though, as discussed above, there is no formal justification for this). Together with these 17 graphs (note that $H_9$ is the same as its complement), we need to include into $\mathcal{H}^D$ the two graphs $D_1$ and $D_2$ shown in Figure~\ref{table:deletion-gang}. 
In the case of editing, we can prove incompressibility for these two graphs by a reduction from \HEE\ where $H$ is the graph with 5 vertices and one edge. However, \HED\ for this $H$ is polynomial-time solvable.

Finally, let us consider the \HEC\ problem, where we have to make $G$ induced $H$-free by adding at most $k$ edges. As \HEC\ is essentially the same problem as \HBED, we can obtain a counterpart of Theorem~\ref{thm:main-deletion} by simple complementation:
\begin{theorem}
\label{thm:main-completion}
    There exists a set $\mathcal{H}^C$ of \nogd\ graphs, each with either five or six vertices such that if \HEC\ is incompressible for every $H\in \mathcal{H}^C$ then for a graph $H$ with at least five vertices, \HEC\ is incompressible if and only if $H$ is a graph with at least two nonedges but not empty, where the incompressibility assumes \NOPH. 
\end{theorem}

\textbf{Our techniques.}
We crucially use two earlier results. First, Cai and Cai \cite{CaiC15incompressibility} proved that \HEE\ is incompressible (assuming \NOPH) when $H$ or~$\overline{H}$ is a cycle or a path of length at least 4, or 3-connected but not complete. While these result handle many graphs and prove to be very useful for our proofs, they do not come close to a complete classification. Second, we use a key tool in the polynomial-time dichotomy result of Aravind et al.~\cite{AravindSS17}: if $V_\ell$ is the set of lowest degree vertices of $H$, then \pname{$(H-V_\ell)$-free Edge Editing} can be reduced to \HEE. The same statement holds for the set $V_h$ of highest degree vertices.

Our proofs of Theorems~\ref{thm:main-editing}--\ref{thm:main-completion} introduce new incompressibility results and new reductions, which we put together to obtain an almost complete classification by a graph-theoretic analysis.  Additionally, to make the arguments simpler, we handle small graphs by an exhaustive computer search. In the following, we highlight some of the main ideas that appear in the paper.

\begin{itemize}
\item \textbf{Analysis of graphs.} Our goal is to prove Theorem~\ref{thm:main-editing} by induction on the size of~$H$. 
First we handle the case when $H$ is regular: 
we show that this typically implies that either $H$ or $\overline{H}$ is 3-connected, 
and the result of Cai and Cai \cite{CaiC15incompressibility} can be used. 
If $H$ is not regular, then the graphs $H-V_\ell$ and $H-V_h$ are nonempty and have stricly fewer vertices than $H$. 
If one of them, say $H-V_\ell$, has at least 5 vertices and is neither complete nor empty, 
then the induction hypothesis gives an incompressibility result for \pname{$(H-V_\ell)$-free Edge Editing}, 
which gives an incompressibility result for \HEE\ by the reduction of Aravind et al.~\cite{AravindSS17}. 
Therefore, we only need to handle those graphs $H$ where it is true for both $H-V_\ell$ and $H-V_h$ that they are either 
small, complete, or empty. But we can obtain a good structural understanding of $H$ in each of these cases, 
which allows us to show that either $H$ or~$\overline{H}$ is 3-connected, or $H$ has some very well 
defined structure. With these arguments, we can reduce the problem to the incompressibility of 
\HEE\ for a few dozen specific graphs $H$ and for a few well-structured infinite families (such as $K_{2,t}$).

For \HED, we have the additional complication that one or both of $H-V_\ell$ and $H-V_h$ can be near-empty (i.e., has exactly one edge), which is not an incompressible case for this problem. We need additional case analysis to cover such graphs, but the spirit of the proof remains the same.

\item\textbf{Computer search.} Our analysis of graphs becomes considerably simpler if we assume that $H$ is not too small. In this case, we can assume that at least one of $H-V_\ell$ and $H-V_h$ is a complete or empty graph of certain minimum size, which is a very helpful starting point for proving the 3-connectivity of $H$ or $\bar H$, respectively. Therefore, we handle every graph with at most 11 vertices using an exhaustive computer search and assume in the proof that $H$ has at least 12 vertices. The list provided by McKay~\cite{mckayGraphs} shows that there are 1031291298 different graphs with at most 11 vertices. Verifying all these graphs with a computer search was a nontrivial, but doable task. 
We ran the search parallelly on 45 threads of a computing server (2.5 GHz CPUs). It took around 24 hours to complete the search.
We remark that there is trade off between the amount of computer search done and the complexity of the proof. For example, in earlier versions of the paper, the exhaustive search was performed up to a smaller number of vertices, which meant that several small graphs needed to be treated as special cases in the proofs. In principle, it would be possible to extend our case analysis to completely avoid computer search, but it would significantly complicate the proof and is not clear what additional insight it would give.

\item\textbf{Reductions.}  We investigate different reductions that allow us to reduce \pname{$H'$-free Edge Editing} to \HEE\ when $H'$ is an induced subgraph of $H$ satisfying certain conditions. With extensive use of such reductions, we can reduce the remaining cases of \HEE\ that needs to be handled to a smaller finite set.

\item\textbf{Incompressibility results.} We carefully revisit the proof of Cai and Cai \cite{CaiC15incompressibility} showing the incompressibility of \HEE\ when $H$ is 3-connected, and observe that, with additional ideas, it can be made to work also for certain 2-connected graphs that are not 3-connected (the set $\mathcal{A}$ of graphs shown in Figure~\ref{table:soft-gang} and the set $\mathcal{B}$ of graphs shown in Figure~\ref{table:deletion-gang}). This allows us to handle every graph, except those finite sets that are mentioned in Theorems~\ref{thm:main-editing}--\ref{thm:main-completion}. A key step in many of these incompressibility results is to establish first incompressibility for the 
\RHED\ problem, which is the generalization of \HED\ where some of the edges of $G$ are marked as forbidden in the input, and the solution is not allowed to delete forbidden edges. Then we use deletion and completion enforcer gadgets specific to $H$ to reduce \RHED\ to \HEE.
\end{itemize}

The paper is organized as follows. Preliminaries are in Section~\ref{sec:prelim}. Section~\ref{sec:larger-gang} presents the churning procedure, our main technical tool in the analysis of graphs, and shows that it reduces the problem to a finite number of graphs, plus a few well-defined infinite families. Section~\ref{sec:smaller-gang} presents reductions (old and new) that allow us to further reduce the number of graphs we need to handle. Finally, in Setion~\ref{sec:cai}, we give new incompressibility results, showing that only the cases stated in Theorems~\ref{thm:main-editing}--\ref{thm:main-completion} need to be proved incompressible to complete the exploration of the complexity landscape of the problems.



%% file: tables/gang-table.tex
\begin{figure}
\begin{center}
\scalebox{0.8}{
\begin{tabular}{| c | c | c | c | c | c | c | c | c | c | c |}
\cline{1-3} \cline{5-7} \cline{9-11}
\# & \textbf{$H$} & \fbox{\textbf{$\overline{H}$}} & & \# & \textbf{$H$} & \textbf{$\overline{H}$} & & \# & \textbf{$H$} & \textbf{$\overline{H}$}\\ \cline{1-3} \cline{5-7} \cline{9-11}
1 & \input{figs/gang/521} & \input{figs/gang/c521} & & 
  4 & \input{figs/gang/534} & \input{figs/gang/c534} & & 
    7 & \input{figs/gang/543} & \input{figs/gang/c543} \\ \cline{1-3} \cline{5-7} \cline{9-11}
2 & \input{figs/gang/531} & \input{figs/gang/c531} & &
  5 & \input{figs/gang/541} & \input{figs/gang/c541} & & 
    8 & \input{figs/gang/551} & \input{figs/gang/c551} \\ \cline{1-3} \cline{5-7} \cline{9-11}
3 & \input{figs/gang/533} & \input{figs/gang/c533} & & 
  6 & \input{figs/gang/542} & \input{figs/gang/c542} & & 
    9 & \input{figs/gang/552} & same\\ \cline{1-3} \cline{5-7} \cline{9-11}
\end{tabular}}
\end{center}
\caption{The set $\mathcal{H}$ of graphs}
\label{table:tough-gang}
\end{figure}

\begin{figure}
\begin{center}
\scalebox{0.8}{
\begin{tabular}{| c | c | c | c | c | c | c | c | c | c | c |}
\cline{1-3} \cline{5-7} \cline{9-11}
\# & \textbf{$A$} & \fbox{\textbf{$\overline{A}$}} & & \# & \textbf{$A$} & \textbf{$\overline{A}$} & & \# & \textbf{$A$} & \textbf{$\overline{A}$}\\ \cline{1-3} \cline{5-7} \cline{9-11}
1 & \input{figs/gang/532} & \input{figs/gang/c532} & & 
  4 & \input{figs/gang/8141} & same & & 
    7 & \input{figs/gang/663} & \input{figs/gang/c663} \\ \cline{1-3} \cline{5-7} \cline{9-11}
2 & \input{figs/gang/661} & \input{figs/gang/c661} & &
  5 & \input{figs/gang/9181} & same & & 
    8 & \input{figs/gang/672} & \input{figs/gang/c672} \\ \cline{1-3} \cline{5-7} \cline{9-11}
3 & \input{figs/gang/673} & \input{figs/gang/c673} & & 
  6 & \input{figs/gang/662} & \input{figs/gang/c662} & & 
    9 & \input{figs/gang/791} & \input{figs/gang/c791}\\ \cline{1-3} \cline{5-7} \cline{9-11}
\end{tabular}}
\end{center}
\caption{The set $\mathcal{A}$ of graphs}
\label{table:soft-gang}
\end{figure}

\begin{figure}
\begin{center}
\scalebox{0.8}{
\begin{tabular}{| c | c | c | c | c | c | c | c | c | c | c |}
\cline{1-3} \cline{5-7} \cline{9-11}
\# & \textbf{$D$} & \fbox{\textbf{$\overline{D}$}} & & \# & \textbf{$B$} & \textbf{$\overline{B}$} & & \# & \textbf{$B$} & \textbf{$\overline{B}$}\\ \cline{1-3} \cline{5-7} \cline{9-11}
1 & \input{figs/gang/s31} & \input{figs/gang/s31c} & & 
  1 & \input{figs/gang/clawuk2} & \input{figs/gang/clawuk2c} & & 
    3 & \input{figs/gang/z8135} & \input{figs/gang/z8135c} \\ \cline{1-3} \cline{5-7} \cline{9-11}
2 & \input{figs/gang/664} & \input{figs/gang/c664} & &
  2 & \input{figs/gang/782} & \input{figs/gang/c782} & & 
     &  &  \\ \cline{1-3} \cline{5-7} \cline{9-11}
\end{tabular}}
\end{center}
\caption{The sets $\mathcal{D}$ and $\mathcal{B}$ of graphs}
\label{table:deletion-gang}
\end{figure}

%% file: preliminaries.tex
\section{Preliminaries}
\label{sec:prelim}
\textbf{Graph-theoretic notation and terminology.}
For a graph $G$, $V(G)$ and $E(G)$ denote the set of vertices and the set of edges of $G$ respectively. 
For a set $V'\subseteq V(G)$, $G-V'$ denotes the graph obtained by removing all vertices in $V'$ and their incident
edges from $G$. For a set $F$ of pairs of vertices and a graph $G$, $G\triangle F$ denotes the graph $G'$ such that
$V(G') = V(G)$ and $E(G') = \{(u, v)\, |\, ((u,v)\in E(G)\, \text{and}\, (u,v) \notin F)\, \text{or}\, (u,v \in V(G),\, (u,v)\notin E(G),\, \text{and}\, (u,v) \in F) \}$. 
Similarly, $G-F$ denotes the graph $G'$ such that $V(G') = V(G)$ and $E(G') = E(G)\setminus F$, and $G+F$ denotes the graph $G'$ such that
$V(G') = V(G)$ and $E(G') = E(G)\cup \{(u,v)\, |\, u,v \in V(G)\, \text{and}\, (u,v)\in F\}$.
Whenever we say that a set of (non)edges $F$ is a solution of an instance $(G,k)$
of a problem, we refer to a subset of $F$ containing all (non)edges where both the end vertices are in $V(G)$. 
A graph is \textit{empty} if it does not have any edges. A graph is \textit{near-empty} if it has exactly one edge. A graph is \textit{complete} if it has no nonedges. 
A component of a graph is a \textit{largest component} if it has maximum number of vertices among all components of the graph.
Similarly, a component of a graph is a \textit{smallest component} if it has minimum number of vertices among all components of the graph.
For a graph $H$ which is not complete, the \textit{vertex connectivity} of $H$ is the minimum integer $c$ such that there exists a set $S\subseteq V(H)$ such that $|S|=c$ and $H-S$ is disconnected.
For a complete graph on $n$ vertices, the vertex connectivity is defined to be $n-1$.
For a graph $H$ with vertex connectivity 1, a vertex $v$ in $H$ is known as a \textit{cut vertex} if $H-v$ is disconnected.
A graph is \textit{$k$-connected}, if its vertex connectivity is at least $k$.
An induced subgraph $H'$ of~$H$ is known as a \textit{2-connected component} if $H'$ is a maximal 2-connected induced subgraph of~$H$.
The adjectives `largest' and `smallest' can be applied to 2-connected components as done for components. 
A \textit{twin-star} graph $T_{\ell_1,\ell_2}$ for $\ell_1,\ell_2\geq 0$ is defined as the tree with two adjacent vertices $u$ and $v$ such that $|N(u)\setminus \{v\}|=\ell_1, |N(v)\setminus \{u\}|=\ell_2$, and 
every vertex in $N(u)\cup N(v)\setminus \{u,v\}$ has degree 1.
A graph $G$ is \textit{$H$-free} if $G$ does not contain any induced subgraph isomorphic to $H$. 
For two graphs $G_1$ and $G_2$, where $V(G_1)$ and $V(G_2)$ are disjoint, the \textit{disjoint union} of $G_1$ and $G_2$ denoted by $G_1\cup G_2$ (or $G_2\cup G_1$) is the graph $G$
such that $V(G) = V(G_1)\cup V(G_2)$ and $E(G) = E(G_1)\cup E(G_2)$.
For two graphs $G_1$ and $G_2$, the \textit{join} of $G_1$ and $G_2$ denoted by $G_1\vee G_2$ (or $G_2\vee G_1$), is the graph $G$ such that
$V(G) = V(G_1)\cup V(G_2)$ and $E(G) = E(G_1)\cup E(G_2)\cup \{(x,y)\, |\, x\in V(G_1),\, y\in V(G_2) \}$.
A complete graph, a cycle, and a path with $t$ vertices are denoted by $K_t, C_t,$ and $P_t$ respectively.
By $K_t-e$, we denote the graph obtained by deleting an edge from a complete graph on $t$ vertices. 
A graph is $r$-regular if the degree of each vertex is $r$. A graph is regular if it is $r$-regular for some integer $r\geq 0$.
We call a graph \textit{non-regular} if it is not regular.
A \textit{modular decomposition} $\mathcal{M}$ of a graph $G$ is a partitioning of its vertices into maximal sets, known as \textit{modules}, such that for every set $M\in \mathcal{M}$,
every vertex in $M$ has the same neighborhood outside $M$. 
Let $\mathcal{M'}\subseteq \mathcal{M}$. Let $V'=\bigcup_{M\in\mathcal{M'}}M$. Then we say that $\mathcal{M'}$ corresponds to $V'$. For a set $\mathcal{S}$ of graphs, by $\overline{\mathcal{S}}$ we denote the set of complements of graphs in $\mathcal{S}$. 
Figure~\ref{fig:important-graphs}
shows all graphs with at most four vertices which are neither empty nor complete. 

For $t\geq 3$, let $J_t$ be the graph obtained from $K_2\vee tK_1$ and $C_4$ by identifying an edge of $C_4$ with the edge between the highest degree vertices in $K_2\vee tK_1$. Let $Q_t$ be the graph graph obtained from $K_{2,t}$, for some $t\geq 3$, by adding a path of length three between the highest degree vertices in $K_{2,t}$.
Let $\mathcal{H}$, $\mathcal{A}$, $\mathcal{D}$, $\mathcal{B}$, $\mathcal{S}$ denote the graphs ($H$, $A$, $D$, $B$, $S$ respectively) shown in Figures~\ref{table:tough-gang}, \ref{table:soft-gang}, \ref{table:deletion-gang}, and \ref{table:shortlist-smallbig}. Let $\mathcal{F}$ be the union of graphs in the classes of graphs shown in column $\mathcal{F}$ of Figure~\ref{table:shortlist-class}. 
The graphs in $\mathcal{S}$ and $\mathcal{F}$ are handled in Section~\ref{sec:smaller-gang} and the graphs in $\mathcal{A}$ and $\mathcal{B}$ are handled in Section~\ref{sec:cai}.
For all these classes of graphs, we use subscripts to identify each graph/graph class. For example $H_1$ is $P_3\cup 2K_1$ and $\mathcal{F}_2$ is the class of graphs $K_{1,t}$. Let $\mathcal{W}$ be the set  $\mathcal{H}\cup \overline{\mathcal{H}}\cup \mathcal{A}\cup \overline{\mathcal{A}}\cup \mathcal{D}\cup\overline{\mathcal{D}}\cup \mathcal{B}\cup\overline{\mathcal{B}}\cup \mathcal{S}\cup\overline{\mathcal{S}}\cup \mathcal{F}\cup\overline{\mathcal{F}}$. We observe that $\overline{\LG} = \LG$.

\input{important-graphs}
\input{tables/shortlist}

\textbf{Parameterized problems and transformations.}
Here, we very briefly recall the definitions related to parameterized algorithms and complexity that are required in this paper. 
We refer to the book~\cite{book:pa} for a detailed exposition of the field.
A parameterized problem is a classical problem with an additional integer input known as the parameter. 
A parameterized problem admits a polynomial kernel if there is a polynomial-time algorithm which takes as input an instance $(I,k)$ of the problem and outputs an
instance $(I',k')$ of the same problem, where $|I'|, k'\leq p(k)$, where $p(k)$ is a polynomial in $k$, such that $(I,k)$ is a yes-instance if and only if $(I',k')$ is a yes-instance.
A parameterized problem is incompressible if it does not admit a polynomial kernel. 
A \textit{Polynomial Parameter Transformation} (PPT) from one parameterized problem $Q$ to another parameterized problem $Q'$ is a polynomial-time algorithm which takes as input an instance $(I,k)$ of $Q$ and produces an instance $(I',k')$ of 
$Q'$ such that $(I,k)$ is a yes-instance of $Q$ if and only if $(I',k')$ is a yes-instance of $Q'$, and $k'\leq p(k)$, for some polynomial $p(.)$. 
It is known that if there is a PPT from $Q$ to $Q'$, then if $Q$ is incompressible, then so is $Q'$.

The parameterized problems we deal with in this paper are listed below.

\begin{mdframed}
  \textbf{\HEE}: Given a graph $G$ and an integer $k$, does there exist a set $F$ of at most $k$ edges/nonedges
  such that $G\triangle F$ is $H$-free?\\
  \textbf{Parameter}: $k$
\end{mdframed}

\begin{mdframed}
  \textbf{\HED}: Given a graph $G$ and an integer $k$, does there exist a set $F$ of at most $k$ edges
  such that $G-F$ is $H$-free?\\
  \textbf{Parameter}: $k$
\end{mdframed}

\begin{mdframed}
  \textbf{\HEC}: Given a graph $G$ and an integer $k$, does there exist a set $F$ of at most $k$ nonedges such that
  $G+F$ is $H$-free?\\
  \textbf{Parameter}: $k$
\end{mdframed}

\textbf{Basic results.}
Proposition~\ref{pro:folklore} follows from the observations that $(G,k)$ is a yes-instance of \HEE(\textsc{Deletion}) if and only if $(\overline{G}, k)$ is a yes-instance of \HBEE(\textsc{Completion}). It enables us to focus only on \HEE\ and \HED. 

\begin{proposition}[folklore]
  \label{pro:folklore}
  Let $H$ be any graph. Then \HED\ is incompressible if and only if \HBEC\ is incompressible. Similarly,
  \HEE\ is incompressible if and only if \HBEE\ is incompressible.
\end{proposition}

For graphs $H$ and $H'$, by ``\textit{$H$ \simulates\ $H'$}'' and by ``\textit{$H'$ \issimulatedby\ $H$}'', we mean that, there is a PPT from \HDEE\ to \HEE, there is a PPT from \HDED\ to \HED, and there is a PPT from \HDEC\ to \HEC.
We observe that this is transitive, i.e., if $H$ \simulates\ $H'$ and $H'$ \simulates\ $H''$, then $H$ \simulates\ $H''$. 
A set of graphs $\mathcal{H}$ is called a \textit{\baseset} for a set $\mathcal{G}$ of graphs if for every graph $H\in \mathcal{G}$ there is a graph $H'\in \mathcal{H}$ such that $H$ \simulates\  $H'$.
The objective of the rest of the paper is to find, for each of the problems, a \baseset\ $\mathcal{H}\cup \mathcal{X}$ for all graphs  with at least five
vertices, except the trivial cases, such that the following conditions are satisfied: (i) $\mathcal{H}$ is finite and the incompressibility is not known for any graph in it; (ii) for every graph in $\mathcal{X}$, the problem is known to be incompressible. 

Proposition~\ref{pro:folklore} implies Corollary~\ref{cor:folklore} and Proposition~\ref{pro:birdeye} can be deduced directly from the definitions.
\begin{corollary}
  \label{cor:folklore}
  Let $H$ and $H'$ be graphs such that $H$ \simulates\ $H'$. Then $\overline{H}$ \simulates\ $\overline{H'}$.
\end{corollary}

\begin{proposition}
  \label{pro:birdeye}
  Let $\mathcal{H}$ be a \baseset\ for a set $\mathcal{G}$ of graphs. Assume that for every graph $H'\in \mathcal{H}$, \HDEE\ (\textsc{Deletion}) is incompressible. Then for every graph $H\in \mathcal{G}$, \HEE\ (\textsc{Deletion}) is incompressible.
\end{proposition}

Intuitively, if $H'$ is an induced subgraph of $H$, then \HEE\ (\textsc{Deletion}) seems harder than \HDEE\ (\textsc{Deletion}). However, there is no general argument why this should be true: there does not seem to be a completely general reduction that would reduce \HDEE\ (\textsc{Deletion}) to \HEE\ (\textsc{Deletion}). There is, however, a fairly natural idea for trying to do such a reduction: we extend the graph by attaching copies of $H-H'$ at every place where a copy of $H'$ can potentially appear.
The following construction is essentially the same as the main construction used in \cite{AravindSS17}. 

\begin{construction}[see \cite{AravindSS17}]
  \label{con:main}
  Let $(G',k, H, V')$ be an input to the construction, where $G'$ and $H$ are graphs, $k$
  is a positive integer and $V'$ is a subset of vertices of $H$. We construct a graph $G$
  from $G'$ as follows.
  For every injective function $f: V'\longrightarrow V(G')$, do the following:
  \begin{itemize}
    \item Introduce $k+1$ sets of vertices $V_1, V_2,\ldots, V_{k+1}$, each of size $|V(H)\setminus V'|$\,, and $k+1$ bijective functions
      $g_i: V(H)\longrightarrow (f(V')\cup V_i)$, for $1\leq i\leq k+1$, such that $g_i(v') = f(v')$ for every $v'\in V'$;
    \item For each set $V_i$, introduce an edge set $E_i = \{(u,v)\, |\, u\in (f(V')\cup V_i),\, v\in V_i,\, (g_i^{-1}(u), g_i^{-1}(v))\in E(H)\}$. 
  \end{itemize}
  This completes the construction. Let the constructed graph be $G$. 
\end{construction}

For convenience, we call every set $V_i$ of vertices introduced in the construction a \textit{satellite} and the vertices in it \textit{satellite vertices}. This reduction works correctly in one direction: it ensures that the operations that make the new graph $G$ $H$-free should ensure that the copy of $G'$ inside $G$ is $H'$-free. 
\begin{proposition}[see Lemma 2.6 in~\cite{AravindSS17}]
  \label{pro:con:main-backward}
  Let $G$ be obtained by Construction~\ref{con:main} on 
  the input $(G',k,H,V')$, where $G'$ and $H$ are graphs, $k$ is a positive integer and $V'\subseteq V(H)$.
  Then, if $(G,k)$ is a yes-instance of \HEE\ (\textsc{Deletion}), 
  then $(G',k)$ is a yes-instance of \HDEE\ (\textsc{Deletion}), 
  where $H'$ is $H[V']$.
\end{proposition}
However, the other direction of the correctness of the reduction does not hold in general (this is easy to see for example for $H=K_{1,2}$ and $H'=K_2$). As we shall see, there are particular cases where we can prove the converse of Proposition~\ref{pro:con:main-backward}, for example, when $H-H'$ consists of exactly the highest- or lowest-degree vertices.
Application of such arguments will be our main tool in reducing the complexity of \HEE\ (\textsc{Deletion}) to simpler cases. 
The first known incompressible (assuming \NOPH) $H$-free edge modification problems are \HEE\ and \HED\ when $H$ is $K_1\boxtimes (2K_1\cup 2K_2)$~\cite{DBLP:journals/disopt/KratschW13}.
It is known that when $H$ is a star graph on at least 11 vertices, \HED\ is incompressible under the same complexity assumption~\cite{cai2012polynomial}.
Propositions~\ref{pro:base:incompress:all} to \ref{pro:small-kernel} summarize other major results on the incompressibility of $H$-free edge modification problems known so far.
\begin{proposition}[\cite{CaiC15incompressibility}]
  \label{pro:base:incompress:all}
  Assuming \NOPH, \HEE, \HED, and \HEC\ are incompressible if $H$ is either of the following graphs.
  \begin{enumerate}[(i)]
  \item\label{item:incompress:c} $C_\ell$ for any $\ell \geq 4$;
  \item\label{item:incompress:p} $P_\ell$ for any $\ell \geq 5$;
  \end{enumerate}
\end{proposition}



\begin{proposition}[\cite{CaiC15incompressibility}]
  \label{pro:base:incompress:3conn}
  Assuming \NOPH, for 3-connected graphs $H$, \HEE\ and \HED\ are incompressible if $H$ is not complete and \HEC\ is incompressible
  if $H$ has at least two nonedges. 
\end{proposition}  

\begin{proposition}[\cite{CaiC15incompressibility}, folklore]
  \label{pro:trivial-kernel}
  If $H$ is a complete or empty graph, then \HEE\ admits polynomial kernelization. 
  If $H$ is complete or has at most one edge then \HED\ admits polynomial kernelization.
  If $H$ is an empty graph or has at most one nonedge then \HEC\ admits polynomial kernelization.
\end{proposition}

\begin{proposition}
  \label{pro:small-kernel}
  \HEE, \HED, and \HEC\ admit polynomial kernels when $H$ is a $P_3$~\cite{CaoC12, GrammGHN03}, $P_4$~\cite{GuillemotHPP13}, paw~\cite{CaoKY20, EibenLS15}, or a diamond~\cite{cai2012polynomial, CRSY18polynomial}.
\end{proposition}

We end this section by proving the incompressibility of the problems for regular nontrivial graphs $H$.
\input{regular}



%% file: important-graphs.tex
\begin{figure}
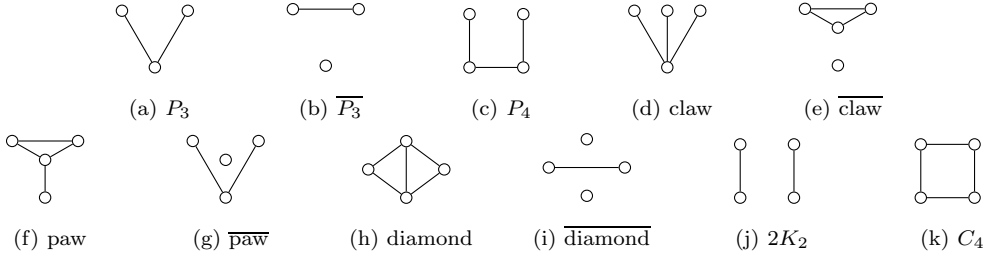

  \centering
  \begin{subfigure}[b]{0.15\textwidth}
    \centering
    \input{figs/tiny/p3}
    \caption{$P_3$}
    \label{fig:p3}
  \end{subfigure}%
  \begin{subfigure}[b]{0.15\textwidth}
    \centering
    \input{figs/tiny/p3c}
    \caption{$\overline{P_3}$}
    \label{fig:p3c}
  \end{subfigure}%
  \begin{subfigure}[b]{0.15\textwidth}
    \centering
    \input{figs/tiny/p4}
    \caption{$P_4$}
    \label{fig:p4}
  \end{subfigure}%
  \begin{subfigure}[b]{0.15\textwidth}
    \centering
    \input{figs/tiny/claw}
    \caption{claw}
    \label{fig:claw}
  \end{subfigure}%
  \begin{subfigure}[b]{0.15\textwidth}
    \centering
    \input{figs/tiny/clawc}
    \caption{$\overline{\text{claw}}$}
    \label{fig:clawc}
  \end{subfigure}%
  
  \begin{subfigure}[b]{0.15\textwidth}
    \centering
    \input{figs/tiny/paw}
    \caption{paw}
    \label{fig:paw}
  \end{subfigure}
  \begin{subfigure}[b]{0.15\textwidth}
    \centering
    \input{figs/tiny/pawc}
    \caption{$\overline{\text{paw}}$}
    \label{fig:pawc}
  \end{subfigure}
  \begin{subfigure}[b]{0.15\textwidth}
    \centering
    \input{figs/tiny/diamond}
    \caption{diamond}
    \label{fig:diamond}
  \end{subfigure}
  \begin{subfigure}[b]{0.15\textwidth}
    \centering
    \input{figs/tiny/diamondc}
    \caption{$\overline{\text{diamond}}$}
    \label{fig:diamondc}
  \end{subfigure}
  \begin{subfigure}[b]{0.15\textwidth}
    \centering
    \input{figs/tiny/2k2}
    \caption{$2K_2$}
    \label{fig:2k2}
  \end{subfigure}
  \begin{subfigure}[b]{0.15\textwidth}
    \centering
    \input{figs/tiny/c4}
    \caption{$C_4$}
    \label{fig:c4}
  \end{subfigure}
  \caption{All non-empty and non-complete graphs with at most four vertices}
  \label{fig:important-graphs}
\end{figure}

%% file: tables/shortlist.tex
\begin{figure}
\begin{center}
\scalebox{0.7}{
\begin{tabular}{| c | c | c | c | c | c | c | c | c | c | c | c | c | c | c |}
\cline{1-3} \cline{5-7} \cline{9-11} \cline{13-15}
\# & \textbf{$S$} & \fbox{\textbf{$\overline{S}$}} & & \# & \textbf{$S$} & \textbf{$\overline{S}$} & & \# & \textbf{$S$} & \textbf{$\overline{S}$} & & \# & \textbf{$S$} & \textbf{$\overline{S}$}\\ \cline{1-3} \cline{5-7} \cline{9-11} \cline{13-15}
1 & \input{figs/small/k23} & \input{figs/small/k23c} & &
  10 & \input{figs/small/z7101} & \input{figs/small/z7101c} & &
    19 & \input{figs/small/z891} & \input{figs/small/z891c} & &
      28 & \input{figs/large/z10181} & \input{figs/large/z10181c} \\
            \cline{1-3} \cline{5-7} \cline{9-11} \cline{13-15}
2 & \input{figs/small/c651} & \input{figs/small/651} & &
  11 & \input{figs/small/z781} & \input{figs/small/z781c} & &
    20 & \input{figs/small/z8171c} & \input{figs/small/z8171} & &
      29 & \input{figs/large/z11221} & \input{figs/large/z11221c} \\
            \cline{1-3} \cline{5-7} \cline{9-11} \cline{13-15}
3 & \input{figs/small/c671} & \input{figs/small/671} & &
  12 & \input{figs/small/z8141} & same & &
    21 & \input{figs/small/z8134} & \input{figs/small/z8134c} & &
      30 & \input{figs/large/z10231} & \input{figs/large/z10231c} \\
            \cline{1-3} \cline{5-7} \cline{9-11} \cline{13-15}
4 & \input{figs/small/claw-2k1} & \input{figs/small/claw-2k1c} & &
  13 & \input{figs/small/z8101} & \input{figs/small/z8101c} & &
      22 & \input{figs/small/z8121} & \input{figs/small/z8121c} & &
        31 & \input{figs/large/z10232} & \input{figs/large/z10232c} \\
              \cline{1-3} \cline{5-7} \cline{9-11} \cline{13-15}
5 & \input{figs/small/z671} & \input{figs/small/z671c} & &
  14 & \input{figs/small/z8151c} & \input{figs/small/z8151} & &
      23 & \input{figs/small/z9211c} & \input{figs/small/z9211} & &
        32 & \input{figs/large/z11341} & \input{figs/large/z11341c} \\
              \cline{1-3} \cline{5-7} \cline{9-11} \cline{13-15}
6 & \input{figs/small/paw-2k1} & \input{figs/small/paw-2k1c} & &
  15 & \input{figs/small/z8132} & \input{figs/small/z8132c} & &
    24 & \input{figs/small/z9181} & same & &
      33 & \input{figs/large/z10221c} & \input{figs/large/z10221} \\
            \cline{1-3} \cline{5-7} \cline{9-11} \cline{13-15}
7 & \input{figs/small/z771} & \input{figs/small/z771c} & &
  16 & \input{figs/small/z8133} & \input{figs/small/z8133c} & &
    25 & \input{figs/small/z9161} & \input{figs/small/z9161c} & &
      34 & \input{figs/large/z11281c} & \input{figs/large/z11281} \\
            \cline{1-3} \cline{5-7} \cline{9-11} \cline{13-15}
8 & \input{figs/small/z7111} & \input{figs/small/z7111c} & &
  17 & \input{figs/small/z8122c} & \input{figs/small/z8122} & &
    26 & \input{figs/small/z9151} & \input{figs/small/z9151c} & &
      35 & \input{figs/large/z12281} & \input{figs/large/z12281c} \\
            \cline{1-3} \cline{5-7} \cline{9-11} \cline{13-15}
9 & \input{figs/small/z7102} & \input{figs/small/z7102c} & &
  18 & \input{figs/small/z8131} & \input{figs/small/z8131c} & &
    27 & \input{figs/small/z9141} & \input{figs/small/z9141c} & &
       36 & \input{figs/large/z10261} & \input{figs/large/z10261c}\\
      \cline{1-3} \cline{5-7} \cline{9-11} \cline{13-15}
\end{tabular}}
\end{center}
\caption{The set $\mathcal{S}$ of graphs}
\label{table:shortlist-smallbig}
\end{figure}

\begin{figure}
\begin{center}
\scalebox{0.515}{
\begin{tabular}{| c | c | c | c | c | c | c | c | c | c | c | c | c |}
\cline{1-6} \cline{8-13} 
\# & \textbf{$\mathcal{F}$} & \fbox{\textbf{$\overline{\mathcal{F}}$}} & Constraint & $F$ when $t=6$ & $\overline{F}$ when $t=6$ & & \# & \textbf{$\mathcal{F}$} & \fbox{\textbf{$\overline{\mathcal{F}}$}} & Constraint & $F$ when $t=6$ & $\overline{F}$ when $t=6$ \\ \cline{1-6} \cline{8-13} 
1 & $K_{2,t}$ & $K_t\cup K_2$ & $4\leq t$ & \input{figs/small/f8121} & \input{figs/small/f8121c} & &
  6 & $\overline{(K_t-e)\cup K_2}$ & $(K_t-e)\cup K_2$ & $4\leq t$ & \input{figs/small/f8132} & \input{figs/small/f8132c}\\
  \cline{1-6} \cline{8-13}
2 & $K_{1,t}$ & $K_t\cup K_1$ & $5\leq t$ & \input{figs/small/f761} & \input{figs/small/f761c} & &
  7 & $K_{1,t}\cup K_2$ & $\overline{K_{1,t}\cup K_2}$ & $4\leq t$ & \input{figs/small/f971} & \input{figs/small/f971c}\\
  \cline{1-6} \cline{8-13}
3 & $K_2\vee tK_1$ & $K_t\cup 2K_1$ & $4\leq t$ & \input{figs/small/f8131} & \input{figs/small/f8131c} & & 
  8 & $\overline{(K_t-e)\cup K_1}$ & $(K_t-e)\cup K_1$ & $6\leq t$ & \input{figs/small/f771} & \input{figs/small/f771c}\\
  \cline{1-6} \cline{8-13}
4 & $T_{t,1}$ & $\overline{T_{t,1}}$ & $4\leq t$ & \input{figs/small/f981} & \input{figs/small/f981c} & &
  9 & $J_t$ & $\overline{J_t}$ & $3\leq t$ & \input{figs/small/f10161} &\input{figs/small/f10161c} \\
  \cline{1-6} \cline{8-13}
5 & $\overline{(K_t-e)\cup 2K_1}$ & $(K_t-e)\cup 2K_1$ & $4\leq t$ & \input{figs/small/f8141} & \input{figs/small/f8141c} & &
  10 & $Q_t$ & $\overline{Q_t}$ & $3\leq t$ & \input{figs/small/f10151} & \input{figs/small/f10151c}\\
  \cline{1-6} \cline{8-13}
\end{tabular}}
\end{center}
\caption{The set $\mathcal{F}$ of infinite sets of graphs}
\label{table:shortlist-class}
\label{table:shortlist}
\end{figure}

%% file: regular.tex
%




 
\begin{theorem}
  \label{thm:regular}
  Let $H$ be a regular graph. Then \HED, \HEC, and \HEE\ are incompressible
  if and only if $H$ is neither complete nor empty, where the incompressibility assumes \NOPH.
\end{theorem}
\begin{proof}
  Let $H$ be an $r$-regular graph. 
  If $H$ is either empty or complete, then by Proposition~\ref{pro:trivial-kernel}, the problems admit polynomial kernels.
  To prove the other direction, assume that $H$ is an $r$-regular graph, which is neither complete nor empty. 
  It can be easily verified that if $H$ has exactly one nonedge then $H$ must be a $2K_1$, an empty graph. 
  Similarly, if $H$ has exactly one edge then $H$ is a $K_2$, a complete graph.
  Therefore, assume that both $H$ and $\overline{H}$ has at least two edges and two nonedges.
  Now, it is sufficient to prove that either $H$ or $\overline{H}$ is 3-connected or a cycle with at least four vertices (see Propositions~\ref{pro:base:incompress:all} and \ref{pro:base:incompress:3conn}). 
  
  Suppose that $4\leq r\leq |V(H)|-5$.  
  Assume that $H$ is not 3-connected.
  Then there exists a set $S\subseteq V(H)$ such that $H-S$ is disconnected and $|S|\leq 2$. Let $A$ be the set of vertices of any component in $H-S$. 
  Since $r\geq 4$, we obtain that $|A|\geq 3$. Therefore $\overline{H-S}$ is a 3-connected graph. 
  Since $r\leq |V(H)|-5$, every vertex in $S$ has at least three
  neighbors outside $S$ in $\overline{H}$. Therefore, $\overline{H}$ is 3-connected. 
  Suppose that $r > |V(H)|-5$ (the case $r < 4$ can be handled by considering $\overline{H}$). Let $|V(H)|\geq 9$. Then every pair of non-adjacent vertices has at least three common
  neighbors. Therefore, $H$ is 3-connected. By using a computer search, we verified that if $r > |V(H)|-5$ and $|V(H)|\leq 8$, then $H$ or $\overline{H}$ is either 3-connected, or $H$ is a $2K_2$, or a  $C_4$, or a $C_5$. 
\end{proof}

%% file: larger-gang.tex
\section{Churning}
\label{sec:larger-gang}

In this section, we introduce and analyze the churning procedure. The main result of the section is that incompressibility for the class $\mathcal{W}$ of graphs defined in the previous section implies incompressibility for every graph with at least five vertices, except the trivial cases. Recall that $\mathcal{W}$ is not finite, as it contains the infinite families shown in Figure~\ref{table:shortlist-class}. In Sections~\ref{sec:smaller-gang} and \ref{sec:cai}, we will further reduce $\mathcal{W}$ to a finite set. We formally state below the main results proved in this section.

\begin{lemma}
  \label{lem:editing}
  If \HEE\ is incompressible for every $H\in \LG$, 
  then \HEE\ is incompressible for every $H$ having at least five vertices but is neither
  complete nor empty, where the incompressibility assumes \NOPH.
\end{lemma}

\begin{lemma}
  \label{lem:deletion}
  If \HED\ is incompressible for every $H\in \LG$, 
  then \HED\ is incompressible for every $H$ having at least five vertices and at least two edges but not complete, where the incompressibility assumes \NOPH.
\end{lemma}

Corollary~\ref{cor:completion} follows from Lemma~\ref{lem:deletion}, Proposition~\ref{pro:folklore} and from the fact that $\overline{\LG}=\LG$.

\begin{corollary}
  \label{cor:completion}
  If \HEC\ is incompressible for every $H\in \LG$, 
  then \HEC\ is incompressible for every $H$ having at least five vertices and at least two nonedges but not empty, where the incompressibility assumes \NOPH.
\end{corollary}




By $\mathcal{X}_E$ we denote the set of all graphs (and their complements) listed in Proposition~\ref{pro:base:incompress:all},
Proposition~\ref{pro:base:incompress:3conn}, and Theorem~\ref{thm:regular} for which the incompressibility is known (assuming \NOPH) for \HEE.
By $\mathcal{Y}_E$, we denote the set of all graphs (and their complements) listed in 
Proposition~\ref{pro:trivial-kernel} and \ref{pro:small-kernel} for which there exist polynomial kernels for \HEE; additionally, we include into $\mathcal{Y}_E$ the claw and its complement (as we do not want to conjecture the incompressibility for these cases). Similarly, we define the set $\mathcal{X}_D$ of ``hard'' and the set $\mathcal{Y}_D$ of ``nonhard'' cases for \HED. More formally,

\begin{equation*}
\begin{split}
\mathcal{X}_D &= \{ C_{\ell}, \overline{C_{\ell}}  \text{ for all } \ell\geq 4,\\
                & P_{\ell}, \overline{P_{\ell}} \text{ for all } \ell\geq 5,\\ 
                & H \text{ such that } H \text{ is regular but is neither complete nor empty,}\\
                & H \text{ such that either } H \text{ is } 3\text{-connected but not complete} \\&\text{or } \overline{H} \text{ is 3-connected with at least two nonedges} \}\\
\mathcal{X}_E &= \mathcal{X}_D\cup \{H \text{ such that } H \text{ has exactly one edge and at least five vertices}\}\\
\mathcal{Y}_E &= \{ K_t, \overline{K_t} \text{ for all } t\geq 1, \\
                & P_3, \overline{P_3}, P_4,\\
                & \text{diamond}, \overline{\text{diamond}}, \text{paw}, \overline{\text{paw}}, \text{claw}, \overline{\text{claw}}\}\\
\mathcal{Y}_D &= \mathcal{Y}_E\cup \{H \text{ such that } H \text{ has exactly one edge and at least five vertices}\}
\end{split}
\end{equation*}

Additionally we define $\mathcal{Y}'=\{P_3,$ $\overline{P_3},$ $P_4,$ $\text{claw},$ $\overline{\text{claw}},$ $\text{paw},$ $\overline{\text{paw}},$ $\text{diamond},$ $\overline{\text{diamond}}\}$.
We observe that $\mathcal{Y}'\subseteq \mathcal{Y}_E\cap \mathcal{Y}_D$ and the set of graphs with at most four vertices is a subset of $\mathcal{X}_E\cup\mathcal{Y}_E$ and $\mathcal{X}_D\cup\mathcal{Y}_D$. Further, we observe that near-empty graphs with at least five vertices are in $\mathcal{Y}_D$ but their complements are 3-connected and are in $\mathcal{X}_D$. We also note that both these graphs and their complements are in $\mathcal{X}_E$.

The main technical result of the section is the following lemma. 
It states that if a graph is not in the set $\mathcal{Y}_D$ of ``easy'' graphs, then it \simulates\ a ``hard'' graph in $
\mathcal{X_D}$ or $\mathcal{W}$, and there is a similar result for $\mathcal{X}_E$ and $\mathcal{Y}_E$. The two 
statements in the lemma are not comparable: the latter has a weaker assumption and a weaker consequence compared to the former.
\begin{lemma}
  \label{lem:churn}
  If $H\notin \mathcal{Y}_D$, then $H$ \simulates\ a graph in $\mathcal{X}_D\cup \LG$. If $H\notin \mathcal{Y}_E$, then $H$ \simulates\ a graph in $\mathcal{X}_E\cup \LG$.
\end{lemma}

In the rest of the paper, integer $\ell$ and set $V_\ell$ denote the lowest degree and the set of lowest degree vertices in $H$
respectively; integer $h$ and set $V_h$ denote the highest degree and the set of highest degree vertices in $H$ respectively; and set $V_m$
denotes the set $V(H) \setminus (V_\ell \cup V_h)$. By $h^*$ we denote the degree of vertices of $V_h$ in $\overline{H}$, i.e., $h^*=|V(H)|-h-1$.

Now we introduce a procedure (Churn) which is similar to the one used to obtain dichotomy results on the polynomial-time solvable and NP-hard cases of these problems (see Section 5 in \cite{AravindSS17}). The basic observation is that $H$ can simulate the graphs $H-V_\ell$ and $H-V_h$. This follows from proving that Construction~\ref{con:main} gives a PPT in these cases. 
\begin{proposition}[Corollary~2.9 in \cite{AravindSS17}]
  \label{pro:high-low}
  Let $H'$ be $H-V_\ell$ or $H-V_h$. Then $H$ \simulates\ $H'$.
\end{proposition}

To deal with both \HEE\ and \HED\ in a uniform way, we define $\mathcal{X}= \mathcal{X}_E$ and $\mathcal{Y} = \mathcal{Y}_D$. We observe that $\overline{\mathcal{X}\cup \mathcal{Y}}=\mathcal{X}\cup\mathcal{Y}$ and
$\mathcal{X}\cup\mathcal{Y}= \mathcal{X}_E\cup \mathcal{Y}_E=\mathcal{X}_D\cup\mathcal{Y}_D$.

\begin{mdframed}
  \textbf{Churn($H$)}: 
  \begin{description}
  \item[Step 1:] If $H$ is regular, then return $H$. 
  \item[Step 2:] If $H-V_\ell \notin \mathcal{Y}$, then return Churn($H-V_\ell)$. 
  \item[Step 3:] If $H-V_h\notin \mathcal{Y}$, then return Churn($H-V_h$). 
  \item[Step 4:] Return $H$.  
  \end{description}
\end{mdframed}
Proposition~\ref{pro:high-low} implies Corollary~\ref{cor:high-low}.

\begin{corollary}
\label{cor:high-low}
  Let $H'$ be the output of Churn($H$). Then $H$ \simulates\ $H'$.
\end{corollary}

We prove Lemma~\ref{lem:churn} by analyzing Churn() and showing that the graph returned by it always satisfies the requirements of the lemma. The procedure first handles the case when $H$ is regular. 
If $H$ is regular and not in $\mathcal{Y}_E\subseteq \mathcal{Y}_D$, then $H$ is in $\mathcal{X}_D\subseteq \mathcal{X}_E$. 
Therefore, it is safe to return $H$. 
If $H$ is not regular, then $H-V_\ell$ and $H-V_h$ are both defined. 
If one of these two graphs is not in $\mathcal{Y}$, then Proposition~\ref{pro:high-low} allows us to proceed by recursion on that graph. 
Step 4 is reached when both $H-V_\ell$ and $H-V_h$ are in $\mathcal{Y}$. 
However, at this point the conditions on $H-V_\ell$ and $H-V_h$ give us important structural information about the graph $H$, which can be exploited to show that it is in $\mathcal{X}_D\cup \LG$. Recall that $\mathcal{Y}$ is the union of complete, empty, near-empty, and the finite graphs in $\mathcal{Y}'$. 
This means we can split the problem into $4\cdot 4$ different cases, with very strict structural restrictions on $H$ in each case. These cases are analysed in a sequence of lemmas/corollaries (Lemma~\ref{lem:small-small} to Lemma~\ref{lem:small-ne} in Sections~\ref{sec:smallsmall}--\ref{sec:nearempty}). This is summarized in Figure~\ref{table:gist}.

Some of these proofs require a case analysis based on e.g., $\ell$ and $h$. As mentioned earlier, we used a systematic search of all graphs up to 11 vertices to reduce the number of corner cases that need to be handled in the proof. It may be possible to further simplify the proof by the exhaustive search of even larger graphs. But we want to point out that it would not be possible to eliminate all the case distinctions. Even if we perform a systematic search of all graphs up to a larger number of vertices, the infinite families in $\mathcal{F}$ (see Figure~\ref{table:shortlist-class}) would need to be recognized as separate cases in the proof.

\input{tables/gist}
\input{churn/small-small}
\subsection{Cliques and empty graphs}
\label{sec:c-is}
In this section, we consider the cases when both $H-V_\ell$ and $H-V_h$ are cliques or empty graphs. In this case, the structure of $H$ is very limited. In principle, we need to consider four cases separately depending on the type of $H-V_\ell$ and $H-V_h$. However, a simple complementation argument shows that the case when both of them are cliques is equivalent to the case when both of them are empty.
\input{churn/c-c}

\input{churn/c-is}
Our last case is when $H-V_\ell$ is empty and $H-V_h$ is complete. Let us observe that this case \textit{does not} follow from Lemma~\ref{lem:c-is} by complementation. If $\overline{V_\ell}$ and $\overline{V_h}$ are the lowest- and highest-degree vertices in $\overline{H}$, then $\overline{V_\ell}=V_h$, $\overline{V_h}=V_\ell$ and hence $\overline{H}-\overline{V_\ell}$ is empty and $\overline{H}-\overline{V_h}$ is a clique, that is, we have the same condition as for $H$. Fortunately, this last case is very simple to handle.
\input{churn/is-c}
\subsection{Cliques/empty graphs plus small graphs}
\label{sec:c-is-small}
Next we consider the cases when one of $H-V_\ell$ or $H-V_h$ is a clique or an empty graph, while the other is a graph from the finite set $\mathcal{Y}'$. Assuming that $H$ is not too small, this means that $H$ is essentially a clique or an empty graph, and intuitively it should follow that $H$ or $\overline{H}$ is 3-connected, respectively. However, this requires a detailed proof considering several cases.

\input{churn/small-c}
\input{churn/small-is}
\subsection{Near-empty graphs}
\label{sec:nearempty}
Finally, we consider the cases when one of $H-V_\ell$ or $H-H_h$ is near empty. These cases are similar to the corresponding ones for empty graphs, but more technical and a higher number of corner cases need to be handled. Let us remark that this part of the proof is needed only for the \HED\ problem: near-empty graphs are not in $\mathcal{Y}_E$, hence if our goal is to prove Theorem~\ref{thm:main-editing} for \HEE, then the churning procedure can recurse on such graphs.

\input{churn/c-ne}
\input{churn/ne-c}

Next we state and prove an observation which will be used in the proofs of a few lemmas in this section.

\begin{observation}
\label{obs:ene-ene}
Let $H$ be a graph with at least twelve vertices such that its vertices can be partitioned into two sets $A$ and $B$ such that the following conditions are satisfied:
\begin{enumerate}[(i)]
    \item $|A|, |B|\geq 3$;
    \item $H[A]$ and $H[B]$ are empty or near-empty;
    \item every vertex in $A$ has degree $r$ for some $r\geq 1$, and every vertex in $B$ has degree $s$ for some $s\geq 1$;
    \item every vertex in $A$ has a non-neighbor in $B$, and every vertex in $B$ has a non-neighbor in $A$;
\end{enumerate}
Then $\overline{H}$ is 3-connected.
\end{observation}
\begin{proof}
  Since $|A|, |B|\geq 3$ and $H[A]$ and $H[B]$ are empty or near-empty, we
  obtain that the complement of $H[A]$ and the complement of $H[B]$ are connected. 
  Therefore, due to condition (iv), if we remove at most two vertices from $A$ or at most two vertices from $B$ in $\overline{H}$,
  then the resultant graph remains connected. Now assume that we remove one vertex $a$ from $A$ and one vertex $b$ from $B$ in $\overline{H}$. Let the resultant
  graph be $J$. It is sufficient to prove that $J$ is connected. 
  For a contradiction, assume that $J$ is disconnected.
  If $J[A\setminus \{a\}]$ and $J[B\setminus \{b\}]$ are connected, 
  then $J$ is disconnected only when every vertex in $A$ is adjacent to all vertices except $b$ in $B$, and 
  every vertex in $B$ is adjacent to all vertices except $a$ in $A$ in $H$. This contradicts condition (iii). 
  Therefore, either $J[A\setminus \{a\}]$ or $J[B\setminus \{b\}]$ is disconnected. Without loss of generality, assume that $J[A\setminus \{a\}]$
  is disconnected. Then $|A|=3$ and $H[A]$ induces a $K_2\cup K_1$, where the $K_1$ is formed by the vertex $a$. 
  Let $A=\{a, a_1, a_2\}$. Since $H$ has at least twelve vertices, we obtain that $|B|\geq 9$. Therefore, $J[B\setminus \{b\}]$ is connected. 
  Then $J-\{a,b\}$ is disconnected only when either $a_1$ or $a_2$ is adjacent to only $b$ in $B$ in $\overline{H}$. 
  This implies that the degree of a vertex in $A$ is $2$ in $\overline{H}$. 
  This implies that the vertex $a$ is adjacent to all vertices in $B$ in $H$, which contradicts condition (iv).
\end{proof}

\input{churn/is-ne}
\input{churn/ne-is}
\input{churn/ne-ne}
\input{churn/ne-small}
\input{churn/small-ne}
\subsection{Putting it together}
\label{sec:putting}
Now we are able to formally prove the main results of the section.
\begin{proof}[Proof of Lemma~\ref{lem:churn}]
We prove by induction on the number of vertices in $H$.
To prove the first statement, let $H\notin \mathcal{Y}_D$. 
The base case is when $H\in \mathcal{X}_D\cup \LG$. Then the statement is trivial, as $H$ \simulates\ $H$. Assume that $H\notin \mathcal{X}_D\cup \LG$. 
Hence $H\notin \mathcal{Y}_D\cup \mathcal{X}_D = \mathcal{X}\cup \mathcal{Y}$. 
If $H-V_\ell\notin \mathcal{Y}$, then by the inductive hypothesis (as $H-V_\ell\notin \mathcal{Y}_D$), $H-V_\ell$ \simulates\ a graph in $\mathcal{X}_D\cup \LG$. By Proposition~\ref{pro:high-low}, $H$ \simulates\ $H-V_\ell$. Therefore, by trasitivity, $H$ \simulates\ a graph in $\mathcal{X}_D\cup \LG$. The same arguments apply when $H-V_h\notin \mathcal{Y}$. Therefore, assume that both $H-V_\ell$ and $H-V_h$ are in $\mathcal{Y}$. Then by Lemma~\ref{lem:c-c} to Lemma~\ref{lem:small-ne}, we have that $H\in \mathcal{W}$, which is a contradiction.

To prove the second statement, let $H\notin \mathcal{Y}_E$. 
The base case is when $H\in \mathcal{X}_E\cup \LG$, then the statement is trivially true. Assume that $H\notin \mathcal{X}_E\cup \LG$.
Therefore, $H\notin\mathcal{X}_E\cup \mathcal{Y}_E=\mathcal{X}\cup\mathcal{Y}$.
If $H-V_\ell\in \mathcal{Y}\setminus \mathcal{Y}_E$, then $H-V_\ell$ is a graph with exactly one edge and at least five vertices. Therefore, $H-V_\ell\in \mathcal{X}_E$. Then by Proposition~\ref{pro:high-low}, $H$ \simulates\ a graph in $\mathcal{X}_E$. Assume that $H-V_\ell\notin \mathcal{Y}_E$. Then we are done by induction hypothesis and Proposition~\ref{pro:high-low}. Therefore, $H-V_\ell\notin \mathcal{Y}$. Similar arguments apply when $H-V_h\in \mathcal{Y}$. Therefore, assume that both $H-V_\ell$ and $H-V_h$ are not in $\mathcal{Y}$. Then by Lemma~\ref{lem:c-c} to Lemma~\ref{lem:small-ne}, we obtain that $H\in \LG$, which is a contradiction.
\end{proof}

\begin{proof}[Proof of Lemma~\ref{lem:editing}]
  Assume that for every $H\in \LG$, \HEE\ is incompressible. Let $H$ be a graph with at least five vertices but neither complete nor empty.
  That is $H\notin \mathcal{Y}_E$. Then by Lemma~\ref{lem:churn}, $H$ \simulates\ a graph in $\mathcal{X}_E\cup \LG$. Then the statement follows from Proposition~\ref{pro:birdeye}. 
\end{proof}

\begin{proof}[Proof of Lemma~\ref{lem:deletion}]
  Assume that for every $H\in \LG$, \HED\ is incompressible. Let $H$ be a graph with at least five and at least two edges but not complete. That is $H\notin \mathcal{Y}_D$.
  Then by Lemma~\ref{lem:churn}, $H$ \simulates\ a graph in $\mathcal{X}_D\cup \LG$. Then the statement follows from Proposition~\ref{pro:birdeye}.
\end{proof}




%% file: tables/gist.tex
\begin{figure}
\begin{center}
\scalebox{1}{
\begin{tabular}{|l||l|l|l|l|l|}\hline
\backslashbox{$H-V_\ell$}{$H-V_h$}
&Complete&Empty&Near-empty&$\mathcal{Y'}$\\\hline\hline
Complete & Lemma~\ref{lem:c-c} & Lemma~\ref{lem:c-is} & Lemma~\ref{lem:c-ne} & Corollary~\ref{cor:c-small}\\\hline
Empty & Lemma~\ref{lem:is-c} & Corollary~\ref{cor:is-is} & Lemma~\ref{lem:is-ne} & Corollary~\ref{cor:is-small} \\\hline
Near-empty & Lemma~\ref{lem:ne-c} & Lemma \ref{lem:ne-is} & Lemma \ref{lem:ne-ne} & Lemma \ref{lem:ne-small}\\\hline
$\mathcal{Y'}$ & Lemma~\ref{lem:small-c} & Lemma~\ref{lem:small-is} & Lemma \ref{lem:small-ne} & Lemma \ref{lem:small-small}\\\hline
\end{tabular}}
\end{center}
\caption{Gist of results in Sections~\ref{sec:smallsmall} to \ref{sec:nearempty}}
\label{table:gist}
\end{figure}

%% file: churn/small-small.tex
\subsection{Small graphs}
\label{sec:smallsmall}
If both $H-V_\ell$ and $H-V_h$ are in the finite set $\mathcal{Y'}$ of graphs, then $H$ has bounded size. An exhaustive computer search showed the correctness of the procedure in this case.

\begin{lemma}
  \label{lem:small-small}
  Let $H\notin \mathcal{X}\cup \mathcal{Y}$ be such that both 
  $H-V_\ell$ and $H-V_h$ are in $\mathcal{Y'}$. Then $H\in \LG$.
\end{lemma}
\begin{proof}
  Since every graph in $\mathcal{Y}'$ has only at most four vertices, $H$ has only at most eight vertices. By a computer search we found that $H\in \LG$. 
\end{proof}

%% file: churn/c-c.tex
\begin{lemma}
  \label{lem:c-c}
  Let $H\notin \mathcal{X}\cup \mathcal{Y}$ be such that both $H - V_\ell$ and $H - V_h$ are complete graphs. Then $H\in \LG$. 
\end{lemma}
\begin{proof}
  If $H$ has only at most eleven vertices, by a computer search we found that $H\in \LG$. 
  Assume that $H$ has at least twelve vertices and $H\notin \LG$. We claim that either $H$ or $\overline{H}$ is 3-connected, which is a contradiction.
  
  Since $H-V_\ell$ and $H-V_h$ are complete graphs,
  both $V_m\cup V_\ell$ and $V_m\cup V_h$ induce complete graphs. This implies that every vertex 
  in $V_m$ is universal and hence is having the highest degree, which is a contradiction. Therefore, $V_m=\emptyset$.
  Assume that there exists at least one edge between $V_\ell$ and $V_h$.
  Then every vertex in $V_\ell$ has a neighbor in $V_h$ and every vertex in $V_h$ has a neighbor in $V_\ell$. 
  Then it can be easily verified that that $H$ is 3-connected.
  Therefore, assume that there is no edge between $V_\ell$ and $V_h$. Then $H$ is $K_t\cup K_s$ where $t>s$. 
  If $s=1$, then $H\in \overline{\mathcal{F}_2}$, a contradiction. 
  If $s=2$, then $H\in \overline{\mathcal{F}_1}$, a contradiction. 
  Therefore, $s\geq 3$. Then $\overline{H}$ is 3-connected.  
\end{proof}
Corollaries in this section and in Sections~\ref{sec:c-is-small} and \ref{sec:nearempty} use the facts that various sets we consider are self-complementary, i.e., $\mathcal{X}\cup \mathcal{Y}=\overline{\mathcal{X}\cup \mathcal{Y}}, \LG=\overline{\LG}, \mathcal{Y}'=\overline{\mathcal{Y}'}$.

\begin{corollary}
  \label{cor:is-is}
  Let $H\notin \mathcal{X}\cup \mathcal{Y}$ be such that both $H - V_\ell$ and $H - V_h$ are empty graphs. Then $H\in \LG$.
\end{corollary}

%% file: churn/c-is.tex
\begin{lemma}
  \label{lem:c-is}
  Let $H\notin \mathcal{X}\cup \mathcal{Y}$ be such that $H-V_\ell$ is a complete graph and
  $H-V_h$ is an empty graph. Then $H\in \LG$.
\end{lemma}
\begin{proof}
  If $H$ has at most eleven vertices, by a computer search we verified that $H\in \LG$. 
  Assume that $H$ has at least twelve vertices. For a contradiction, assume that $H\notin \LG$.
  Then we will show that either $H$ or $\overline{H}$ is 3-connected, which is a contradiction. 
  If $\ell\geq 3$, then $H$ is 3-connected.
  Therefore, $\ell\leq 2$.
  If $h^*=|V(H)|-h-1\geq 3$, then $\overline{H}$ is 3-connected. Therefore, $h^*\leq 2$, which implies that $h\geq 12-3=9$, as $H$ has at least twelve vertices. 
  We observe that every vertex in $V_h$ has same number of neighbors, say $t$, in $V_\ell$. 
  Since $V_m\cup V_\ell$ is an independent set and $V_m\cup V_h$ is a clique, we obtain that $|V_m|\leq 1$.
  Further, by symmetry, we can assume that $|V_\ell|\geq |V_h|$ (otherwise we can consider $\overline{H}$). 
  
  Case 1: $\ell = 0$. Then the graph is $K_r\cup sK_1$ (for $s\geq r$) and $V_m=\emptyset$. If $s\geq 3$, then $\overline{H}$ is 3-connected. If $s\leq 2$, then $H$ has only at most four vertices, which is a contradiction. 
  
  Case 2: $\ell = 1$. Since $H$ is not $K_{1,t}$ ($t\geq 9$, $\in \mathcal{F}_2$), we obtain that $|V_h|\geq 2$. 
  If $|V_h|=2$ and $t\geq 3$, then $h^*\geq 3$, a contradiction. 
  If $|V_h|=2$ and $t\leq 2$, then $|V(H)|\leq 7$, a contradiction. 
  If $|V_h|=3$ and $t\geq 2$, then $h^*\geq 3$, a contradiction. 
  If $|V_h|=3$ and $t\leq 1$, then $H$ has only at most seven vertices, which is a contradiction.  
  If $|V_h|\geq 4$, then $h^*\geq 3$, a contradiction. 
  
  Case 3: $\ell = 2$. Clearly, $|V_h|\geq 2$. If $|V_h|=2$, then $H$ is $K_2\vee sK_1$ (for $s\geq 8$), a contradiction as it is in $\mathcal{F}_3$. 
  We have that $2|V_\ell|=|V_h|(|V_\ell|-h^*)$. 
  This implies that $|V_\ell|=(h^*|V_h|)/(|V_h|-2)$. 
  Hence if $|V_h|\geq 3$, then $|V_\ell|\leq 6$ as $h^*\leq 2$. 
  Therefore, if $|V_h|\leq 4$, then $|V_\ell|\leq 6$ and $|V(H)\leq 11$, a contradiction. 
  If $|V_h|\geq 5$, then $|V_\ell|<|V_h|$, a contradiction.
\end{proof}

%% file: churn/is-c.tex
\begin{lemma}
  \label{lem:is-c}
  There exists no graph $H\notin \mathcal{X}\cup \mathcal{Y}$ such that $H-V_\ell$ is an empty graph and $H-V_h$ is a complete graph.
\end{lemma}
\begin{proof}
  The constraints imply that $H$ is a split graph with a partitioning $\{V_\ell\cup V_m, V_h\}$, where $V_\ell\cup V_m$ forms a clique and $V_h$ forms an 
  independent set. Then we obtain that the degree of a vertex in $V_\ell\cup V_m$ is at least that of a vertex in $V_h$, which is a contradiction.
\end{proof}

%% file: churn/small-c.tex
\begin{lemma}
  \label{lem:small-c}
  Let $H\notin \mathcal{X}\cup \mathcal{Y}$ be such that $H - V_\ell\in \mathcal{Y'}$ and $H - V_h$ is a complete graph. Then $H\in \LG$.
\end{lemma}
\begin{proof}
  If $H$ has only at most eleven vertices, by using a computer search we verified that $H\in \LG$. 
  Assume that $H$ has at least twelve vertices. 
  For a contradiction, assume that $H\notin \LG$. 
  We will prove that $H$ is 3-connected, which is a contradiction. 
  Since every graph in $\mathcal{Y}'$ has at most four vertices, $H-V_h$, which is a complete graph, has at least eight vertices. 
  This implies that a vertex in $V_\ell\cup V_m$ has degree at least seven and hence $h\geq 8$. 
  Since the maximum degree of every graph in $\mathcal{Y}'$ is at most $3$, every vertex in $V_h$ has at least five neighbors in $V_\ell\cup V_m$. 
  Hence $H$ is a 3-connected graph. 
\end{proof}
\begin{corollary}
  \label{cor:is-small}
  Let $H\notin \mathcal{X}\cup \mathcal{Y}$ be such that $H - V_\ell$ is an empty graph and $H - V_h\in \mathcal{Y'}$. Then $H\in \LG$.
\end{corollary}

%% file: churn/small-is.tex
\begin{lemma}
  \label{lem:small-is}
  Let $H\notin \mathcal{X}\cup \mathcal{Y}$ be such that $H - V_\ell\in \mathcal{Y'}$ and $H - V_h$ is an empty graph. Then $H\in \LG$.
\end{lemma}
\begin{proof}
  If $H$ has only at most eleven vertices, by using a computer search we verified that $H\in \mathcal{W}$. 
  Let $H$ has at least twelve vertices. 
  For a contradiction, assume that $H\notin \LG$. 
  We will show that either $H$ or $\overline{H}$ is 3-connected, which is a contradiction. 
  We observe that $|V_\ell|\geq 8$. 
  Therefore, if $\ell=0$, then $\overline{H}$ is 3-connected. 
  Therefore, $\ell\geq 1$.  
  If every vertex in $V_\ell$ is adjacent to every vertex in $V_h\cup V_m$, then $H$ is 3-connected. Therefore, assume that $\ell\leq |V_h\cup V_m|-1\leq 3$. 
  Additionally, we note that $|V_m|\leq 3$.

  Case 1: $\ell = 1$. 
  Clearly, every vertex in $V_\ell$ has exactly one neighbor in $V_h$ and has no neighbors in $V_m$ ($V_\ell\cup V_m$ forms an independent set).
  If $|V_m|=3$, then $H-V_\ell$ has an independent set of three vertices such that each has degree at least two. 
  This is not true as every graph in $\mathcal{Y}'$ has only at most four vertices. 
  Therefore, $|V_m|\leq 2$.
  Since degrees of vertices in $V_m$ is at least two, we get that $|V_h|\geq 2$ 
  (if $V_m=\emptyset$, then $|V_h|\geq 3$ as graphs in $\mathcal{Y}'$ has at least three vertices).
  Therefore, if every vertex in $V_h$ has at least $5-|V_h|$ neighbors in $V_\ell$ 
  (as $\ell=1$, there are no two vertices in $V_h$ having a common neighbor in $V_\ell$), 
  then every vertex in $V_h$ has at least $(|V_h|-1)(5-|V_h|)\geq 3$ non-neighbors in $V_\ell$. 
  Therefore, $\overline{H}$ is 3-connected. 
  If a vertex in $V_h$ has only at most $4-|V_h|$ neighbors in $V_\ell$, 
  then $h^*\geq |V_\ell|-(4-|V_h|) \geq 8-(4-|V_h|)=|V_h|+4\geq 6$. 
  Since $H-V_\ell$ has only at most four vertices, we obtain that every vertex in $V_h$ has at least three non-neighbors in $V_\ell$.
  Therefore, $\overline{H}$ is 3-connected.
  
  Case 2: $\ell=2$. Since every graph in $\mathcal{Y}'$ has only at most four vertices, 
  there are no two nonadjacent vertices in it having degree at least three. 
  Therefore, $|V_m|=0$ or 1. 
  Therefore, $|V_h|=3$ or 4 (the case $|V_h|=2$ and $|V_m|=1$ does not arise as every vertex in $V_m$ should have degree at least 3).
  
  Suppose $|V_h|=3$.
  Since $H-V_\ell$ is not a clique and every vertex in $V_\ell$ is nonadjacent to exactly one vertex in $V_h$, we obtain that the sum of degrees of vertices of $V_h$ in $\overline{H}$ 
  is at least $|V_\ell|+2\geq 10$ (if $V_m$ is nonempty, then the vertex in $V_m$ is adjacent to all the three vertices in $V_h$ and there must be a missing edge among the vertices in $V_h$). This means that there is a vertex in $V_h$ 
  whose degree is at least 4 in $\overline{H}$, which is then true for every vertex 
  in $V_h$. 
  Then it 
  is easy to see that $\overline{H}$ is 3-connected. 
  
  Suppose that $|V_h|=4$.
  Since there are $2|V_\ell|$ edges between $V_\ell$ and $V_h$, at least one vertex in $V_h$ has at most $2|V_\ell|/4$ neighbors in $V_\ell$. Therefore, there is a vertex
  in $V_h$ that is adjacent to at least $|V_\ell|/2\geq 4$ vertices of $V_\ell$ in $\overline{H}$.
  So the degree of every vertex in $V_h$ is at least four in $\overline{H}$. 
  Then $\overline{H}$ is 3-connected 
  (note that it is not possible to separate 3 vertices of $V_h$ from the rest of the graph by deleting at most two vertices, 
  as any set $\{x,y,z\}\subseteq V_h$ is adjacent to every vertex in $V_\ell$).

  Case 3: $\ell = 3$. Since $\ell\leq |V_h\cup V_m|-1$, $|V_h\cup V_m|=4$. 
  Since every graph in $\mathcal{Y}'$ has only at most four vertices, there is no vertex in it with degree at least four.
  Therefore, $V_m=\emptyset$.
  Let $V_h=\{w,x,y,z\}$. 
  We claim that $H$ is 3-connected. For a contradiction, assume that the vertex-connectivity of $H$ is at most two. 
  We partition $V_\ell$ into four sets $V_w, V_x, V_y,$ and $V_z$, where $V_u$ is the set of all vertices in $V_\ell$ not adjacent to $u$, for $u\in \{w,x,y,z\}$.
  If any of these sets is empty, then there is a vertex in $V_h$ adjacent to all vertices in $V_\ell$. Since the difference in degrees of vertices in graphs in 
  $\mathcal{Y}'$ is at most two, we obtain that every other set in the partition has size at most two. 
  Then $H$ has only at most ten vertices, which is a contradiction.
  Therefore, every set in the partition is nonempty. Consider an auxiliary graph $J$ such that it is a bipartite graph with partition $V_h$ and  $U_\ell$,
  where $U_\ell$ has a vertex for each set in the partition. A vertex, say $w$ in $V_h$ is adjacent to a vertex $U_u$ in $U_\ell$, if and only if 
  $w$ is adjacent to all vertices in $V_u$ in $H$. It is straight-forward to verify that $J$ is 3-connected. It implies that $H$ is 3-connected.
\end{proof}
\begin{corollary}
  \label{cor:c-small}
  Let $H\notin \mathcal{X}\cup \mathcal{Y}$ be such that $H - V_\ell$ is a complete graph and $H - V_h\in \mathcal{Y'}$. Then $H\in \LG$.
\end{corollary}

%% file: churn/c-ne.tex
\begin{lemma}
  \label{lem:c-ne}
  Let $H\notin \mathcal{X}\cup \mathcal{Y}$ be such that $H-V_\ell$ is a complete graph and $H-V_h$ is a near-empty graph. Then $H\in \LG$.
\end{lemma}
\begin{proof}
  If $H$ has at most eleven vertices, then by a computer search we found that $H\in \LG$. Therefore, let $H$ has at least twelve vertices. 
  For a contradiction, assume that $H\notin \LG$. We will prove that either $H$ or $\overline{H}$ is 3-connected, a contradiction.
  
  If $h^*=|V(H)|-h-1\geq 3$, then $\overline{H}$ is 3-connected. 
  Therefore, assume that $h^*\leq 2$. 
  If $|V_\ell\cup V_m|\leq 4$, then $V_\ell\cup V_m$ induces either a $K_2$ or a graph in $\mathcal{Y}'$.
  Then the statement follows from Lemma~\ref{lem:c-c} and Corollary~\ref{cor:c-small}.
  Therefore, assume that $|V_\ell\cup V_m|\geq 5$. 
  Since $V_h\cup V_m$ induces a clique and $V_\ell\cup V_m$
  induces a near-empty graph, we obtain that $|V_m|\leq 2$.
  Therefore, $|V_\ell|\geq 3$.
  Clearly, every vertex in $V_h$ has the same number $t\geq 0$ of neighbors in $V_\ell$. 
  Since $h^*\leq 2$ and $|V_\ell|\geq 3$, we obtain that $t\geq 1$, which means that there is a vertex in $V_\ell$
  which has at least one neighbor in $V_h$.
  Hence $\ell\geq 1$.
  If $\ell\geq 4$, then $H$ is 3-connected.
  Therefore, $1\leq \ell\leq 3$.
  Every vertex in $V_\ell$ is adjacent to at least one vertex in $V_h$ except possibly for two vertices (due to the single edge in $H-V_h$).
  Let $uv$ be the edge in $H[V_\ell\cup V_m]$ and let $x=|V_\ell\cap \{u,v\}|$.
  We observe that the number of edges between $V_h$ and $V_\ell$ is $t|V_h| = \ell|V_\ell|-x$. 
  Since $h^*\leq 2$, $t\geq |V_\ell|-2$.
  Thus we obtain that 
  \begin{equation}
      \label{eq:c-ne}
      |V_\ell|\cdot (|V_h|-\ell)\leq 2|V_h|-x
  \end{equation}
  

  
  Case 1: $\ell =1$. Let $|V_h|\geq 2$. Then by (\ref{eq:c-ne}), $|V_\ell|\leq 2|V_h|/(|V_h|-1)$. 
  Then $|V_\ell|\geq 3$ is possible only when $|V_h|$ is 2 or 3, but then $|V_m|\leq 2$ implies that $H$
  has less than 12 vertices, a contradiction.
  Therefore, $|V_h|=1$.
  If $V_m=\emptyset$, then $H$ is $K_{1,t}\cup K_2$ ( $\in\mathcal{F}_7$), a contradiction. 
  If $|V_m|=1$, then $H$ is $T_{t,1}$ ($\in\mathcal{F}_4$), a contradiction. 
  If $|V_m|=2$, then $H$ is $\overline{(K_{t+2}-e)\cup K_1}$ ($\in\mathcal{F}_8$), a contradiction. 
  
  Case 2: $\ell=2$. 
  Clearly, $|V_h|\geq 2$. 
  If $|V_h|\geq 7$, then by (\ref{eq:c-ne}), $|V_\ell|\leq 2$, a contradiction. Therefore, $2\leq |V_h|\leq 6$. 
  
  Let $|V_h|=2$. 
  If $V_m=\emptyset$, then $H$ is $J_t$ ($\in\mathcal{F}_9$), a contradiction. 
  It can be easily verified that there is no $H$ with $|V_m|=1$. If $|V_m|=2$, then $H$ is $\overline{(K_{t+2}-e)\cup 2K_1}$ ($\in\mathcal{F}_5$), a contradiction. 
  
  Therefore, $|V_h|\in \{3, 4, 5, 6\}$. 
  By (\ref{eq:c-ne}), if $|V_h|=6$, then $|V_\ell|\leq 3$. Since $H$ has at least twelve vertices, $|V_m|\geq 3$, a contradiction. 
  Hence assume that $|V_h|\in \{3,4,5\}$. Then by (\ref{eq:c-ne}), $|V_\ell|\leq 6$. Thus $|V_m|\geq 3$, a contradiction.

  Case 3: $\ell=3$.
  Clearly, $|V_h|\geq 3$.
  It is easy to verify that $H$ is 3-connected unless both $u$ and $v$ are in $V_\ell$ (i.e., $x=2$). 
  Therefore, $x=2$ and $|V_m|\leq 1$.
  Then by (\ref{eq:c-ne}), 
  \begin{equation}
  \label{eq:c-ne2}
  |V_\ell|\cdot (|V_h|-3)\leq 2|V_h|-2    
  \end{equation}
  If $|V_h|\geq 8$, we obtain from (\ref{eq:c-ne2}) that $|V_\ell|\leq 2$, a contradiction.
  Therefore, $3\leq |V_h|\leq 7$.
  
  If $|V_h|\in \{4,5,6,7\}$, then using (\ref{eq:c-ne2}), we obtain that $|V_\ell|+|V_h|+|V_m| < 12$, a contradiction.
  Let $|V_h|=3$.
  Then the number of edges between $V_\ell$ and $V_h$ is $3|V_\ell|-2$ which is not a 
  multiple of $|V_h|$, a contradiction.
  
  
  
  
\end{proof}

%% file: churn/ne-c.tex
\begin{lemma}
  \label{lem:ne-c}
  There exists no $H\notin \mathcal{X}\cup \mathcal{Y}$ such that $H-V_\ell$ is a near-empty graph and $H-V_h$ is a complete graph. 
\end{lemma}
\begin{proof}
  By a computer search, we found that there exists no such graph $H$ with at most eleven vertices. 
  For a contradiction, assume that there exists such a graph $H$ with at least twelve vertices. 
  Let $H-V_h$ has $s$ vertices. 
  Since it forms a clique, $\ell\geq s-1$. 
  Therefore, $h\geq s$. 
  This implies that a vertex in $V_h$ is adjacent to at least $s-1$ vertices in $H-V_h$ 
  (recall that $H[V_h]$ has only at most one edge). 
  Therefore, a vertex in $V_m\cup V_\ell$ has degree at least $s$ and hence $h\geq s+1$. 
  Therefore, every vertex in $V_h$ has degree 1 in $H[V_h]$ and is adjacent to every vertex in $V_m\cup V_\ell$. 
  This implies that vertices in $V_\ell$ are universal, which is a contradiction.
\end{proof}

%% file: churn/is-ne.tex
\begin{lemma}
  \label{lem:is-ne}
  Let $H\notin \mathcal{X}\cup \mathcal{Y}$ such that $H - V_\ell$ is an empty graph and $H - V_h$ is a near-empty graph. Then $H\in \LG$.
\end{lemma}
\begin{proof}
  If $H$ has only at most eleven vertices, by a computer search we found that $H\in \LG$. 
  Assume that $H$ has at least twelve vertices. 
  For a contradiction, assume that $H\notin\LG$. 
  Then we will show that either $H$ or $\overline{H}$ is 3-connected, a contradiction.

  If $H-V_h$ is a $K_2$, then by Lemma~\ref{lem:is-c}, $H$ does not exist. Therefore, let $|V_m\cup V_\ell|\geq 3$. 
  If $\ell=0$, then $H[V_m]$ contains an edge, which is a contradiction as $V_h\cup V_m$ is an independent set. 
  Therefore, $\ell\geq 1$. Since $H-V_\ell$ has no edge and $H-V_h$ has exactly one edge, a vertex in $V_m$ has degree at most 1, which is a contradiction as it is not more than $\ell$. Therefore, $V_m=\emptyset$. 
  Hence $|V_\ell|\geq 3$.
  Further, $|V_h|\geq \ell$ and $|V_\ell|\geq h$.
  If a vertex in $V_h$ is adjacent to all vertices in $V_\ell$, then every vertex in $V_h$ is adjacent to all vertices in $V_\ell$. 
  Then there will be a discrepancy in the degrees of vertices in $V_\ell$ (due to the single edge in $H[V_\ell]$ and the fact that $|V_\ell|\geq 3$).
  Therefore, $h<|V_\ell|$. 
  Let $|V_h|\geq 3$.
  If $\ell = |V_h|$, then $H$ is 3-connected.
  Hence $\ell \leq |V_h|-1$.
  Therefore, every vertex in $V_h$ has a non-neighbor in $V_\ell$ and every vertex in $V_\ell$ has a non-neighbor in $V_h$ in $H$. 
  Hence by Observation~\ref{obs:ene-ene}, $\overline{H}$ is 3-connected.
  Therefore, $|V_h|\leq 2$. Hence $\ell=1$ or $2$.
  Let $\ell=1$. 
  Then $H$ is $sK_{1,h}\cup K_2$. If $s>1$, then $\overline{H}$ is 3-connected. 
  When $s=1$, $H$ is $K_{1,h}\cup K_2$ ($\in\mathcal{F}_7$), a contradiction.
  Let $\ell=2$. Then $|V_h|=2$ and $H$ is $Q_t$ ($\in\mathcal{F}_{10}$), a contradiction. 
\end{proof}

%% file: churn/ne-is.tex
\begin{lemma}
  \label{lem:ne-is}
  Let $H\notin \mathcal{X}\cup \mathcal{Y}$ such that $H - V_\ell$ is a near-empty graph and $H - V_h$ is an empty graph. Then $H\in \LG$.
\end{lemma}
\begin{proof} 
  If $H$ has only at most eleven vertices, then by a computer search we found that $H\in \LG$. 
  Let $H$ has at least twelve vertices. 
  For a contradiction, assume that $H\notin\LG$. We will show that $H$ or $\overline{H}$ is 3-connected, a contradiction.

  If $H-V_\ell$ is a $K_2$, then the statement follows from Lemma~\ref{lem:c-is}. Therefore, let $|V_m\cup V_h|\geq 3$. 
  If $\ell=0$, then $H$ is $K_2\cup |V_\ell|K_1$ ($\in \mathcal{X}\cup \mathcal{Y}$), a contradiction. Therefore, $\ell\geq 1$. Since $H-V_h$ has no edge and $H-V_\ell$ has exactly one edge, a vertex in $V_m$ has degree at most 1, which is a contradiction as it is not more than $\ell$. Therefore, $V_m=\emptyset$. 
  Hence $|V_h|\geq 3$.
  Further, $|V_h|\geq \ell$ and $|V_\ell|\geq h$.
  If a vertex in $V_\ell$ is adjacent to all vertices in $V_h$, then every vertex in $V_\ell$ is adjacent to all vertices in $V_h$. 
  Then there will be a discrepancy in the degrees of vertices in $V_h$ (due to the single edge in $H[V_h]$ and the fact that $|V_h|\geq 3$).
  Therefore, $\ell<|V_h|$. 
  Let $|V_\ell|\geq 3$.
  Then it can be easily verified that $h<|V_\ell|$.
  Therefore, every vertex in $V_h$ has a non-neighbor in $V_\ell$ and every vertex in $V_\ell$ has a non-neighbor in 
  $V_h$ in $H$. Hence by Observation~\ref{obs:ene-ene}, $\overline{H}$ is 3-connected.
  Therefore, $|V_\ell|\leq 2$. Hence $h=1$ or $2$. Since $\ell\geq 1$,  $h=2$.
  Then $\ell=1$. Then $H$ is $P_4\cup sP_3$, for $s\geq 3$ (as $|V(H)|\geq 12$).
  Then $\overline{H}$ is 3-connected.
  %
\end{proof}

%% file: churn/ne-ne.tex
\begin{lemma}
  \label{lem:ne-ne}
  Let $H\notin \mathcal{X}\cup \mathcal{Y}$ be such that both $H - V_\ell$ and $H-V_h$ are near-empty graphs. Then $H\in \LG$.
\end{lemma}
\begin{proof}
  If $H$ has only at most eleven vertices, by a computer search we found that $H\in \LG$. 
  Assume that $H$ has at least twelve vertices. For a contradiction, assume that $H\notin \LG$. 
  Then we will show that either $H$ or $\overline{H}$ is 3-connected, a contradiction. 
  
  If $H-V_\ell$ is a $K_2$, then the statement follows from Lemma~\ref{lem:c-ne}. Therefore, let $|V_m\cup V_h|\geq 3$.
  If $H-V_h$ is a $K_2$, then by Lemma~\ref{lem:ne-c}, $H$ does not exist.
  Therefore, let $|V_\ell\cup V_m|\geq 3$. 
  Let $\ell = 0$.
  Then $V_\ell$ and $V_h$ are independent sets and $V_m$ induces a $K_2$.
  This is a contradiction as there are no edges incident to vertices in $V_h$.
  Therefore, $\ell\geq 1$.
  Since $H-V_h$ and $H-V_\ell$ have exactly one edge each, a vertex in $V_m$ has degree at most 2. 
  Therefore, if $V_m\neq \emptyset$, then $|V_m|=1$, $\ell=1$, and $h\geq 3$. 
  Then $H$ is $T_{h-1,1}\cup sK_{1,h}$, for $s\geq 1$ (recall that $|V_h\cup V_m|\geq 3$).
  Then $\overline{H}$ is 3-connected. 
  Therefore, $V_m=\emptyset$.
  We observe that there is at least one isolated vertex in $H[V_\ell]$ and in $H[V_h]$.
  Let $uv$ be the edge in $H[V_\ell]$ and $xy$ be the edge in $H[V_h]$.
  If a vertex in $V_\ell$ is adjacent to all vertices in $V_h$ or 
  if a vertex in $V_h$ is adjacent to all vertices in $V_\ell$, then
  $H$ has all possible edges between $V_\ell$ and $V_h$ except two. In that case, the two nonedges between 
  $V_\ell$ and $V_h$ are $ux$ and $vy$ or $uy$ and $vx$. Then $H$ is 3-connected.
  Therefore, every vertex in $V_h$ has a non-neighbor in $V_\ell$ and every vertex in $V_\ell$ has a non-neighbor in $V_h$ 
  in $H$. Hence by Observation~\ref{obs:ene-ene}, $\overline{H}$ is 3-connected.
\end{proof}

%% file: churn/ne-small.tex
\begin{lemma}
  \label{lem:ne-small}
  Let $H\notin \mathcal{X}\cup \mathcal{Y}$ be such that $H - V_\ell$ is a near-empty graph and $H - V_h\in \mathcal{Y'}$. Then $H\in \LG$. 
\end{lemma}
\begin{proof}
  If $H$ has only at most eleven vertices, by using a computer search we verified that $H\in \LG$. 
  Let $H$ has at least twelve vertices. Since every graph in $\mathcal{Y}'$ has only at most four vertices, $V_h$ has at least eight vertices. 
  Since $H[V_h]$ has at most one edge and $H[V_\ell\cup V_m]$ has at least one edge, a simple degree counting gives us a contradiction to the fact that 
  $\ell < h$.
\end{proof}

%% file: churn/small-ne.tex
\begin{lemma}
  \label{lem:small-ne}
  Let $H\notin \mathcal{X}\cup \mathcal{Y}$ be such that $H - V_\ell\in \mathcal{Y'}$ and $H - V_h$ is a near-empty graph. Then $H\in \LG$.
\end{lemma}
\begin{proof}  
  If $H$ has only at most eleven vertices, by using a computer search we verified that $H\in \LG$. 
  Assume that $H$ has at least twelve vertices. 
  For a contradiction, assume that $H\notin \LG$. 
  We will show that either $H$ or $\overline{H}$ is 3-connected, a contradiction. 
  
  We observe that $|V_\ell|\geq 8$ if $|V_h\cup V_m|=4$, 
  and $|V_\ell|\geq 9$ if $|V_h\cup V_m|=3$ 
  (recall that a graph in $\mathcal{Y}'$ has either 3 or 4 vertices). 
  
  Therefore, if $\ell=0$, then $\overline{H}$ is 3-connected. Hence assume that $\ell\geq 1$. 
  If a vertex in $V_\ell$ is adjacent to all vertices in $V_h\cup V_m$, then every degree-0 vertex in $H[V_\ell]$ is adjacent to every vertex in $V_h\cup V_m$. 
  Since there are at least three vertices in $V_\ell$ adjacent to all vertices in $V_h\cup V_m$, $H$ is 3-connected. 
  Therefore, assume that none of the vertices in $V_\ell$ is adjacent to all vertices in $V_h\cup V_m$. 
  Therefore, $\ell\leq |V_h\cup V_m|-1$. 
  Let $H[V_\ell\cup V_m]$ be $K_2\cup tK_1$, for $t\geq 6$. 
  If $|V_m|=3$, then the single edge in $H[V_\ell\cup V_m]$ and the edges from the single vertex in $V_h$ cannot give degree at least two for every vertex in $V_m$.
  Therefore, $|V_m|\leq 2$. 
  Further, since $H[V_\ell\cup V_m]$ induces a $K_{t+2}-e$ in $\overline{H}$, 
  if every vertex in $V_h$ has at least three non-neighbors in $V_\ell\cup V_m$ in $H$, then $\overline{H}$ is 3-connected. 
  Therefore, assume that there exists a vertex in $V_h$ with at most two non-neighbors in $V_\ell\cup V_m$.
  Since the difference in degrees of vertices in graphs in $\mathcal{Y}'$ is at most two, we obtain that every vertex in
  $V_h$ has only at most four non-neighbors in $V_\ell\cup V_m$. 
  
  Case 1: $\ell = 1$. Since there is a vertex in $V_h$ with at most two non-neighbors in $V_\ell$ and 
  every vertex in $V_h$ has only at most four non-neighbors in $V_\ell$, 
  we obtain that, if $|V_h|\geq 2$, then there is a vertex in $V_\ell$ with degree at least 2, a contradiction.
  Therefore, $|V_h|\leq 1$. 
  Therefore, $|V_m|=2$.
  Since both the vertices in $V_m$ has degree at least two, both must be mutually adjacent and adjacent to the vertex in $V_h$. Then $H[V_h\cup V_m]$ is a $K_3$, a contradiction. 
  
  
  
  Case 2: $\ell = 3$. Since $\ell\leq |V_h\cup V_m|-1$, we obtain that $|V_h\cup V_m|=4$. Since there is no vertex in $\mathcal{Y}'$ with degree more than three, we obtain that either $V_m=\emptyset$ or $|V_m|=1$ and the vertex in $V_m$ has exactly three neighbors in $V_h$ and one neighbor in $V_\ell$. 
  
  Case 2a: $V_m=\emptyset$. 
  Let $u,v$ be the two adjacent vertices in $H[V_\ell]$.
  We observe that $V_\ell$ induces $K_{|V_\ell|}-e$ in $\overline{H}$.
  Since $\ell=3$ and $|V_h|=4$, $u$ and $v$ has exactly
  two neighbors each in $V_h$ in $\overline{H}$.
  Further, every other vertex in $V_\ell$ has exactly one neighbor each in $V_h$ in $\overline{H}$.
  Therefore, there are $2\cdot 2 + (|V_\ell|-2) = |V_\ell|+2$ edges between $V_\ell$ and $V_h$ in $\overline{H}$.
  Let $p$ be the number of edges in the graph induced by $V_h$ in $\overline{H}$.
  We observe that $1\leq p\leq 5$.
  Then we have 
  \begin{equation}
      \label{eq:small-ne}
      |V_\ell|+2=4h^*-2p
  \end{equation}
  Recall that $|V_\ell|\geq 8$.
  Therefore, if $h^*\leq 2$, we get a contradiction by (\ref{eq:small-ne}).
  Let $h^*=3$.
  Then by (\ref{eq:small-ne}), $p=1$ and $|V_\ell|=8$. 
  This implies that $H[V_h]$ is a diamond. Since $H$ is not 3-connected, 
  $u$ and $v$ are adjacent to two vertices $\{x,y\}$ in $V_h$.
  This gives rise to three cases - (i) $x$ and $y$ are the two degree-3 vertices in the diamond induced by $H[V_h]$;
  (ii) $x$ and $y$ are the two degree-2 vertices in the diamond;
  (iii) $x$ is a degree-2 vertex and $y$ is a degree-3 vertex in the diamond.
  It can be verified that in case (i), $H$ is $S_{35}$ and in case (ii) and (iii),
  $\overline{H}$ is 3-connected.
  If $h^*\geq 4$, it can be easily verified that $\overline{H}$ is 3-connected.
  
  Case 2b: $|V_m|=1$ and the vertex, say $v_m$, in $V_m$ has exactly three neighbors in $V_h$ and one neighbor, say $w$, in $V_\ell$. 
  Further, $V_\ell$ forms an independent set of at least eight vertices. Since $\ell=3$, every vertex in $V_\ell$, except $w$ is adjacent to all the three vertices in $V_h$. Then $w$ must be adjacent to exactly two vertices in $V_h$. 
  This gives rise to a 3-connected graph $H$, which is a contradiction.
  
  Case 3: $\ell=2$. 
  
  Case 3a: $|V_h\cup V_m|=3$. Since a vertex in $V_m$ has degree at least three, we obtain that either $V_m=\emptyset$ or $|V_m|=1$ and the vertex in $V_m$ is adjacent to both the vertices in $V_h$ and adjacent to exactly one vertex in $V_\ell$. 
  
  Case 3ai: $V_m=\emptyset$. Clearly $V_\ell$ induces $K_2\cup tK_1$, for $t\geq 7$. 
  As claimed earlier, there is a vertex in $V_h$ non-adjacent to only at most two vertices in $V_\ell$. 
  Since the difference in degrees of vertices in $H[V_h]$ is only at most one, we obtain that the other two vertices in $V_h$ is adjacent to at least $t-1$ vertices in $V_\ell$. 
  Then there are at least $3t-2$ edges between $V_h$ and $V_\ell$. Since every vertex (except for two adjacent vertices) in $V_\ell$
  has exactly two edges to $V_h$, this count must be $2t+2$. 
  This is a contradiction as $3t-2 > 2t+2$ for $t\geq 7$.
  
  Case 3aii: $|V_m|=1$ and the vertex in $V_m$ is adjacent to both the vertices in $V_h$ and adjacent to exactly one vertex, say $v$, in $V_\ell$. Then all the vertices in $V_\ell$, except $v$, is adjacent to both the vertices in $V_h$. Then $v$ must be adjacent to exactly one vertex in $V_h$, which is a contradiction.
  
  Case 3b: $|V_h\cup V_m|=4$. 
  
  Case 3bi: $V_m=\emptyset$. Since $H$ has at least 12 vertices, $V_\ell$ induces $K_2\cup tK_1$ for some $t\geq 6$. 
  Since there is a vertex in $V_h$ adjacent to at least $t$ vertices in $V_\ell$ and the difference in degrees of vertices in every graph in $\mathcal{Y}'$ is at most two, we obtain that there are at least $t+3(t-2)=4t-6$ edges between $V_h$ and $V_\ell$. Since $\ell=2$ and $H[V_\ell]$ induces $K_2\cup tK_1$, we obtain that there are exactly $2t+2$ edges between $V_h$ and $V_\ell$. Therefore, $t\leq 4$, which is a contradiction. 
  
  Case 3bii: $|V_m|=1$. Let $v_m$ be the vertex in $V_m$. 
  Assume that $v_m$ is not adjacent to $V_\ell$. 
  Then $V_\ell$ induces $K_2\cup tK_1$ for some $t\geq 6$.
  Since $v_m$ has degree at least three, $v_m$ must be adjacent to all the three vertices in $V_h$. 
  Recall that there exists a vertex in $V_h$ which is adjacent to at least $t$ vertices in $V_\ell$. 
  Since the difference in degrees of vertices in $V_h$ in $H[V_h\cup V_m]$ is at most one, 
  there are at least $t+2(t-1)=3t-2$ edges between $V_h$ and $V_\ell$. 
  Since $H[V_\ell]$ induces $K_2\cup tK_1$, there are exactly $2t+2$ edges between $V_\ell$ and $V_h$.
  Therefore, $2t+2\geq 3t-2$. 
  Then $t\leq 4$, which is a contradiction. 
  Assume that $v_m$ is adjacent to a vertex, say $w$, in $V_\ell$. 
  Then $H[V_\ell]$ is an independent set of at least $t+1\geq 8$ vertices. 
  There are exactly $2(t+1)-1= 2t+1$ edges between $V_\ell$ and $V_h$.
  Since there exists a vertex in $V_h$ with only at most two non-neighbors in $V_\ell$,
  and the degrees of the vertices in $V_h$ in $H[V_h\cup V_m]$ differs only by one,
  we obtain that there are at least $(t-1)+2(t-2) = 3t-5$ edges between $V_\ell$ and $V_h$.
  Therefore, $2t+1 \geq 3t-5$. Then $t\leq 6$, which is a contradiction.
  
  Case 3biii: $|V_m|=2$. 
  Let $x,y\in V_m$. 
  Assume that $x$ and $y$ are adjacent. 
  Then $V_\ell$ is an independent set of size at least $t\geq 8$ vertices. 
  Since vertices in $V_m$ has degree at least 3, both $x$ and $y$ are adjacent to both the vertices in $V_h$ and $H[V_h\cup V_m]$ is a diamond. 
  Further, every vertex in $V_\ell$ is adjacent to both the vertices in $V_h$. 
  Then the graph is $\overline{(K_{t+2}-e)\cup K_2}$ ($\in\mathcal{F}_6$), a contradiction. 
  Assume that $x$ and $y$ are nonadjacent. 
  Then either $x$ or $y$ cannot have degree at least three, a contradiction.
\end{proof}

%% file: reductions.tex
\section{Reductions}
\label{sec:smaller-gang}

Recall that we defined $\mathcal{W}=\mathcal{H}\cup \overline{\mathcal{H}}\cup \mathcal{A}\cup \overline{\mathcal{A}}\cup \mathcal{D}\cup\overline{\mathcal{D}}\cup \mathcal{B}\cup\overline{\mathcal{B}}\cup \mathcal{S}\cup\overline{\mathcal{S}}\cup \mathcal{F}\cup\overline{\mathcal{F}}$ and Section~\ref{sec:larger-gang} reduced our main questions to assuming incompressibility for the set $\mathcal{W}$. In this section, we further refine the result and show that incompressibility needs to be assumed only for the finite set 
$\LGD=\mathcal{H}\cup \overline{\mathcal{H}}\cup \mathcal{A}\cup \overline{\mathcal{A}}\cup \mathcal{B}\cup \mathcal{D}$.
That is, we recall and introduce some further simple reductions and use them to prove that every graph in $\LG\setminus \LGD$ \simulates\ a graph in $\LGD\cup \mathcal{X}_D$. 
Summary of results in this section handling graphs in $\mathcal{S}$ and $\mathcal{F}$ are given in Figure~\ref{table:summary-s} and \ref{table:summary-class} respectively.
\input{tables/summary}

To begin with, we observe that deleting the lowest degree vertices in the graphs in $\overline{\mathcal{B}}\cup\overline{\mathcal{D}}$ 
results in 3-connected graphs which are not complete. Then by Proposition~\ref{pro:high-low}, we have:

\begin{proposition}
  \label{pro:lgd-direct}
If  $H\in \overline{\mathcal{B}}\cup\overline{\mathcal{D}}$, then \HEE\ and \HED\ are incompressible, assuming \NOPH. 
\end{proposition}

Proofs in the rest of this section are written only for \textsc{Editing}. The proofs can be replicated for \textsc{Deletion} and \textsc{Completion} by replacing `\textsc{Editing}' with `\textsc{Deletion}' and `\textsc{Completion}' respectively. The reductions are based on Construction~\ref{con:main} and a few other similar constructions.

\input{reductions-1}
\input{reductions-2}
\input{reductions-3}
\input{reductions-4}


Lemma~\ref{lem:smaller-gang} follows from Corollary~\ref{cor:folklore}, Proposition~\ref{pro:lgd-direct}, the transitivity of PPTs, and other results in this section (see Figures~\ref{table:summary-s} and \ref{table:summary-class}) for details.
\begin{lemma}
  \label{lem:smaller-gang}
  Let $H\in \LG\setminus \LGD$. Then $H$ \simulates\ a graph in $\LGD\cup \mathcal{X}_D$. 
\end{lemma}

%% file: tables/summary.tex
\begin{figure}
\begin{center}
\scalebox{0.75}{
\begin{tabular}{| c | c | c | c | c | c | c | c | c | c | c |}
\cline{1-3} \cline{5-7} \cline{9-11}
\textbf{$H$} & \textbf{Simulates} & \textbf{By} & & \textbf{$H$} & \textbf{Simulates} & \textbf{By} & & \textbf{$H$} & \textbf{Simulates} & \textbf{By}\\ \cline{1-3} \cline{5-7}\cline{9-11} 
$S_1$ & $C_4$ & Lemma~\ref{lem:k23} & & 
  $S_{13}$ & $S_3$  & Corollary~\ref{cor:mod} & &
    $S_{25}$ & $\overline{H_1}$ & Corollary~\ref{cor:mod}\\ \cline{1-3} \cline{5-7} \cline{9-11} 
$S_2$ & $\overline{H_7}$ & Corollary~\ref{cor:near-uni} & & 
  $S_{14}$ & $B_1$  & Corollary~\ref{cor:mod} & &
    $S_{26}$ & $S_8$ & Corollary~\ref{cor:mod}\\ \cline{1-3} \cline{5-7} \cline{9-11} 
$S_3$ & $H_6$ & Corollary~\ref{cor:near-uni} & & 
  $S_{15}$ & $H_9$  & Corollary~\ref{cor:path} & &
    $S_{27}$ & a graph in $\mathcal{F}_6$ & Corollary~\ref{cor:cut}\\ \cline{1-3} \cline{5-7} \cline{9-11} 
$S_4$ & $H_2$ & Corollary~\ref{cor:jutk1} & & 
  $S_{16}$ & $S_3$  & Corollary~\ref{cor:near-uni} & &
    $S_{28}$ & $\overline{A_9}$ & Corollary~\ref{cor:mod}\\ \cline{1-3} \cline{5-7} \cline{9-11} 
$S_5$ & $C_4$ & Corollary~\ref{cor:path} & & 
  $S_{17}$ & $\overline{H_1}$  & Corollary~\ref{cor:near-uni} & &
    $S_{29}$ & $S_{17}$ & Corollary~\ref{cor:mod}\\ \cline{1-3} \cline{5-7} \cline{9-11} 
$S_6$ & $H_6$ & Corollary~\ref{cor:jutk1} & & 
  $S_{18}$ & $\overline{A_1}$  & Corollary~\ref{cor:mod} & &
    $S_{30}$ & $\overline{S_{19}}$ & Corollary~\ref{cor:mod}\\ \cline{1-3} \cline{5-7} \cline{9-11} 
$S_7$ & $H_9$ & Corollary~\ref{cor:mod} & & 
  $S_{19}$ & $\overline{S_3}$  & Corollary~\ref{cor:mod} & &
    $S_{31}$ & $S_3$ & Corollary~\ref{cor:near-uni}\\ \cline{1-3} \cline{5-7} \cline{9-11} 
$S_8$ & $\overline{A_1}$ & Corollary~\ref{cor:mod} & & 
  $S_{20}$ & a graph in $\mathcal{F}_1$  & Corollary~\ref{cor:cut} & &
    $S_{32}$ & $S_{16}$ & Corollary~\ref{cor:mod}\\ \cline{1-3} \cline{5-7} \cline{9-11} 
$S_9$ & $\overline{H_4}$ & Corollary~\ref{cor:path} & & 
  $S_{21}$ & $\overline{H_4}$  & Corollary~\ref{cor:mod} & &
    $S_{33}$ & a graph in $\mathcal{F}_7$ & Corollary~\ref{cor:mod}\\ \cline{1-3} \cline{5-7} \cline{9-11} 
$S_{10}$ & $C_4$ & Corollary~\ref{cor:mod} & & 
  $S_{22}$ & a graph in $\overline{\mathcal{F}_6}$  & Corollary~\ref{cor:path} & &
    $S_{34}$ & a graph in $\mathcal{F}_7$ & Corollary~\ref{cor:mod}\\ \cline{1-3} \cline{5-7} \cline{9-11} 
$S_{11}$ & $\overline{H_7}$ & Corollary~\ref{cor:mod} & & 
  $S_{23}$ & $\overline{A_7}$  & Corollary~\ref{cor:mod} & &
    $S_{35}$ & a graph in $\mathcal{X}_D$ & Lemma~\ref{lem:z12281}\\ \cline{1-3} \cline{5-7} \cline{9-11} 
$S_{12}$ & $S_2$ & Corollary~\ref{cor:mod} & & 
  $S_{24}$ & $\overline{S_7}$  & Corollary~\ref{cor:mod} & &
     $S_{36}$ & $S_{14}$ & Corollary~\ref{cor:mod}\\ \cline{1-3} \cline{5-7} \cline{9-11} 
\end{tabular}}
\end{center}
\caption{Summary of results in Section~\ref{sec:smaller-gang} handling graphs in $\mathcal{S}$}
\label{table:summary-s}
\end{figure}
\begin{figure}
\begin{center}
\scalebox{0.85}{
\begin{tabular}{| c | c | c | c | c | c | c |}
\cline{1-3} \cline{5-7}
\textbf{$H\in$} & \textbf{Simulates a graph in} & \textbf{By} & & \textbf{$H\in$} & \textbf{Simulates a graph in} & \textbf{By} \\ \cline{1-3} \cline{5-7}
$\mathcal{F}_1$ & $\{H_5\}\cup \mathcal{F}_2$ & Corollary~\ref{cor:jukt} & & 
  $\mathcal{F}_6$ & $\{H_8, D_2\}\cup \mathcal{F}_8$ & Corollary~\ref{cor:jukt} \\
  \cline{1-3} \cline{5-7}
$\mathcal{F}_2$ & $\{H_5\}$ & Corollary~\ref{cor:k1t} & & 
  $\mathcal{F}_7$ & $\{H_5\}\cup \mathcal{F}_2$ & Corollary~\ref{cor:k1tuk2} \\
  \cline{1-3} \cline{5-7}
$\mathcal{F}_3$ & $\{\overline{H_3}\}$ & Corollary~\ref{cor:k2xsk1} & & 
  $\mathcal{F}_8$ & $\{H_5\}\cup \mathcal{F}_2$ & Lemma~\ref{lem:near-uni-kt-euk1} \\
  \cline{1-3} \cline{5-7}
$\mathcal{F}_4$ & $\{H_5\}\cup \mathcal{F}_2$ & Corollary~\ref{cor:twin-star} & & 
  $\mathcal{F}_9$ & $\{\overline{H_3}\}\cup \mathcal{F}_3$ & Corollary~\ref{cor:path-jt} \\
  \cline{1-3} \cline{5-7}
$\mathcal{F}_5$ & $\{H_8, D_2\}\cup \mathcal{F}_8$ & Corollary~\ref{cor:jutk1} & & 
  $\mathcal{F}_{10}$ & $\{S_1\}\cup \mathcal{F}_1$ & Corollary~\ref{cor:path-qt} \\
  \cline{1-3} \cline{5-7}
\end{tabular}}
\end{center}
\caption{Summary of results in Section~\ref{sec:smaller-gang} handling graphs in $\mathcal{F}$}
\label{table:summary-class}

\end{figure}

%% file: reductions-1.tex
\subsection{Reductions based on Construction~\ref{con:main}}
\label{sec:red-1}
The following lemma can be proved using a straight-forward application of 
Construction~\ref{con:main}. 

\begin{lemma}
  \label{lem:jukt}
  Let $H$ be $J\cup K_t$, for some graph $J$ and integer $t\geq 1$, where the $K_t$ is induced by $V_\ell$.
  Let $V'$ be $V(H)\setminus \{v\}$, where $v$ is any vertex in the $K_t$. Let $H'$ be $H[V']$. 
  Then $H$ \simulates\ $H'$. In particular, $H$ \simulates\ $J\cup K_1$.
\end{lemma}
\begin{proof}
  Let $(G',k)$ be an instance of \HDEE. Apply Construction~\ref{con:main} on $(G',k,H,V')$ to obtain $G$.
  We claim that $(G',k)$ is a yes-instance if and only if $(G,k)$ is a yes-instance of \HEE.
  Proposition~\ref{pro:con:main-backward} proves the backward direction.
  For the forward direction, let $(G',k)$ be a yes-instance and $F'$ be a solution of $(G',k)$ of size at most $k$.
  Since $H[V_\ell]$ is $K_t$ and every satellite vertex has degree $\ell$, it can be part of only an 
  induced $K_t$ (in an induced $H$ in $G\triangle F'$). Then if $G\triangle F'$ has an induced $H$, then $G'\triangle F'$ has an induced 
  $H'$, which is a contradiction.  
\end{proof}
\begin{corollary}
  \label{cor:jukt}
  \begin{enumerate}[(i)]
      \item Let $H$ be $K_t\cup K_2$, for $t\geq 4$ ($\in \overline{\mathcal{F}_1}$). Then $H$ \simulates\ $K_t\cup K_1$ ($\in\{\overline{H_5}\}\cup\overline{\mathcal{F}_2}$).
      \item Let $H$ be $(K_t-e)\cup K_2$, for $t\geq 4$ ($\in \overline{\mathcal{F}_6}$).
      Then $H$ \simulates\ $(K_t-e)\cup K_1$ ($\in \{\overline{H_8}, \overline{D_2}\}\cup \overline{\mathcal{F}_8}$).
  \end{enumerate}
\end{corollary}

Next we consider the removal of a path of degree-2 vertices. We can prove the correctness of the reduction only under a certain uniqueness condition on the path.
\begin{lemma}
  \label{lem:path}
  Let $H$ be a graph with minimum degree two and let $p\geq 2$ be an integer such that there is a unique induced path $P$ of length $p$ with the property that all the internal vertices of the path are having degree exactly two in $H$. 
  Let $H'$ be obtained from $H$ by removing all internal vertices of $P$. Then $H$ \simulates\ $H'$.
\end{lemma}
\begin{proof}
  Let $(G',k)$ be an instance of \HDEE.
  We apply Construction~\ref{con:main} on $(G',k,H,V')$ to obtain $G$, where $V'$ is the set of vertices inducing $H'$ in $H$. We claim that $(G',k)$ is a yes-instance
  of \HDEE\ if and only if $(G,k)$ is a yes-instance of \HEE.

  Before proving the claim, we note that there exist no induced path in $H$ with length more than $p$ such that every internal vertex has degree two.
  Proposition~\ref{pro:con:main-backward} proves the backward direction of the claim. 
  For the forward direction,  
  let $(G',k)$ be a yes-instance and let $F'$ be a solution. For a contradiction, assume that $G\triangle F'$ has an induced $H$ with a vertex set $V''$. 
  If there is no satellite vertex in $V''$, then clearly, $G'\triangle F'$ has an induced $H'$, a contradiction. Therefore, $V''$ contains at least one satellite vertex $v_j\in V_j$
  for some satellite $V_j$.  
  Since the minimum degree of $H$ is two and the path $P$ (with length $p$ and having degree two for all internal vertices)
  is unique in $H$, all vertices in $V_j$ must be in $V''$ and forms the internal vertices of the path $P$ in $H$ induced by $V''$. Then $G'\triangle F'$
  has an induced $H'$, a contradiction. 
\end{proof}
\input{path-graphs}

\begin{corollary}
  \label{cor:path-jt}
  Let $H$ be $J_t$, for some $t\geq 3$ ($\in\mathcal{F}_9$).
  Then $H$ \simulates\ $K_2\vee tK_1$ ($\in\{\overline{H_3}\}\cup\mathcal{F}_3$).
\end{corollary}

\begin{corollary}
  \label{cor:path-qt}
  Let $H$ be $Q_t$, for some $t\geq 3$ ($\in\mathcal{F}_{10}$).
  Then $H$ \simulates\ $K_{2,t}$ ($\in\{S_1\}\cup\mathcal{F}_1$).
\end{corollary}

\begin{corollary}
\label{cor:path}
\begin{enumerate}[(i)]
    \item $S_5$ \simulates\ $C_4$.
    \item $S_9$ \simulates\ $\overline{H_{4}}$.
    \item $S_{15}$ \simulates\ $H_9\vee K_1$. Further $S_{15}$ \simulates\ $H_9$ (Proposition~\ref{pro:high-low}).
    \item $S_{22}$ \simulates\ $\overline{\text{diamond}\cup K_2}$ ($\in\overline{\mathcal{F}_6}$).
\end{enumerate}
\end{corollary}
Figure~\ref{fig:path-graphs} shows the graphs handled by Corollary~\ref{cor:path}.
Lemma~\ref{lem:cut} essentially says the following: 
If $H$ has vertex connectivity 1 and has a unique smallest 2-connected component which is a `leaf' in the tree formed by the 2-connected components, then 
$H$ \simulates\ a graph obtained by removing all vertices in the 2-connected component except the cut vertex.
\begin{lemma}
  \label{lem:cut}
  Let $H$ be a graph with vertex connectivity 1 and be not a complete graph. Let $\mathcal{C}$ be the set of all 2-connected components of $H$ having exactly one cut vertex of $H$. Assume that there exists a unique smallest (among $\mathcal{C}$) 2-connected component $J$ in $\mathcal{C}$. Let $v$ be the cut vertex of $H$ in $J$. Let $H'$ be $H-\{J\setminus \{v\}\}$. Then $H'$ \issimulatedby\ $H$.
\end{lemma}
\begin{proof}
  Let $(G,k)$ be an instance of \HDEE. 
  Let $G$ be obtained by applying Construction~\ref{con:main} on $(G', k, H, V(H'))$. 
  We claim that $(G',k)$ is a yes-instance \HDEE\ if and only if $(G,k)$ is a yes-instance.
  Proposition~\ref{pro:con:main-backward} proves one direction. For the other direction, assume that $(G',k)$ is a yes-instance and let $F'$ be a solution. For a contradiction, assume that $G\triangle F'$ has an $H$ induced by $U$. Since every satellite corresponds to a unique smallest 2-connected component (with exactly one cut vertex of $H$) sans the cut vertex, if a satellite has nonempty intersection with $U$ then every vertex in the satellite is in $U$ and no other satellite vertices can be in $U$. Then $G'-F'$ has an induced $H'$, which is a contradiction. 
\end{proof}
\input{cut-graphs}

\begin{corollary}
  \label{cor:cut}
  \begin{enumerate}[(i)]
      \item $S_{20}$ \simulates\ $K_{2,4}$ ($\in\mathcal{F}_1$).
      \item $S_{27}$ \simulates\ $\overline{(K_5-e)\cup K_2}$ ($\in\mathcal{F}_6$).
  \end{enumerate}
\end{corollary}

\begin{lemma}
  \label{lem:z12281}
  $S_{35}$ \simulates\ a 3-connected graph, which is not complete $(\in \mathcal{X}_D)$.
\end{lemma}
\begin{proof}
  Let $H$ be $S_{35}$.
  We will show that $H$ \simulates\ $H'$, where $H'$ is the graph shown in Figure~\ref{fig:s10231}. We observe that $H'$ is 3-connected. 
  
  Let $(G',k)$ be an instance of \HDEE. Let $G$ be obtained from $(G',k,H,V')$ by applying Construction~\ref{con:main}, 
  where $V'$ is the set all vertices of $H$ except the two adjacent degree-3 vertices. We claim that $(G',k)$ is a yes-instance of \HDEE\ if and only if $(G,k)$ is a yes-instance of \HEE. Proposition~\ref{pro:con:main-backward} proves the backward direction. For the forward direction, let $(G',k)$ be a yes-instance of \HDEE\ and let $F'$ be a solution. For a contradiction, assume that $G\triangle F'$ has an $H$ induced by $U$. Clearly, $U$ contains at least one satellite vertex. Since every satellite has a pair of adjacent degree-3 vertices, if a satellite vertex is in $U$ then the other vertex in the satellite must be in $U$. Since there is a unique pair of adjacent degree-3 vertices in $H$, the rest of the vertices in $U$ must be from the copy of $G'$ in $G$. Then $G'\triangle F'$ has an induced $H'$, a contradiction.
\end{proof}
Figure~\ref{fig:cut-graphs} shows the graphs handled by Corollary~\ref{cor:cut} and Lemma~\ref{lem:z12281}.

%% file: path-graphs.tex
\begin{figure}
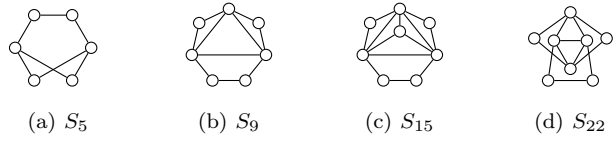

  \centering
  \begin{subfigure}[b]{0.15\textwidth}
    \centering
    \input{figs/small/z671}
    \caption{$S_5$}
    \label{fig:z671}
  \end{subfigure}%
  \begin{subfigure}[b]{0.15\textwidth}
    \centering
    \input{figs/small/z7102.tex}
    \caption{$S_9$}
    \label{fig:z7102}
  \end{subfigure}%
  \begin{subfigure}[b]{0.15\textwidth}
    \centering
    \input{figs/small/z8132}
    \caption{$S_{15}$}
    \label{fig:z8132}
  \end{subfigure}%
  \begin{subfigure}[b]{0.15\textwidth}
    \centering
    \input{figs/small/z8121}
    \caption{$S_{22}$}
    \label{fig:z8121}
  \end{subfigure}%
  \caption{Graphs handled by Corollary~\ref{cor:path}}
  \label{fig:path-graphs}
\end{figure}

%% file: cut-graphs.tex
\begin{figure}
  \centering
  \begin{subfigure}[b]{0.15\textwidth}
    \centering
    \input{figs/small/z8171c}
    \caption{$S_{20}$}
    \label{fig:z8171c}
  \end{subfigure}%
  \begin{subfigure}[b]{0.15\textwidth}
    \centering
    \input{figs/small/z9141}
    \caption{$S_{27}$}
    \label{fig:z9141}
  \end{subfigure}%
  \begin{subfigure}[b]{0.15\textwidth}
    \centering
    \input{figs/large/z12281}
    \caption{$S_{35}$}
    \label{fig:z12281}
  \end{subfigure}%
  \begin{subfigure}[b]{0.15\textwidth}
    \centering
    \input{figs/large/s10231}
    \caption{}
    \label{fig:s10231}
  \end{subfigure}%
  \caption{Graphs handled by Corollary~\ref{cor:cut} (\ref{fig:z8171c}, \ref{fig:z9141}) and Lemma~\ref{lem:z12281} (\ref{fig:z12281}, \ref{fig:s10231})}
  \label{fig:cut-graphs}
\end{figure}


%% file: reductions-2.tex
\subsection{Reductions based on Construction~\ref{con:mod}}
\label{sec:red-2}
The following is a simplified version of Construction~\ref{con:main}.

\begin{construction}
  \label{con:mod}
  Let $(G', k, \ell)$ be an input to the construction, where $G'$ is a graph and $k$ and $\ell$ are positive integers. 
  For every set $S$ of $\ell$ vertices in $G'$ introduce a clique $C$
  of $k+1$ vertices and make all the vertices in $C$ adjacent to all the vertices in $S$.
\end{construction}

As before, we call every clique $C$ introduced during the construction as a \textit{satellite} and the vertices in it as \textit{satellite vertices}.
Lemma~\ref{lem:mod} can be proved using a straight-forward application of Construction~\ref{con:mod}.
It says that if $H$ satisfies some properties, then $H$ simulates $H'$ where $H'$ is obtained by removing one vertex from 
each module of $H$ contained within $V_\ell$.

\begin{lemma}
  \label{lem:mod}
  Let $H$ be a non-regular graph such that the following conditions hold true:
  \begin{enumerate}[(i)]
  \item $1\leq \ell\leq 2, |V(H)|\geq 5$;
  \item $V_\ell$ is an independent set, $V_h\cup V_m$ induces a connected graph, and every vertex in $V_h$ is adjacent to at least one vertex in $V_\ell$;
  \item Every vertex in $V_m$ has at least $\ell+1$ neighbors outside $V_m$ or 
    there exists no pair $u,v$ of adjacent vertices in $V_m$ such that $N(u)\setminus \{v\} = N(v)\setminus \{u\}$.
  \end{enumerate} 
  Assume that $H$ admits a modular decomposition $\mathcal{M}$ such that no module in $\mathcal{M}$
  contains vertices from both $V_\ell$ and $V_m\cup V_h$.
  Let $\mathcal{M'}\subseteq \mathcal{M}$
  corresponds to $V_\ell$.  
  Let $H'$ be the graph obtained from $H$ by removing one vertex from each module in $\mathcal{M'}$. 
  Then $H$ \simulates\ $H'$.
\end{lemma}
\begin{proof}
  Let $(G',k)$ be an instance of \HDEE. We apply Construction~\ref{con:mod} on $(G',k,\ell)$ to obtain $G$.
  We claim that $(G',k)$ is a yes-instance of \HDEE\ if and only if $(G,k)$ is a yes-instance of \HEE.

  For the backward direction, let $(G,k)$ be a yes-instance of \HEE\ and let $F$ be a solution. We claim that $G'\triangle F$ is $H'$-free.
  For a contradiction, assume that $G'\triangle F'$ has an $H'$ induced by $V'$. Since every set $S$ of $\ell$
  vertices in $G'$ has a corresponding clique $C$ of $k+1$ vertices, every vertex of which is adjacent to every vertex of $S$,
  we obtain that $V'\cup V''$ induces an $H$ in $G\triangle F$, which is a contradiction, where $V''$ is a carefully chosen subset of all satellite vertices.

  To prove the forward direction, let $(G',k)$ be a yes-instance and let $F'$ be a solution. We claim that $G\triangle F'$ is $H$-free. For a contradiction, assume that $G\triangle F'$ has
  an induced $H$ with a vertex set $V'$. Let $C$ be any satellite. Since $V_\ell$ is an independent set, $|V_\ell\cap C|\leq 1$ for the $H$ induced by $V'$.
  Therefore, for every module $M'\in \mathcal{M'}$ of $V_\ell$, $M'$ contains at most one satellite vertex in a satellite.
  Since every pair of satellite has distinct neighborhood, this implies that $M'$ contains at most one satellite vertex. 
  Therefore, if $(V_h\cup V_m)\cap C= \emptyset$ for every satellite $C$, then $G'\triangle F'$ contains an
  induced $H'$, a contradiction. Therefore, assume that $|(V_h\cup V_m)\cap C|\geq 1$ for at least one satellite $C$.
  
  Case 1: $|V_h\cap C|\geq 1$ for some satellite $C$: Let $v\in V_h$ be in the $H$ induced by $V'$ in $G\triangle F'$, where $v$ is 
  in $C$ which was introduced for a set $S$ of $\ell$ vertices in $G'$. Since every 
  vertex in $C\cap V(H)$ has the same degree in the $H$, $C\cap V(H)\subseteq V_h$. 
  If $\ell = 1$, then $|S| = 1$ and $S\subseteq V_\ell$ as $v$ must be adjacent to some vertex in $V_\ell$. This gives a 
  contradiction as the vertex in $S$ has degree at least that of $v$ in the induced $H$. 
  Let $\ell=2$. Then $h\geq 3$. As $|S|=2$, we note that $v$ cannot be adjacent to more than two vertices in $V_\ell$. If $v$ is adjacent to two vertices in $V_\ell$, then $S$ is exactly the set of
  those two vertices. Then $H$ is $K_2\vee 2K_1$ which is a contradiction, as $H$ has at least five vertices. 
  So, $v$ is adjacent to exactly one vertex in $V_\ell$. Then one vertex, say $s_1\in S$ is in $V_\ell$ 
  and the other vertex $s_2\in S$ is in $V_h\cup V_m$ (if $s_2\notin V_h\cup V_m$ then $s_1$ will have degree at least that of $v$ in the $H$). If $|V_h\cap C|\geq 3$, then $s_1$ has degree at least three in the $H$, which is a contradiction. Therefore, since $h\geq 3$, $|V_h\cap C|=2$. Hence $h=3$. Since $\ell=2$ and $h=3$, we obtain that $V_m=\emptyset$. Therefore, $s_2\in V_h$. 
  We observe that $s_1$ and $s_2$ are nonadjacent in $G\triangle F'$ as otherwise the degree of $s_1$ becomes at least $h$, which is a contradiction.  
  Since every vertex in $V_h$ is adjacent to
  at least one vertex in $V_\ell$, $s_2$ is adjacent to a vertex $u$ (which is not $s_1$) in $V_\ell$. Now, $s_2$ cannot have any other neighbors in $H$ (other than $u$, and two vertices in $V_h\cap C$). Since $V_h\cup V_m$ induces a connected graph, $|V_h\cup V_m|=|V_h|=3$ and $u$ has degree at most one, which is a contradiction.
  
  Case 2: $|V_m\cap C|\geq 1$: 
  Since every vertex in $V'\cap C$ has the same degree in the induced $H$, $V'\cap C\subseteq V_m$. 
  Consider condition (iii). 
  Assume that every vertex in $V_m$ has at least $\ell+1$ neighbors outside $V_m$. 
  Then we get a contradiction, as a vertex in $C$ has only at most $\ell$ neighbors outside $C$. 
  Now assume that there is no pair $u,v$ of adjacent vertices in $V_m$ such that their neighbors outside them are same.
  Then $|V_m\cap C|=1$. 
  Then the vertex in $V_m\cap C$ has degree only at most $\ell$ in the $H$, a contradiction. 
\end{proof}
The following corollary lists many graphs 
(see Figure~\ref{fig:mod-graphs}) 
that can be handled by Lemma~\ref{lem:mod}.
\input{mod-graphs}

\begin{corollary}
\label{cor:mod}
\begin{multicols}{2}
\begin{enumerate}[(i)]
    \item $S_7$ \simulates\ $H_{9}$.    
    \item $S_8$ \simulates\ $\overline{A_1}$.
    \item $S_{10}$ \simulates\ $C_4$.
    \item $S_{11}$ \simulates\ $\overline{H_7}$. 
    \item $S_{12}$ \simulates\ $S_2$.
    \item $S_{13}$ \simulates\ $S_3$. 
    \item $\overline{S_{14}}$ \simulates\ $\overline{B_1}$. 
    \item $S_{18}$ \simulates\ $\overline{A_1}$.
    \item $S_{19}$ \simulates\ $\overline{S_3}$.
    \item $S_{21}$ \simulates\ $\overline{H_4}$.
    \item $S_{23}$ \simulates\ $\overline{A_{7}}$.
    \item $S_{24}$ \simulates\ $\overline{S_7}$.
    \item $S_{25}$ \simulates\ $\overline{H_1}$. 
    \item $S_{26}$ \simulates\ $S_8$. 
    \item $S_{28}$ \simulates\ $\overline{A_{9}}$. 
    \item $S_{29}$ \simulates\ $S_{17}$.
    \item $S_{30}$ \simulates\ $\overline{S_{19}}$.
    \item $\overline{S_{32}}$ \simulates\ $\overline{S_{16}}$.
    \item $\overline{S_{33}}$ \simulates\ $\overline{K_{1,4}\cup K_2}\in \overline{\mathcal{F}_7}$.
    \item $\overline{S_{34}}$ \simulates\ $\overline{K_{1,5}\cup K_2}\in \overline{\mathcal{F}_7}$.
    \item $\overline{S_{36}}$ \simulates\ $\overline{S_{14}}$.
\end{enumerate}
\end{multicols}
\end{corollary}

The following three corollaries are obtained by application of Lemma~\ref{lem:mod}: they show that in certain families of graphs, every member simulates the simplest member.
Corollary~\ref{cor:k1t} deals with star graphs ($K_{1,t}$). 
For every graph $H$ in this class, $V_\ell$ is a single module of the graph and $H$ \simulates\ a graph $H'$, where $H'$ is obtained by removing one vertex from $V_\ell$. Corollary~\ref{cor:k2xsk1} handles $K_2\vee sK_1$, where $V_\ell$ forms a single module of the graph. As in the previous case, $H'$ is obtained by removing one vertex from $V_\ell$.
Corollary~\ref{cor:twin-star} deals with the set of twin-star graphs ($T_{t_1,t_2}$). For every graph $H$ in this class, there are two modules of $H$ in $V_\ell$ : $t_1$ vertices adjacent to one vertex in $H-V_\ell$ and $t_2$ vertices adjacent to the other vertex in $H-V_\ell$. Then $H$ \simulates\ a graph $H'$, where $H'$ is obtained by removing one vertex each from the two modules.
\begin{corollary}[see Lemma~6.4 in \cite{AravindSS17} for a partial result]
  \label{cor:k1t}
  Let $H$ be $K_{1,t}$, for any $t\geq 5$ ($\in\mathcal{F}_2$). Let $H'$ be $K_{1,t-1}$.
  Then $H$ \simulates\ $H'$. Furthermore, $H$ \simulates\ $H_5$ ($K_{1,4}$).
\end{corollary}

\begin{corollary}[see Lemma~4.5 in \cite{AravindSS16} for a partial result]
  \label{cor:k2xsk1}
  Let $H$ be $K_2\vee sK_1$, for any $s\geq 4$ ($\in\mathcal{F}_3$) and
  let $H'$ be $K_2\vee (s-1)K_1$. Then $H$ \simulates\ $H'$. Furthermore, $H$ \simulates\ $\overline{H_3}$ ($K_2\vee 3K_1$).
\end{corollary}

\begin{corollary} [see Lemma~6.6 in \cite{AravindSS17} for a partial result]
  \label{cor:twin-star} 
  Let $H$ be a twin-star graph $T_{t_1,t_2}$, such that $t_1,t_2\geq 1$. Let $H'$ be $T_{t_1-1,t_2-1}$. Then $H$ \simulates\ $H'$.
  In particular, if $H$ is $T_{t,1}$, for some $t\geq 4$ ($\in\mathcal{F}_4$), then $H$ \simulates\ $K_{1,t}$ ($\in\{H_5\}\cup\mathcal{F}_2$). 
\end{corollary}


Next we see another application of Construction~\ref{con:mod}.
\begin{lemma}
  \label{lem:k23}
  $K_{2,3}$ ($=S_1$) \simulates\ $C_4$.
\end{lemma}
\begin{proof}
  Let $H$ be $K_{2,3}$. Let $H'$ be $C_4$. Let $(G',k)$ be an instance of \HEE. Let $G$ be constructed from $(G',k,\ell)$ by applying Construction~\ref{con:mod}. We claim that
  $(G',k)$ is a yes-instance of \HDEE\ if and only if $(G,k)$ is a yes-instance of \HEE. 

  Let $(G',k)$ be a yes-instance of \HDEE\ and let $F'$ be a solution. 
  For a contradiction, assume that $G\triangle F'$ has a $K_{2,3}$ induced by $V'$. 
  Clearly, $V'\cap C\neq \emptyset$ for some clique $C$ introduced during the construction for some set $S$ of two vertices in $G'$.
  Since $V'\cap C$ forms a clique and has the same neighborhood outside $V'\cap C$ in the $H$, $|V'\cap C|=1$. 
  Therefore, the vertex in $V'\cap C$ must be a degree-2 vertex in the $H$. Then the two vertices in $S$ act as the two highest degree vertices in the $H$.
  Then the other two degree-2 vertices must be from the copy of $G'$ in $G$ 
  (there is no other constructed clique adjacent to both the vertices in $S$) 
  and hence $G'\triangle F'$ contains an induced $C_4$, a contradiction.

  For the other direction, assume that $(G,k)$ is a yes-instance and let $F$ be a solution. For a contradiction, let $G'\triangle F$ has a $C_4$ induced by $\{v_1,v_2,v_3,v_4\}$, where
  $v_1$ and $v_3$ are nonadjacent. Since there are $k+1$ vertices in a clique $C$ adjacent to both $v_1$ and $v_3$ (due to the construction), there is at least one vertex $v$
  in $C$ adjacent to both $v_1$ and $v_3$ and not adjacent to $v_2$ and $v_4$ in $G\triangle F$. Then $G\triangle F$ has an induced $K_{2,3}$, a contradiction.
\end{proof}

%% file: mod-graphs.tex
\begin{figure}
  \centering
  \begin{subfigure}[b]{0.143\textwidth}
    \centering
    \input{figs/small/z771}
    \caption{$S_7$}
    \label{fig:z771}
  \end{subfigure}%
  \begin{subfigure}[b]{0.143\textwidth}
    \centering
    \input{figs/small/z7111}
    \caption{$S_8$}
    \label{fig:z7111}
  \end{subfigure}%
  \begin{subfigure}[b]{0.143\textwidth}
    \centering
    \input{figs/small/z7101}
    \caption{$S_{10}$}
    \label{fig:z7101}
  \end{subfigure}%
  \begin{subfigure}[b]{0.143\textwidth}
    \centering
    \input{figs/small/z781.tex}
    \caption{$S_{11}$}
    \label{fig:z781}
  \end{subfigure}%
  \begin{subfigure}[b]{0.143\textwidth}
    \centering
    \input{figs/small/z8141}
    \caption{$S_{12}$}
    \label{fig:z8141}
  \end{subfigure}%
  \begin{subfigure}[b]{0.143\textwidth}
    \centering
    \input{figs/small/z8101}
    \caption{$S_{13}$}
    \label{fig:z8101}
  \end{subfigure}%
  \begin{subfigure}[b]{0.143\textwidth}
    \centering
    \input{figs/small/z8151}
    \caption{$\overline{S_{14}}$}
    \label{fig:z8151}
  \end{subfigure}%
  
  \begin{subfigure}[b]{0.143\textwidth}
    \centering
    \input{figs/small/z8131}
    \caption{$S_{18}$}
    \label{fig:z8131}
  \end{subfigure}%
  \begin{subfigure}[b]{0.143\textwidth}
    \centering
    \input{figs/small/z891}
    \caption{$S_{19}$}
    \label{fig:z891}
  \end{subfigure}%
  \begin{subfigure}[b]{0.143\textwidth}
    \centering
    \input{figs/small/z8134}
    \caption{$S_{21}$}
    \label{fig:z8134}
  \end{subfigure}%
  \begin{subfigure}[b]{0.143\textwidth}
    \centering
    \input{figs/small/z9211c}
    \caption{$S_{23}$}
    \label{fig:z9211c}
  \end{subfigure}%
  \begin{subfigure}[b]{0.143\textwidth}
    \centering
    \input{figs/small/z9181}
    \caption{$S_{24}$}
    \label{fig:z9181}
  \end{subfigure}%
  \begin{subfigure}[b]{0.143\textwidth}
    \centering
    \input{figs/small/z9161}
    \caption{$S_{25}$}
    \label{fig:z9161}
  \end{subfigure}%
  \begin{subfigure}[b]{0.143\textwidth}
    \centering
    \input{figs/small/z9151.tex}
    \caption{$S_{26}$}
    \label{fig:z9151}
  \end{subfigure}%
  
  \begin{subfigure}[b]{0.143\textwidth}
    \centering
    \input{figs/large/z10181}
    \caption{$S_{28}$}
    \label{fig:z10181}
  \end{subfigure}%
  \begin{subfigure}[b]{0.143\textwidth}
    \centering
    \input{figs/large/z11221}
    \caption{$S_{29}$}
    \label{fig:z11221}
  \end{subfigure}%
  \begin{subfigure}[b]{0.143\textwidth}
    \centering
    \input{figs/large/z10231}
    \caption{$S_{30}$}
    \label{fig:z10231}
  \end{subfigure}%
  \begin{subfigure}[b]{0.143\textwidth}
    \centering
    \input{figs/large/z11341c}
    \caption{$\overline{S_{32}}$}
    \label{fig:z11341c}
  \end{subfigure}%
  \begin{subfigure}[b]{0.143\textwidth}
    \centering
    \input{figs/large/z10221}
    \caption{$\overline{S_{33}}$}
    \label{fig:z10221}
  \end{subfigure}%
  \begin{subfigure}[b]{0.143\textwidth}
    \centering
    \input{figs/large/z11281}
    \caption{$\overline{S_{34}}$}
    \label{fig:z11281}
  \end{subfigure}%
  \begin{subfigure}[b]{0.143\textwidth}
    \centering
    \input{figs/large/z10261c}
    \caption{$\overline{S_{36}}$}
    \label{fig:z10261c}
  \end{subfigure}%
  \caption{Graphs handled by Corollary~\ref{cor:mod}}
  \label{fig:mod-graphs}
\end{figure}

%% file: reductions-3.tex
\subsection{Reductions based on Construction~\ref{con:near-uni}}
\label{sec:red-3}
Now we give another construction that will be used in a few reductions.

\begin{construction}
  \label{con:near-uni}
  Let $(G', k, t)$ be an input to the construction, where $G'$ is a graph and $k$ and $t$ are  positive integers. 
  For every set $S$ of $t$ vertices in $G'$ introduce an independent set $I_S$ of $k+2$ vertices such that
  every vertex in $I_S$ is adjacent to every vertex in $G'$ except those in $S$. Let $\bigcup_{S\subseteq V(G'), |S|=t}I_S = I$. Let the resultant graph be $G$.
\end{construction}

\begin{lemma}
  \label{lem:near-uni}
  Let $H$ be a graph such that $V_h$ forms a clique and for every pair of vertices $u,v\in V_h$, $H-u$ is isomorphic to $H-v$. 
  Further assume that there exists no independent set $S$ of size $s\geq 2$ where each vertex in $S$ has degree at least $h-s+1$ in $H$. 
  Then $H$ \simulates\ $H-u$, where $u$ is any vertex in $V_h$.
\end{lemma}
\begin{proof}
  Let $H'$ be obtained from $H$ by deleting a vertex in $V_h$.
  Let $(G',k)$ be an instance of \HDEE. Apply Construction~\ref{con:near-uni} on $(G',k,h^*=|V(H)|-h-1)$ to obtain $G$.
  We claim that $(G',k)$ is a yes-instance of \HDEE\ if and only if $(G,k)$ is a yes-instance of \HEE.
  
  Let $(G',k)$ be a yes-instance of \HDEE\ and let $F'$ be a solution. 
  We claim that $G\triangle F'$ is $H$-free. 
  For a contradiction, assume that $U\subseteq V(G)$ induces $H$ in $G\triangle F'$. 
  Clearly, $I\cap U\neq \emptyset$. 
  Let $|U\cap I|=1$. 
  Since a vertex in $I$ is nonadjacent to only $h^*$ vertices in the copy of $G'$ in $G$, 
  we obtain that the vertex in $I\cap U$ must be a vertex in $V_h$ in the $H$. 
  Therefore, $G'\triangle F$ has an induced $H'$, a contradiction. 
  Let $|U\cap I|=s>1$. Since $U\cap I$ is an independent set and each vertex in it is nonadjacent to only at most $h^*$ vertices in the copy of $G'$ in $G$, 
  we get that each vertex in $U\cap I$ has degree at least $|V(H)|-1-(s-1)-h^*$ in the $H$ induced by $U$ in $G\triangle F'$. 
  Since $h^*=|V(H)|-h-1$, we obtain that each vertex in $U\cap I$ has degree at least $h-s+1$ in the $H$, which is a contradiction. 
  
  Let $(G,k)$ be a yes-instance and let $F$ be a solution. 
  For a contradiction, assume that $G'\triangle F$ has an $H'$ induced by $U$. 
  Let $U'\subseteq U$ be such that $|U'|=h^*$ and introducing a new vertex $u$ and making it adjacent to every vertex in $U\setminus U'$ of $(G'\triangle F')[U]$ results in $H$. 
  Since there are at least $k+1$ vertices adjacent to every vertex, except those in $U'$ in the copy of $G'$ in $G$, 
  we obtain that $G\triangle F$ has an induced $H$, a contradiction.
\end{proof}
\input{near-uni-graphs}
Figure~\ref{fig:near-uni-graphs} shows the graphs handled by Corollary~\ref{cor:near-uni}.
Sequences of reductions used by Corollary~\ref{cor:near-uni} are shown in parenthesis, unless the result is obtained by a direct application of Lemma~\ref{lem:near-uni}.

\begin{corollary}
  \label{cor:near-uni}
  \begin{enumerate}[(i)]
      \item   $S_2$ \simulates\ $\overline{H_7}$.
      \item $S_3$ \simulates\ $H_6$.
      \item   $S_{16}$ \simulates\ $S_3$ (delete a high degree vertex (Lemma~\ref{lem:near-uni}) \textrightarrow\ delete $V_\ell$ (Proposition~\ref{pro:high-low})).
      \item $S_{17}$ \simulates\ $\overline{H_1}$ (delete a high degree vertex (Lemma~\ref{lem:near-uni}) \textrightarrow\ delete $V_\ell$ (Proposition~\ref{pro:high-low})).
      \item   $S_{31}$ \simulates\ $S_3$ (delete a high degree vertex (Lemma~\ref{lem:near-uni}) \textrightarrow\ delete $V_\ell$ (Proposition~\ref{pro:high-low})\textrightarrow\ delete $V_h$ (Proposition~\ref{pro:high-low}))
  \end{enumerate}
\end{corollary}

\begin{lemma}
  \label{lem:near-uni-kt-euk1}
  Let $H$ be $(K_t-e)\cup K_1$ for $t\geq 6$ ($\overline{\mathcal{F}_8}$). Let $H'$ be $K_{t-2}\cup K_1$ ($\in\{\overline{H_5}\}\cup \overline{\mathcal{F}_2}$). Then $H$ \simulates\ $H'$.
\end{lemma}
\begin{proof}
  We observe that $H'$ is obtained by removing the two vertices with degree $t-2$ from $H$.
  Let $(G',k)$ be an instance of \HDEE. Let $G$ be obtained by applying Construction~\ref{con:near-uni} on $(G', k, 1)$. We claim that $(G',k)$ is a yes-instance of \HDEE\ if and only if $(G,k)$ is a yes-instance of \HEE.
  
  For the forward direction, let $(G',k)$ be a yes-instance of \HDEE. Let $F'$ be a solution of it. For a contradiction, assume that $G\triangle F'$ has an $H$ induced by $U$. Clearly $U\cap I\neq \emptyset$.
  Since an isolated vertex in $H$ is not adjacent to $t\geq 6$ vertices in $H$, a vertex in $I$ cannot be the isolated vertex in the $H$ induced by $U$. Therefore, $|U\cap I|\leq 2$. Let $|U\cap I=\{u\}|=1$. Since $u$ is adjacent to all except one vertex in the copy of $G'$ in $G$, $u$ must be a vertex with degree $t-1$ in the induced $H$. Then $G'\triangle F'$ contains an induced $(K_{t-1}-e)\cup K_1$ and hence a $K_{t-2}\cup K_1$, which is a contradiction (we note that $K_{t-2}\cup K_1$ is an induced subgraph of $(K_{t-1}-e)\cup K_1$). Let $|U\cap I=\{u,u'\}|=2$. Clearly, $u,u'\in I_{\{v\}}$ for some vertex $v$ in the copy of $G'$ in $G$. Then $u$ and $u'$ are the vertices with degree $t-2$ in the $H$. Therefore, $G'\triangle F'$ has an induced $K_{t-2}\cup K_1$, which is a contradiction.
  
  For the other direction, let $(G,k)$ be a yes-instance of \HEE. Let $F$ be a solution of it. For a contradiction, assume that $G'\triangle F$ contains an $H'$ induced by $U$. Let $v$ be the isolated vertex in the induced $H$. Since there are $k+2$ vertices adjacent to all vertices, except $v$, in the copy of $G'$ in $G$, at least two of them along with $U$ induces $H$ in $G\triangle F$, which is a contradiction. 
\end{proof}

%% file: near-uni-graphs.tex
\begin{figure}
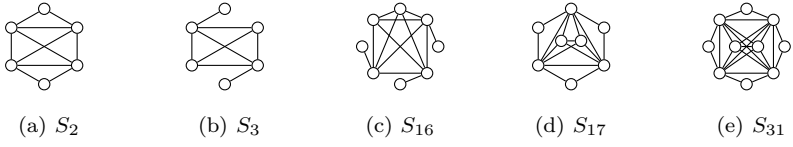

  \centering
  \begin{subfigure}[b]{0.15\textwidth}
    \centering
    \input{figs/small/c651}
    \caption{$S_2$}
    \label{fig:c651}
  \end{subfigure}
  \begin{subfigure}[b]{0.15\textwidth}
    \centering
    \input{figs/small/c671.tex}
    \caption{$S_3$}
    \label{fig:c671}
  \end{subfigure}%
  \begin{subfigure}[b]{0.15\textwidth}
    \centering
    \input{figs/small/z8133}
    \caption{$S_{16}$}
    \label{fig:z8133}
  \end{subfigure}%
  \begin{subfigure}[b]{0.15\textwidth}
    \centering
    \input{figs/small/z8122c}
    \caption{$S_{17}$}
    \label{fig:z8122c}
  \end{subfigure}
  \begin{subfigure}[b]{0.15\textwidth}
    \centering
    \input{figs/large/z10232}
    \caption{$S_{31}$}
    \label{fig:z10232}
  \end{subfigure}%
  \caption{Graphs handled by Corollary~\ref{cor:near-uni}}
  \label{fig:near-uni-graphs}
\end{figure}

%% file: reductions-4.tex
\subsection{Other reductions}
\label{sec:red-4}

To resolve graphs in $\mathcal{F}_7$ ($=K_{1,t}\cup K_2$), we resort to a known reduction. There is a PPT in \cite{AravindSS17} from \HDEE\ to \HEE, where $H'$ is a largest component in $H$. It is a composition of
two reductions: one from \HDEE\ to \HDDEE\ and another from \HDDEE\ to \HEE, where $H''$ is the union of all components in $H$ isomorphic to $H'$. The first reduction uses a simple construction (take a disjoint union of the input graph and join of $k+1$ copies of $H'$) and the second reduction uses Construction~\ref{con:main}.

\begin{proposition}[see Lemma~3.5 in \cite{AravindSS17}]
\label{pro:disconnected}
  Let $H'$ be a largest component of $H$. Then $H$ \simulates\ $H'$.
\end{proposition}
\begin{corollary}
  \label{cor:k1tuk2}
  Let $H$ be $K_{1,t}\cup K_2$, for $t\geq 4$ ($\in \mathcal{F}_7$). Then $H$ \simulates\ $K_{1,t}$ ($\in\{H_5\}\cup \mathcal{F}_2$). 
\end{corollary}

The following statement consider reduction that involve the removal of independent vertices.


\begin{lemma}
  \label{lem:jutk1}
  Let $H$ be $J\cup tK_1$, for any $t\geq 2$ such that $J$ has no component which is a clique.
  Let $H'$ be $J\cup (t-1)K_1$. Then $H$ \simulates\ $H'$. In particular, $H$ \simulates\ $J\cup K_1$.
\end{lemma}
\begin{proof}
  Let $(G',k)$ be a an instance of \HDEE. Let $G$ be $G'\cup K$, where $K$ is $K_{k+1}$. We claim that $(G',k)$ is a
  yes-instance of \HDEE\ if and only if $(G,k)$ is a yes-instance of \HEE.

  Let $(G',k)$ be a yes-instance. Let $F'$ be a solution of size at most $k$. For a contradiction
  assume that $G\triangle F'$ has an induced $J\cup tK_1$ with a vertex set $V'$. Since $G'\triangle F'$ is $H'$-free, $V'\cap K\neq \emptyset$.
  Since $V'\cap K$ induces a clique component in $G\triangle F'$ and $J$ does not have a clique component, $V'\cap K$ is a singleton
  set and induces $K_1$. Therefore $G'\triangle F'$ induces $J\cup (t-1)K_1$, which is a contradiction. For the other direction,
  let $(G,k)$ be a yes-instance and let $F$ be a solution. We claim that $G'\triangle F$ is $H'$-free. For a contradiction,
  assume that $G'\triangle F$ has an induced $H'$ with a vertex set $V'$. Since $K$ is a clique of $k+1$ vertices, 
  there exists at least one vertex $v$ in $K$ such that $v$ is not adjacent to any vertex in $V'$ in $G\triangle F$. Then $V'\cup \{v\}$ induces $H$ in $G\triangle F$, which  is a contradiction.
\end{proof}
\begin{corollary}
  \label{cor:jutk1}
  \begin{enumerate}[(i)]
      \item $S_4$ \simulates\ $H_2$.
      \item $S_6$ \simulates\ $H_6$.
      \item Let $H$ be $(K_{t}-e)\cup 2K_1$, for $t\geq 4$ ($\in\overline{\mathcal{F}_5}$). Then $H$ \simulates\ $(K_t-e)\cup K_1$ ($\in \{\overline{H_8}, \overline{D_2}\}\cup \overline{\mathcal{F}_8}$).
  \end{enumerate} 
\end{corollary}
Some graphs handled by Corollary~\ref{cor:jutk1} is shown in Figure~\ref{fig:jutk1}.
\input{jutk1-graphs}

%% file: jutk1-graphs.tex
\begin{figure}
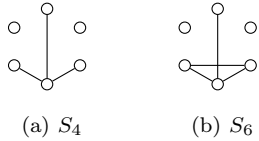

  \centering
  \begin{subfigure}[b]{0.15\textwidth}
    \centering
    \input{figs/small/claw-2k1}
    \caption{$S_4$}
    \label{fig:claw-2k1}
  \end{subfigure}%
  \begin{subfigure}[b]{0.15\textwidth}
    \centering
    \input{figs/small/paw-2k1}
    \caption{$S_{6}$}
    \label{fig:paw-2k1}
  \end{subfigure}%
  \caption{Two of the graphs handled by Corollary~\ref{cor:jutk1}}
  \label{fig:jutk1}
\end{figure}

%% file: cai.tex
\section{Incompressibility results for the  graphs in $\mathcal{A}$ and $\mathcal{B}$}
\label{sec:cai}
In this section, we prove that for every graph $H\in \mathcal{A}\cup \overline{\mathcal{A}}$, all three problems \HEE, \HED, and \HEC\ are incompressible, assuming \NOPH. With the same assumption, we prove that 
\HED\ is incompressible for every graph $H\in\mathcal{B}$; Proposition~\ref{pro:folklore} then implies incompressibility of \HEC\ for every $H\in\overline{\mathcal{B}}$.
\begin{theorem}
  \label{thm:ab}
  Assuming \NOPH:
  \begin{enumerate}[(i)]
      \item\label{thm:ab:editing} Let $H\in \mathcal{A}$. Then \HEE\ is incompressible.
      \item\label{thm:ab:deletion} Let $H\in \mathcal{A}\cup \overline{\mathcal{A}}\cup \mathcal{B}$. Then \HED\ is incompressible.
      \item\label{thm:ab:completion} Let $H\in \mathcal{A}\cup \overline{\mathcal{A}}\cup \overline{\mathcal{B}}$. Then \HEC\ is incompressible.
  \end{enumerate}
\end{theorem}



We apply the technique used by Cai and Cai~\cite{CaiC15incompressibility} by which they obtained a complete dichotomy on the incompressibility of $H$-free edge modification problems on 3-connected graphs $H$. We will give a self-contained summary of their proof technique, with only a few references to proofs of formal statements. The reader is referred to \cite{CaiC15incompressibility} for a more detailed exposition of terminology and concepts discussed in this section.

The first step in the proof is to establish incompressibility for the {\it restricted} versions of \HED\ and \HEC, where only allowed edges can be deleted/added. Then deletion and completion {\it enforcer gadgets} can be used to reduce the restricted problems to the original versions. Cai and Cai~\cite{CaiC15incompressibility} presented constructions that were proved to work correctly when $H$ is 3-connected. We show, by careful inspection, that the same technique works for certain graphs $H$ that are not 3-connected. For certain graphs $H$, we can prove incompressibility of the restricted problem, but enforcer gadgets of the required form provably do not exist. In these cases, we use ad hoc ideas to reduce the restricted version to the original one. In yet further cases, we need even trickier reductions, where we reduce \pname{$H'$-free Edge Deletion} to \HED\ for some $H'\neq H$.


\subsection{Incompressibility results for the restricted problems}

A graph is called \textit{edge-restricted} if a subset of its edges are marked as \textit{forbidden}. All edges other than forbidden are \textit{allowed}. 
A graph is called \textit{nonedge-restricted} if a subset of its nonedges are marked as \textit{forbidden}. All nonedges other than forbidden are \textit{allowed}. 

\begin{mdframed}
  \textbf{\RHED}: Given a graph $G$, an integer $k$, and a set $R$ of edges of $G$, do there exist at most $k$ edges disjoint from $R$
  such that deleting them from $G$ results in an $H$-free graph? \\
  \textbf{Parameter}: $k$
\end{mdframed}

\begin{mdframed}
  \textbf{\RHEC}: Given a graph $G$, an integer $k$, and a set $R$ of nonedges of $G$, do there exist at most $k$ nonedges disjoint from $R$
  such that adding them in $G$ results in an $H$-free graph?\\ 
  \textbf{Parameter}: $k$
\end{mdframed}

\textbf{Propagational formula satisfiability.}
A ternary Boolean function $f(x,y,z)$ (where $x,y,$ and $z$ are either Boolean variables or constants 0 or 1) is \textit{propagational} if $f(1,0,0)=0, f(0,0,0) = f(1,0,1) = f(1,1,0) = f(1,1,1) = 1$. This has the meaning: if 
$x$ is true then either $y$ is true or $z$ is true. 

\PFS: Given a conjunctive formula $\varphi$ of a propagational ternary function $f$ with distinct variables in each clause of $\varphi$, find whether 
there exists a satisfying truth assignment with weight at most $k$. The parameter we consider is $k$. 
\begin{proposition}[Theorem 3.4 in \cite{CaiC15incompressibility}]
  \label{pro:pfs}
  For any propagational ternary Boolean function $f$, \PFS\ on 3-regular conjunctive formulas (every variable appears exactly three times) admits no polynomial kernel, assuming \NOPH.
\end{proposition}

\textbf{Satisfaction-testing components.}
For \HED, a \textit{satisfaction-testing component} $S_D(x,y,z)$ is a constant-size edge-restricted $H$-free graph with exactly three allowed edges $\{x,y,z\}$
such that
there is a propagational Boolean function $f(x,y,z)$ such that $f(x,y,z)=1$ if and only if the graph obtained from $S_D(x,y,z)$ by deleting 
edges in $\{x,y,z\}$ with value 1 is $H$-free.  
For \HEC, a \textit{satisfaction-testing component} $S_C(x,y,z)$ is a constant-size nonedge-restricted $H$-free graph with exactly three allowed nonedges $\{x,y,z\}$
such that
there is a propagational Boolean function $f(x,y,z)$ such that $f(x,y,z)=1$ if and only if the graph obtained from $S_C(x,y,z)$ by adding 
edges in $\{x,y,z\}$ with value 1 is $H$-free.  

There is an easy construction (Lemma 4.3 in \cite{CaiC15incompressibility}) showing that $S_D(x,y,z)$ exists for every connected graph $H$ with at least four vertices but not complete and $S_C\{x,y,z\}$ exists for every connected graph with at least four vertices and at least two nonedges.
The construction for this is as follows. $S_D\{x,y,z\}$: Let $x$ be a nonedge, and $y$ and $z$ be two edges in $H$. Then $H+x$ is a $S_D(x,y,z)$ where $x,y,z$ are the only allowed edges. $S_C\{x,y,z\}$: Let $x$ be an edge, and $y$ and $z$ be two nonedges in $H$. Then $H-x$ is a $S_C(x,y,z)$ where $x,y,z$ are the only allowed nonedges.   

\textbf{Truth-setting components.}
For \HED, a \textit{truth-setting component} ($T_D(u)$) is a constant-sized, edge-restricted $H$-free graph such that it contains
at least three allowed edges $x,y,z$ without a common vertex and admits exactly two deletion sets $\emptyset$ and the set of all allowed edges. 
For \HEC, a \textit{truth-setting component} ($T_C(u)$) is a constant-sized, nonedge-restricted $H$-free graph such that it contains
at least three allowed nonedges $x,y,z$ without a common vertex, and admits exactly two completion sets $\emptyset$ and the set of all allowed nonedges. 

There is a construction given in \cite{CaiC15incompressibility} for $T_D(u)$ and $T_C(u)$ when $H$ is 3-connected but not complete. The constructions are given below.

Construction of $T_D(u)$: Let $e', e$ be a nonedge and an edge sharing no common vertex in $H$. Let the \textit{basic unit} $U=H+e'$ and set all edges except $e$ and $e'$ in $U$
as forbidden. Let $p$ be the number of vertices in $H$. Take $p$ copies $U_1, U_2, \ldots, U_p$ of $U$. Identify the edge $e$ of $U_i$ with the edge
$e'$ of $U_{i+1}$ to form a chain of $U$'s. This is a \textit{basic chain} $B(u)$. Let us call the unidentified edge $e'$ of $U_1$ as the 
\textit{left-most allowed edge} of $B(u)$ and unidentified edge of $U_p$ as the \textit{right-most allowed edge} of $B(u)$. Take three basic chains $B_0, B_1,$ and $B_2$. Attach them
in a cyclic fashion: Identify the right-most allowed edge of $B_i$ with the left-most allowed edge of $B_{i+1}$, where indices are taken mod $3$. This is the claimed
truth-setting component $T_D(u)$. Let us call the allowed edges thus identified as \textit{variable edges}. We note that there are exactly three variable edges in $T_D(u)$.

It is easy to see that, for every $H$, there are only two possible deletion sets in $T_D(u)$: the empty set and the set of all allowed edges. To see this, observe that if we remove any of the allowed edges, then it creates a copy of $H$ in one of the units, forcing us to remove the next allowed edge as well. However, it is not clear if these two deletion sets really make the graph $H$ free. 
As Cai and Cai \cite{CaiC15incompressibility} showed, this construction for $T_D(u)$ works correctly for 3-connected graphs $H$: Since the ``cycle'' of basic units is long enough, every subgraph having vertices from different basic units and having at most $|V(H)|$ vertices has vertex connectivity at most 2. In general, the construction may not give correct truth-setting components for 2-connected graphs $H$. But, as we shall see later, by carefully choosing $e$ and $e'$ in these constructions, we can obtain truth-setting components for many 2-connected graphs $H$. 

Construction of $T_C(u)$: Let $e', e$ be a nonedge and an edge sharing no common vertex in $H$. Let the \textit{basic unit} $U=H-e$ and set all nonedges except $e$ and $e'$ in $U$
as forbidden. Let $p$ be the number of vertices in $H$. Take $p$ copies $U_1, U_2, \ldots, U_p$ of $U$. Identify the nonedge $e$ of $U_i$ with the nonedge
$e'$ of $U_{i+1}$ to form a chain of $U$'s. This is a \textit{basic chain} $B(u)$. Let us call the unidentified nonedge $e'$ of $U_1$ as the 
\textit{left-most allowed nonedge} of $B(u)$ and unidentified nonedge of $U_p$ as the \textit{right-most allowed nonedge} of $B(u)$. Take three basic chains $B_0, B_1,$ and $B_2$. Attach them
in a cyclic fashion: Identify the right-most allowed nonedge of $B_i$ with the left-most allowed nonedge of $B_{i+1}$, where indices are taken mod $3$. This is the claimed
truth-setting component $T_C(u)$. Let us call the allowed nonedges thus identified as \textit{variable nonedges}. We note that there are exactly three variable nonedges in $T_C(u)$. Similarly to $T_D(u)$, we can argue that for any $H$, there are only two potential completion sets (the empty set and the set of all allowed nonedges), and for 3-connected $H$, these two sets are indeed completion sets.

The following is the construction used in the reduction from \PFS\ to \RHED\ (\textsc{Completion}).
\begin{construction}
  \label{con:cai}
  Let $(\varphi,k,H)$ be an input to the construction, where $\varphi$ is a 3-regular conjunctive formula on a propagational ternary Boolean function $f$, $k$ is a positive integer, and $H$ is a graph such that there exists a satisfaction-testing component for \HED\ (\textsc{Completion}) for $f$ and a truth-setting component for \HED\ (\textsc{Completion}). The construction gives a graph $G_\varphi$, an integer $k'$, and a set of restricted (non)edges in $G_\varphi$.
  
  \begin{itemize}
      \item For every clause in $\varphi$, introduce a satisfaction-testing component $S_D(x,y,z)\ $ $(S_C(x,y,z))$ for \HED\ (\textsc{Completion}). 
      \item If $c\in \{x,y,z\}$ is 1, then the corresponding allowed (non)edge is deleted (added) and if $c=0$ then 
      the corresponding allowed (non)edge is set as forbidden.
      \item For every variable $u$ in $f$, introduce a 
      truth-setting component $T_D(u)\ $ $(T_C(u))$ for \HED\ (\textsc{Completion}).
      \item For every variable $u$, identify each of the variable (non)edges in $T_D(u)\ $ $(T_C(u))$ with an allowed (non)edge in a satisfaction-testing component corresponds to a different clause in which $u$ appears---since $\varphi$ is 3-regular, $u$ appears in exactly three clauses. 
  \end{itemize}
  Let the graph obtained be $G_\varphi$ and let $k'=3|V(H)|k$. For the deletion problem the set $R$ of forbidden edges is all the edges in $G_\varphi$ except the allowed edges in the units. For the completion problem, the set $R$ of forbidden nonedges contains every nonedge of $G_\varphi$ except the allowed nonedges in the units.  
\end{construction}

Let $H$ be a graph and $(\varphi, k)$ be an instance of a \PFS\ problem. Let $(G_\varphi,k',R)$ be the output of the Construction~\ref{con:cai} applied on $(\varphi,k,H)$.
The construction works correctly in one direction: If $(G_\varphi,k',R)$ is a yes-instance of \RHED\ (\textsc{Completion}), then $(\varphi,k)$ is a yes-instance of \PFS. To see this, let $F$ be a solution of $(G_\varphi, k', R)$. By the definition of $T_D(u)$ ($T_C(u)$), if an allowed (non)edge is in $F$ then so is every allowed (non)edge in it. Therefore, since $|F|\leq k'=3k|V(H)|$ and every truth-setting component has exactly $3|V(H)|$ many allowed (non)edges, only at most $k$ (non)edges of satisfaction-testing components are in $F$.  
By the definition of $S_D(x,y,z)$ ($S_C(x,y,z)$), if $x\in F$ then either $y$ or $z$ is in $F$, otherwise there is an induced $H$ in $G_\varphi+F$. 
Therefore, setting the variables to 1 corresponding to the (non)edges, which are part of $F$, in satisfaction-testing components, we obtain that $(\varphi,k)$ is a yes-instance of \PFS. 
Thus we have the following Proposition.

\begin{proposition}[see Lemma 5.1 in \cite{CaiC15incompressibility}]
  \label{pro:cai-easy}
  Let $(\varphi, k)$ be an instance of \PFS. 
  Let $H$ be a graph such that there exists a satisfaction-testing component for \HED\ (\textsc{Completion}) for $f$ and there exists a truth-setting component for \HED\ (\textsc{Completion}).
  Let $(G_\varphi, k',R)$ be obtained by applying Construction~\ref{con:cai} on $(\varphi,k,H)$. Then, if $(G_\varphi, k',R)$ is a yes-instance of \RHED\ (\textsc{Completion}) then $(\varphi, k)$ is a yes-instance of \PFS.
\end{proposition}

We remark that the proof of  Proposition~\ref{pro:cai-easy} works even if we use a gadget for the truth-setting component which satisfies only a weak property: it has at most two deletion (completion) sets, the $\emptyset$ and the set of all allowed (non)edges. As we have seen, the construction of $T_D(u)$ and $T_C(u)$ discussed above satisfies this weak property.

To prove the other direction, one needs to show that there is no induced $H$ in the ``vicinity'' of a satisfaction-testing component after deleting (adding) the (non)edges corresponding to the variables being set to 1 in $\varphi$. This can be done very easily for 3-connected graphs $H$. Proving this direction for 2-connected graphs $H$ (if provable) requires careful structural analysis of the constructed graph $G_\varphi$. 


In Figure~\ref{table:gadgets}, we give various gadgets required for the proofs of this section.
We use \textit{unit} as a general term to refer to a satisfaction-testing component or a basic unit.

\input{tables/gadgets}

\begin{lemma}
  \label{lem:r-a-del}
  Let $H\in \{\overline{A_1}, \overline{A_2}, A_3, \overline{A_3}, A_4, A_5, \overline{A_7}, \overline{A_{9}}\}$. Then \RHED\ is incompressible, assuming \NOPH.
\end{lemma}
\begin{proof}
 By definition, the gadget shown in the corresponding cell in the column `$S_D(x,y,z)$' (where $x$ is the edge added to $H$, and $y$ and $z$ are the other two darkened edges) of Figure~\ref{table:gadgets} is a satisfaction-testing component for \HED\ and the gadget shown in the corresponding cell in the column `Basic unit' (under \textsc{Deletion}) (with two distinguished edges which are darkened) is a basic unit for \HED. 
 Let $T_D(u)$ be obtained by the construction for truth-setting component using the basic unit for \HED\ (given in Figure~\ref{table:gadgets}).
 We use these $S_D(x,y,z)$ and $T_D(u)$ for the reduction given below.
 
 We give a PPT from \PFS\ to \RHED. Then the statement follows from the incompressibility of the source problem (Proposition~\ref{pro:pfs}).
 Let $(\varphi, k)$ be an instance of \PFS\ such that every variable appears exactly three times in $\varphi$. 
 We apply Construction~\ref{con:cai} on $(\varphi,k,H)$ to obtain $(G_\varphi,k',R)$. 
 We claim that $(\varphi,k)$ is a yes-instance of \PFS\ if and only if $(G_\varphi, k'=3|V(H)|k, R)$ is a yes-instance of \RHED. 
 One direction is proved by Proposition~\ref{pro:cai-easy}. 
 For the other direction,
 let $(\varphi, k)$ be a yes-instance with a satisfying truth assignment with weight at most $k$. 
 Let $F$ contain all allowed edges of $T_D(u)$ for every true variable $u$. 
 Clearly, $|F|\leq3|V(H)|k$.
 We claim that $G_\varphi' = G_\varphi-F$ is $H$-free.
 For a contradiction, assume that there is an $H$ induced by $Z$ in $G_\varphi'$.  
 Clearly, the vertices in $Z$ cannot be from a single unit. Since there is a chain of $|V(H)|$ many basic units between every pair of satisfaction-testing components, we obtain
 that the vertices of an allowed edge act as a 2-separator in the $H$ induced by $Z$. 
 We list down the arguments which lead to contradiction for each graph $H$.
  \begin{description}
  \item[$\overline{A_1}$, $A_3$, $A_4$, $A_5$, $\overline{A_9}$:] The vertices of every 2-separator of $H$ has more number of mutually adjacent common neighbors than that of vertices of every allowed edge in a unit.
  \item[$\overline{A_2}$:] The vertices $u,v$ of every 2-separator have two common neighbors $x,y$ such that $x$ and $y$ have a common neighbor non-adjacent to both $u$ and $v$. 
  Therefore, the vertices $u,v,x,y$ must be from a unit. 
  That is, the four vertices with degree at least three and form a diamond in the $H$ must be from a single unit, say $U$. 
  Now it can be seen that for every  diamond in $U-F$ neither the middle edge of the diamond nor the nonedge of the diamond is an allowed edge in $U$. 
  
  \item[$\overline{A_3}$:] Clearly, the induced paw in the $H$ must be from a single unit, say $U$ (any pair of nonadjacent vertices in a paw has a common neighbor which is adjacent to the remaining vertex). But none of the induced paw in $U-F$ has a nonedge which is an allowed edge in $U$. 
  
  \item[$\overline{A_7}$:] Every 2-separator in $H$ induces a $K_2$. In a unit, we note that if $uv$ is an allowed edge and $x$ a common neighbor of $u$ and $v$, then neither $ux$ nor $vx$ is an allowed edge. Therefore, if $uv$ acts as a 2-separator in the induced $H$, then there must be vertices $x,y,z$ in a unit, say $U$ containing $u$ and $v$ such that, in $U-F$, it must be the case that $x\in N(u)\cap N(v)$, $y\in (N(u)\cap N(x))\setminus N(v), z\in (N(v)\cap N(x))\setminus N(u)$, and $y,z$ are nonadjacent. This is not the case with vertices of any of the allowed edges in a unit.    
  \end{description}
\end{proof}
The following corollary follows from the fact that there is no subgraph isomorphic to a $C_4$ where all edges are allowed in the graph $G_\varphi$ constructed in the proof of Lemma~\ref{lem:r-a-del} for \RHED, when $H$ is a $\overline{A_1}$. We will be using this result later to handle $\overline{A_6}$ (Lemma~\ref{lem:a6c}).

\begin{corollary}
  \label{cor:r-a7c}
  Let $H$ be $\overline{A_1}$. Then, assuming \NOPH, \RHED\ is incompressible even if the input graph does not contain a subgraph (not necessarily induced) 
  isomorphic to a $C_4$ where all the edges of the $C_4$ are allowed.
\end{corollary}

\begin{lemma}
  \label{lem:r-a-com}
  Let $H\in \{\overline{A_2}, \overline{A_7}, \overline{A_8}, \overline{A_9}, \overline{B_1}, \overline{B_2}, \overline{B_3}\}$.
  Then \RHEC\ is incompressible, assuming \NOPH.
\end{lemma}
\begin{proof}
   By definition, the gadget shown in the corresponding cell in the column `$S_C(x,y,z)$' (where $x$ is the nonedge added to $H$, and $y$ and $z$ are the other two dashed nonedges) of Figure~\ref{table:gadgets} is a satisfaction-testing component for \HEC\ and the gadget shown in the corresponding cell in the column `Basic unit' (under \textsc{Completion}) (with two distinguished nonedges which are dashed) is a basic unit for \HEC. 
   Let $T_C(u)$ be obtained by the construction for truth-setting component using the basic unit for \HEC\ (given in Figure~\ref{table:gadgets}). 
 We use these $S_C(x,y,z)$ and $T_C(u)$ for the reduction given below.
 
 We give a PPT from \PFS\ to \RHEC. Then the statement follows from the incompressibility of the source problem (Proposition~\ref{pro:pfs}).
 Let $(\varphi, k)$ be an instance of \PFS\ such that every variable appears exactly three times in $\varphi$. 
 We apply Construction~\ref{con:cai} on $(\varphi,k,H)$ to obtain $(G_\varphi,k'=3|V(H)|k,R)$. 
 We claim that $(\varphi,k)$ is a yes-instance of \PFS\ if and only if $(G_\varphi, k', R)$ is a yes-instance of \RHEC. 
 One direction is proved by Proposition~\ref{pro:cai-easy}. 
 For the other direction,
 let $(\varphi, k)$ be a yes-instance with a satisfying truth assignment with weight at most $k$. Let $F$ contains all allowed nonedges of $T_C(u)$ for every true variable $u$. 
 Clearly, $|F|\leq3|V(H)|k$.
 We claim that $G_\varphi' = G_\varphi+F$ is $H$-free.
 For a contradiction, assume that there is an $H$ induced by $Z$ in $G_\varphi'$.  
 Clearly, the vertices in $Z$ cannot be from a single unit. Since there is a chain of $|V(H)|$ many basic units between every pair of satisfaction-testing components, we obtain
 that the vertices of an allowed edge act as a 2-separator in the $H$ induced by $Z$. 
 We list down the arguments which lead to contradiction for each graph $H$.
 \begin{description}
  \item[$\overline{A_2}$:] The vertices $u,v$ of every 2-separator has two common neighbors $x,y$ such that $x$ and $y$ have a common neighbor non-adjacent to both $u$ and $v$. 
  Therefore, the vertices $u,v,x,y$ must be from a unit. 
  That is, the four vertices with degree at least three and form a diamond in the $H$ must be from a single unit, say $U$. 
  Now it can be seen that the following conditions are satisfied:
  \begin{itemize}
      \item for every diamond in $U+F$ if the nonedge in the diamond is an allowed nonedge, then the middle edge of the diamond is neither 
       allowed nor its vertices have a common neighbor outside the diamond and nonadjacent to the vertices of the nonedge, and
      \item for every diamond in $U+F$ if the middle edge in the diamond is allowed, then the nonedge of the diamond is neither allowed nor
      its vertices have a common neighbor outside the diamond and nonadjacent to the vertices of the middle edge.
  \end{itemize}

  
  \item[$\overline{A_7}$:] Every 2-separator in $H$ induces a $K_2$. In a unit, we note that, if $uv$ is an allowed nonedge and $x$ a common neighbor of $u$ and $v$, then neither $ux$ nor $vx$ is an allowed nonedge. Therefore, if $uv$ acts as a 2-separator in the induced $H$ then there must be vertices $x,y,z$ in a unit, say $U$ containing $u$ and $v$ such that, in $U+F$, it must be the case that $x\in N(u)\cap N(v)$, $y\in (N(u)\cap N(x))\setminus N(v), z\in (N(v)\cap N(x))\setminus N(u)$, and $y,z$ are nonadjacent. This is not the case with vertices of any of the allowed nonedges in a unit.   

  \item[$\overline{A_8}$:] It must be the case that at least one of the following cases is true for a unit $U$: 
  \begin{itemize}
      \item the middle edge of an induced diamond in $U+F$ is an allowed nonedge in $U$;
      \item the nonedge between a degree-1 vertex and the degree-3 vertex of an induced $H'$ in $U+F$ 
      is an allowed nonedge in $U$, where $H'$ is obtained by deleting a degree-2 vertex, adjacent to another degree-2 vertex, from $H$;
      \item a side edge and a nonedge of an induced paw in $U+F$ are allowed (a side edge of a paw is an edge connecting a degree two vertex and the degree three vertex of the paw);
      \item the nonedge and one edge of an induced $P_3$ in $U+F$ are allowed;
  \end{itemize}
  The first three cases do not hold true for the units. $S_C(x,y,z)$ satisfies the fourth condition. But, in this case, an induced $H$
  will be formed only if the end points of both the allowed nonedges in the basic unit have common-neighbors, which is not the case.
  
  \item[$\overline{A_9}$, $\overline{B_2}$, $\overline{B_3}$:] The vertices of every 2-separator of $H$ has more number of mutually adjacent common neighbors than that of vertices of every allowed nonedge in a unit, even if all allowed nonedges are added to the unit.
  
  \item[$\overline{B_1}$:] Clearly, the $K_5-e$ in the $H$ must be from a single unit, say $U$. 
  Further the nonedge in the $K_5-e$ must be an allowed nonedge. But for every $K_5-e$ in $U+F$, the nonedge is not an allowed nonedge.
  \end{description}
\end{proof}


\subsection{Using enforcers to reduce to the unrestricted problems}
If we want to reduce \RHED\ to \HED, then there is a fairly natural idea to try: for each restricted edge $e'=x'y'$, we introduce a copy of $H$ on set $U$ of new vertices and identify $x'y'$ with $xy$, where $x,y\in U$ are nonadjacent vertices. Now $U$ induces a copy of $H$ plus an extra edge, but as soon as $e'$ is deleted, it becomes a copy of $H$, effectively preventing the deletion of $e'$.

There are two problems with this approach. First, the solution could delete other edges from the new copy of $H$, and then it is not necessarily true that the removal of $e'$ automatically creates an induced copy of $H$. However, this problem is easy to avoid by repeating this gadget construction $k+1$ times: a solution of size at most $k$ cannot interfere with all $k+1$ gadgets. The second problem is more serious: it is possible that attaching the new vertices creates a copy of $H$, even when $e$ is not deleted. For certain graphs $H$, with a careful choice of $x$ and $y$ we can ensure that this does not happen: no induced copy of $H$ can go through the separator ${x,y}$.

An \textit{$H$-free deletion enforcer $(X,e)$} consists of an $H$-free graph $X$ and a distinguished edge $e$ in $X$ such that (a) $X-e$ contains an induced $H$, 
and (b) for any graph $G$ vertex disjoint with $X$, and any edge $e'$ of $G$, all induced copies of $H$ in the graph obtained by attaching $X$ to $G$
through identifying $e$ with $e'$ reside entirely inside $G$. 
Similarly, an \textit{$H$-free completion enforcer $(X,e)$} consists of an $H$-free graph $X$ and
a distinguished nonedge $e$ such that (a) $X+e$ contains an induced $H$, and (b) for any graph $G$ vertex disjoint with $X$, and any nonedge $e'$ in $G$, all induced copies of $H$
in the graph obtained by attaching $X$ to $G$ through identifying $e$ with $e'$ reside entirely inside $G$.   
It can be shown that if we can come up with enforcer gadgets satisfying these conditions, then the ideas sketched above can be made to work, and we obtain a reduction from the restricted problem to the unrestricted version. 
\begin{proposition}[See Lemma~6.5 in \cite{CaiC15incompressibility}]
  \label{pro:enforcers}
  For a graph $H$:
  \begin{enumerate}[(i)]
      \item If \RHED\ is incompressible and there exists an $H$-free deletion enforcer, then \HED\ is incompressible.
      \item If \RHEC\ is incompressible and there exists an $H$-free completion enforcer, then \HEC\ is incompressible.
      \item If \HED\ is incompressible and there exists an $H$-free completion enforcer, then \HEE\ is incompressible.
  \end{enumerate} 
\end{proposition}

In the rest of the section, we establish the existence of enforcer gadgets for certain graphs $H$.
\begin{lemma}
  \label{lem:de}
  Let $H\in \{\overline{A_1}, \overline{A_2}, A_3, \overline{A_3}, A_4, A_5\}$. Then the gadget $X$ with a distinguished edge $e$ shown in the corresponding cell in the column `Enforcer' (under \textsc{Deletion}) in Figure~\ref{table:gadgets} is an $H$-free deletion enforcer.  
\end{lemma}
\begin{proof}
  Clearly, $X$ is $H$-free and $X-e$ contains an induced $H$ as required by the definition. Let $G$ be a graph vertex-disjoint with $X$. Let $G'$ be obtained from $G$ and $X$ by identifying $e$ and any edge $e'$ of $G$. Let $u,v$ be the vertices in $G'$ obtained by the identification of $e$ and $e'$. We need to prove that every induced $H$ in $G'$ is induced by a subset of vertices in $G$. For a contradiction, assume that there is an $H$ in $G'$ induced by $Z$ where $Z$ has vertices from $V(X) \setminus \{u, v\}$ and from $V(G)\setminus \{u,v\}$.  Since $H$ is 2-connected, $\{u,v\}\subseteq Z$ and $\{u,v\}$ must act as a 2-separator which induces a $K_2$ in the induced $H$. We list down arguments which lead to contradiction with the assumption for each graph $H$.
  \begin{description}
  \item[$\overline{A_1}$, $\overline{A_3}$:] None of the 2-separators in $H$ induces a $K_2$.
  \item[$\overline{A_2}$, $A_4$, $A_5$:] Every 2-separator $xy$ in $H$ has at least one common neighbor in every component 
  obtained after deleting $x$ and $y$ from $H$. But the vertices of $e$ do not have a common neighbor.
  \item[$A_3$:] Since $H$ is 2-connected, $Z$ induces a graph containing an induced $C_4$ with the vertices in $X$. But $H$ does not have an induced $C_4$.
  \end{description}
\end{proof}
\begin{lemma}
  \label{lem:ce}
  Let $H\in \{\overline{A_1}, \overline{A_2}, A_3, A_4, A_5, \overline{A_6}, \overline{A_7}, \overline{A_8}, \overline{A_{9}}, \overline{B_1}, \overline{B_2}, \overline{B_3}\}$. Then the gadget $X$ with a distinguished nonedge $e$ shown in the corresponding cell in the column `Enforcer' (under \textsc{Completion}) in Figure~\ref{table:gadgets} is an $H$-free completion enforcer.  
\end{lemma}
\begin{proof}
  Clearly, $X$ is $H$-free and $X+e$ contains an induced $H$ as required by the definition. Let $G$ be a graph vertex-disjoint with $X$. Let $G'$ be obtained from $G$ and $X$ by identifying $e$ and any nonedge $e'$ of $G$. Let $u,v$ be the vertices in $G'$ obtained by the identification of $e$ and $e'$. We need to prove that every induced $H$ in $G'$ is induced by a subset of vertices in $G$. For a contradiction, assume that there is an $H$ in $G'$ induced by $Z$ where $Z$ has vertices from $V(X) \setminus \{u, v\}$ and from $V(G)\setminus \{u,v\}$.  Since $H$ is 2-connected, $\{u,v\}\subseteq Z$ and $\{u,v\}$ must act as a 2-separator which induces a $2K_1$ in the induced $H$. We list down arguments which lead to contradiction with the assumption for each graph $H$.
  \begin{description}
  \item[$\overline{A_1}, \overline{A_2}, \overline{A_6}, \overline{A_8}$, $\overline{B_1}$:] Every 2-separator $xy$ in $H$ which induces a $2K_1$ 
  has at least one common neighbor in every component obtained after deleting $x$ and $y$ from $H$. But the vertices of $e$ does not have a common neighbor.
  \item[$A_4, A_5, \overline{A_7}, \overline{A_{9}}, \overline{B_2}, \overline{B_3}$:] There is no 2-separator in $H$ inducing a $2K_1$.
  \item[$A_3$:] Since $H$ is 2-connected, $Z$ induces a graph containing an induced $P_5$ with the vertices in $X$. But $H$ does not have an induced $P_5$ between vertices of a 2-separator inducing $2K_1$.
  \end{description}
\end{proof}
\begin{lemma}
  \label{lem:a-part1-del}
  Let $H\in \{\overline{A_1}, \overline{A_2}, A_3, \overline{A_3}, A_4, A_5\}$. Then \HED\ and \HEE\ are incompressible, assuming \NOPH.
\end{lemma}
\begin{proof}
  By Lemma~\ref{lem:r-a-del}, \RHED\ is incompressible. 
  By Lemma~\ref{lem:de}, we have $H$-free deletion enforcer. 
  Then by Proposition~\ref{pro:enforcers}, \HED\ incompressible, assuming \NOPH.
  By Lemma~\ref{lem:ce}, we have $H$-free completion enforcer, except when $H=\overline{A_3}$. Then by Proposition~\ref{pro:enforcers}, \HEE\ is incompressible,
  except when $H=\overline{A_3}$, assuming \NOPH. Incompressibility of \HEE\ when $H=\overline{A_3}$ follows from that when $H=A_3$ and Proposition~\ref{pro:folklore}.
\end{proof}

Similarly, we can prove Lemma~\ref{lem:a-part1-com}. The cases of $H$ being $A_4$ or $A_5$ follows from the fact that $H$ and $\overline{H}$ are isomorphic (see Proposition~\ref{pro:folklore}). 
\begin{lemma}
  \label{lem:a-part1-com}
  Let $H\in \{\overline{A_2}, A_4, A_5, \overline{A_7}, \overline{A_8}, \overline{A_9}, \overline{B_1}, \overline{B_2}, \overline{B_3}\}$. Then \HEC\ is incompressible, assuming \NOPH.
\end{lemma}

\subsection{Further tricky reductions}
There are graphs for which we can show that no completion/deletion enforcers, as defined in the previous section, exist (this can be checked by going through every pair $x,y$ of (non)adjacent vertices). For some of these graphs, we can find a different way of enforcing that certain edges are forbidden; typically, we introduce some vertices that are used globally by every enforcer gadget.  Furthermore, there are graphs $H$, where we were unable to obtain a reduction from \RHED\ (\textsc{Completion}), but could choose an induced subgraphs $H'\subseteq H$ and obtain a reduction from \RHDED\ (\textsc{Completion}), whose incompressibility was established earlier.
\begin{lemma}
  \label{lem:a7c}
  Let $H = \overline{A_7}$. Then \HED\ and \HEE\ are incompressible, assuming \NOPH.
\end{lemma}
\begin{proof}
  By Lemma~\ref{lem:r-a-del}, \RHED\ is incompressible. We give a PPT from \RHED\ to \HED. Then the incompressibility of \HEE\ follows from the existence of $H$-free completion enforcer (Lemma~\ref{lem:ce}).
  
  Let $(G',k,R)$ be an instance of \RHED. We obtain a graph $G$ from $G'$ as follows. Introduce a set $W$ of $k$ independent vertices such that
  all of them are adjacent to all vertices in $G'$. For every forbidden edge $e=uv\in R$, 
  introduce three sets $X_e, Y_e, Z_e$ of $k$ vertices each such that $ux, vy, wx, wy\in E(G)$ for every $x\in X_e$, 
  for every $y\in Y_e$, and for every $w\in W$. Further, for every $i$ such that $1\leq i\leq k$, $x_iz_i, y_iz_i\in E(G)$, 
  where $x_i\in X_e, y_i\in Y_e$, and $z_i\in Z_e$ (assuming a labelling of vertices in sets $X_e, Y_e$, and $Z_e$). 
  Let $X = \bigcup_{e\in R}X_e$, $Y = \bigcup_{e\in R}Y_e$, and $Z = \bigcup_{e\in R}Z_e$. 
  We add edges to make sure that the set $C=X\cup Y$ forms a clique. Further, $Z$ forms an independent set. 
  This completes the construction and let the resultant graph be $G$ 
  (see Figure~\ref{fig:a8cred}). We claim that $(G',k,R)$ is a yes-instance of \RHED\ if and only if $(G,k)$ is a yes-instance of \HED. 
  
  Let $(G,k)$ be a yes-instance. Let $F$ be a solution. Since $G'$ is an induced subgraph of $G$, $G'-F$ is $H$-free. Let $F$ contains a forbidden edge $e=uv\in R$. Then there are at least $k$ edge disjoint $H$ due to the vertices in the sets $W, X_e, Y_e$ and $Z_e$. Since $|F|\leq k$, this cannot happen. Therefore $F$ does not contain any forbidden edge and hence $F$ is a solution for $(G',k,R)$.
  
  Let $(G',k, R)$ be a yes-instance with a solution $F'$. We claim that $F'$ is a solution for $(G,k)$. 
  For a contradiction, let $U$ induce an $H$ in $G-F'$.  Then $U$ must contain at least one vertex newly introduced in $G$. 
  Since there is no pair of vertices with the same neighborhood in $H$, $|U\cap W|\leq 1$. If $U$ contains no vertex from $W$, and if $U$ contains at least one vertex from $V(G')$, then $U$ induces a graph with either a cut vertex or an induced $C_4$, which is a contradiction. If $U$ contains no vertex from $W$ and $V(G')$, then $U$ is a subgraph of a split graph where every vertex in the clique (formed by $X\cup Y$) is adjacent to exactly one vertex of degree two (a vertex in $Z$). Therefore, $U$ cannot induce $H$. Hence $|U\cap W|=1$. Let $U\cap W=\{w\}$.
  Let $e=uv$ be any forbidden edge in $G'$. 
  For every vertex $x\in X_e$, $ux$ is not the middle edge of any induced diamond (disregarding all vertices in $W$ except $w$) 
  in $G-F'$ and so is the case with every edge $vy$ for every $y\in Y_e$. 
  Therefore, neither $ux$ nor $vy$ can be part of the central triangle 
  (triangle formed by the degree-4 vertices) of the induced $H$. 
  
  Case 1: The vertex $w$ is a degree-4 vertex in the induced $H$. Let the other two vertices in the central triangle
  in the $H$ be $a,b$. Since an edge between $C$ and $G'$ cannot be an edge in the central triangle, either $a,b\in V(G')$ or $a,b\in C$. 
  Assume that $a,b\in V(G')$. Then every common neighbor, not in $W$, of $a$ and $b$ is adjacent to $w$. Hence $U$ cannot induce $H$ in $G-F$.
  Therefore, $a,b\in C$. The only possibility of having a common neighbor, not in $W$ and not adjacent to $w$, of $a$ and $b$ is when
  $a\in X_e$, $b\in Y_e$ for some forbidden edge $e = uv$. But, then the unique common neighbor ($u$), not in $W$ and nonadjacent to $b$, 
  of $a$ and $w$, and the unique
  common neighbor ($v$), not in $W$ and nonadjacent to $a$, of $b$ and $w$ are adjacent. Therefore, $U$ cannot induce $H$ in $G-F'$.
  
  
  Case 2: The vertex $w$ is a degree-2 vertex in the induced $H$. Let the other two vertices in the triangle containing $w$ in $H$ be $a,b$. 
  Clearly, $ab$ is a part of the central triangle of the induced $H$.
  Since an edge between $C$ and $G'$ cannot be an edge in the central triangle, either $a,b\in V(G')$ or $a,b\in C$.
  Assume that $a,b\in V(G')$. 
  Then every common neighbor, not in $W$, of $a$ and $b$ is adjacent to $w$. Hence $U$ cannot induce $H$ in $G-F'$.
  Therefore, $a,b\in C$. 
  The only possibility of having a common neighbor, not in $W$ and not adjacent to $w$, of $a$ and $b$ is when
  $a\in X_e$, $b\in Y_e$ for some forbidden edge $e = uv$.
  Here, such a common neighbor can only be a vertex $z$ from $Z_e$. But, there are no common neighbor of $a$ and $z$
  not adjacent to $w$. Therefore, $U$ cannot induce $H$ in $G-F'$.
\end{proof}
\begin{figure}
  \centering
  \begin{subfigure}[b]{.5\textwidth}
    \centering
    \input{figs/gadgets/a8cred}
    \caption{The gadget used to handle $\overline{A_7}$ in Lemma~\ref{lem:a8c}}
    \label{fig:a8cred}
  \end{subfigure}%
  \begin{subfigure}[b]{.5\textwidth}
    \centering
    \input{figs/gadgets/a10cred}
    \caption{The gadget used to handle $\overline{A_{9}}$ in Lemma~\ref{lem:a9c}}
    \label{fig:a10cred}
  \end{subfigure}%
  \caption{}
  \label{fig:a8ca10cred}
\end{figure}
We can handle $\overline{A_{9}}$ in a similar way.
\begin{lemma}
  \label{lem:a9c}
  Let $H = \overline{A_{9}}$. Then \HED\ and \HEE\ are incompressible, assuming \NOPH.
\end{lemma}
\begin{proof}
  By Lemma~\ref{lem:r-a-del}, \RHED\ is incompressible. We give a PPT from \RHED\ to \HED. Then the incompressibility of \HEE\ follows from the existence of $H$-free completion enforcer (Lemma~\ref{lem:ce}).
  
  Let $(G',k,R)$ be an instance of \RHED. We obtain a graph $G$ from $G'$ as follows. Introduce a set $W$ of $k$ independent vertices such that
  all of them are adjacent to all vertices in $G'$. For every forbidden edge $e=uv\in R$, 
  introduce four sets $Q_e, X_e, Y_e, Z_e$ of $k$ vertices each such that $ux, vy, wq, wx, wy\in E(G)$ for every $x\in X_e$, for every $y\in Y_e$,
  for every $q\in Q_e$, and for every $w\in W$. Further, for every $i$ such that $1\leq i\leq k$, $x_iz_i, y_iz_i, x_iq_i, y_iq_i\in E(G)$, where $q_i\in Q_e, x_i\in X_e, y_i\in Y_e$, and $z_i\in Z_e$ (assuming a labelling of vertices in sets $Q_e, X_e, Y_e$, and $Z_e$). Let $Q = \bigcup_{e\in R}Q_e$, $X = \bigcup_{e\in R}X_e$, $Y = \bigcup_{e\in R}Y_e$, and $Z = \bigcup_{e\in R}Z_e$. The set $C=X\cup Y$ forms a clique and $Q\cup Z$ forms an independent set. This completes the construction and let the resultant graph be $G$ (see Figure~\ref{fig:a10cred}). We claim that $(G',k,R)$ is a yes-instance of \RHED\ if and only if $(G,k)$ is a yes-instance of \HED. 
  
  Let $(G,k)$ be a yes-instance. Let $F$ be a solution. Since $G'$ is an induced subgraph of $G$, $G'-F$ is $H$-free. Let $F$ contains a forbidden edge $e=uv\in R$. Then there are at least $k$ edge disjoint $H$ due to the vertices in the sets $W, Q_e, X_e, Y_e$ and $Z_e$. Since $|F|\leq k$, this cannot happen. Therefore $F$ does not contain any forbidden edge and hence $F$ is a solution for $(G',k,R)$.
  
  Let $(G',k, R)$ be a yes-instance with a solution $F'$. We claim that $F'$ is a solution for $(G,k)$. 
  For a contradiction, let $U$ induces an $H$ in $G-F'$.  Then $U$ must contain at least one vertex newly introduced in $G$. 
  Since there is no pair of vertices with the same neighborhood in $H$, $|U\cap W|\leq 1$. 
  If $U$ contains no vertex from $W$, and if $U$ contains at least one vertex from $V(G')$, 
  then $U$ induces a graph with either a cut vertex or an induced $C_4$, which is a contradiction. 
  If $U$ contains no vertex from $W\cup V(G')$, then $U$ is a subgraph of a split graph where every pair of vertices in the clique 
  (formed by $X\cup Y$) has exactly two common neighbors in $Q\cup Z$ 
  (in this case, the two vertices have the same neighborhood in $Q\cup Z$) or has no common neighbors in $Q\cup Z$. 
  Then it can be verified that $U$ cannot induce $H$. 
  Hence $|U\cap W|=1$. Let $U\cap W=\{w\}$.
  Let $e=uv$ be any forbidden edge in $G'$. For every vertex $x\in X_e$, 
  $ux$ is not the middle edge of any diamond in $G-F'$, disregarding all vertices in $W$ except $w$; 
  and so is the case with every edge $vy$ for every $y\in Y_e$. 
  Therefore, neither $ux$ nor $vy$ can be part of any of the edges between the degree-5 vertices of the $H$. 
  Assume that $ux$ is an edge between a degree-5 vertex and a degree-3 vertex in the $H$.
  Every triangle containing $ux$ (disregarding vertices in $W$ except $w$) contains $w$ or a vertex in $C$.
  It implies that, an edge between $u$ and $C$ is an edge between two degree-5 vertices in the $H$, which is a case we
  already excluded.
  Therefore, none of the edges between $G'$ and $C$ can be an edge in the $K_4$ of the $H$.
  Further, a vertex in $Q$ cannot be a degree-5 vertex in the $H$ as every vertex in $Q$, disregarding the vertices in $W$ other than $w$, has adjacent to only three vertices.
  
  Case 1: The vertex $w$ is a degree-5 vertex in the induced $H$. 
  Let $a,b$ be the other two degree-5 vertices in the induced $H$.
  Since an edge between $G'$ and $C$ cannot be an edge between two degree-5 vertices, we obtain that either $a,b\in V(G')$ or $a,b\in C$ (recall that
  a vertex in $Q$ cannot be a degree-5 vertex).
  Let $a,b\in V(G')$. Since every common neighbor, other than those in $W$, of $a$ and $b$ is adjacent to $w$, $U$ does not induce an $H$. 
  Therefore, $a,b \in C$.
  Since $a$ and $b$ has a common neighbor nonadjacent to $w$ in $H$, we obtain that $a=x_i\in X_e$, $b=y_i\in Y_e$ for some forbidden edge $e=uv$.
  Since the unique common neighbor ($u$) of $w$ and $x_i$ which is not in $C$ and not adjacent to $y_i$, and the unique common neighbor ($v$) of $w$ and $y_i$, 
  which is not in $C$ and not adjacent to $x_i$, are adjacent in $G-F'$, $U$ cannot induce an $H$, which is a contradiction.
  
  Case 2: The vertex $w$ is a degree-2 vertex in the induced $H$. 
  Let the other two vertices in the triangle containing $w$ in $H$ be $a,b$. 
  Clearly, $ab$ is a part of the triangle formed by the degree-5 vertices in the induced $H$.
  Since an edge between $C$ and $G'$ cannot be an edge in that triangle, either $a,b\in V(G')$ or $a,b\in C$.
  Assume that $a,b\in V(G')$. 
  Then every common neighbor, not in $W$, of $a$ and $b$ is adjacent to $w$. Hence $U$ cannot induce $H$ in $G-F'$.
  Therefore, $a,b\in C$. But, then $a$ and $b$ do not have a common neighbor with degree at least 5 and non-adjacent to $w$.
  Therefore, $U$ cannot induce $H$ in $G-F'$.
  
  
  Case 3: The vertex $w$ is a degree-3 vertex in the induced $H$. 
  In this case, a vertex from $V(G')$ cannot be a degree-5 vertex in the induced $H$ as none of the vertices in $V(G')$
  has a neighbor, other than vertices in $W$ and nonadjacent to $w$.
  We have already obtained that a vertex from $Q$ cannot be a degree-5 vertex in the induced $H$.
  Therefore, all the three degree-5 vertices in the induced $H$ must be from $C$.
  Then all the three degree-2 vertices in the $H$ must be from $Z$ as they must be nonadjacent with $w$ (and not belong to $W$). 
  This cannot happen as every vertex in $C$ has only one neighbor in $Z$.  
\end{proof}
Let us observe that $\overline{A_1}$ can be obtained from $\overline{A_6}$ by removing a degree-2 vertex. We can reduce \pname{Restricted $\overline{A_1}$-free Edge Deletion} to  \pname{$\overline{A_6}$-free Edge Deletion}, but we need the additional assumption of Corollary~\ref{cor:r-a7c} to make this reduction work.
\begin{lemma}
  \label{lem:a6c}
  Let $H$ be $\overline{A_6}$. Then \HED\ and \HEE\ are incompressible, assuming \NOPH.
\end{lemma}
\begin{proof}
  We give a PPT from \RHDED\ to \HED, where $H'$ is $\overline{A_1}$. Then it follows that \HED\ is incompressible (Lemma~\ref{lem:r-a-del}) and \HEE\ is incompressible by Proposition~\ref{pro:enforcers}, and by the existence of $H$-free completion enforcer (Lemma~\ref{lem:ce}), assuming \NOPH. 
  
  By Corollary~\ref{cor:r-a7c}, \RHDED\ is incompressible (assuming \NOPH) even if at least one side-edge of a diamond 
  (an edge incident to the degree-2 vertex of a diamond) in every subgraph isomorphic to $H'$ in the input graph is forbidden. 
  Let $(G',k,R)$ be such an instance of \RHDED. 
  We construct a graph $G$ as follows. For every forbidden edge $e=uv\in R$, introduce three sets $X_e, Y_e, Z_e$ of $k+1$ 
  independent vertices each such that $u$ is adjacent to all vertices in $X_e\cup Y_e\cup Z_e$ and $v$ is adjacent to all vertices in $Y\cup Z$. 
  Further, for $1\leq i\leq k+1$, $x_i$ is adjacent to $y_i$ and $y_i$ is adjacent to $z_i$, where $x_i\in X_e, y_i\in Y_e, z_i\in Z_e$ 
  (assuming a labelling of the vertices in $X_e$, $Y_e$, and $Z_e$). 
  This completes the construction (see Figure~\ref{fig:a7cred}). Let the constructed graph be $G$.
  
  We claim that $(G',k,R)$ is a yes-instance of \RHDED\ if and only if $(G,k)$ is a yes-instance of \HED. 
  Let $(G,k)$ be a yes-instance of \HED. 
  Let $F$ be a solution of it. Assume that $F$ contains some forbidden edge $e=uv$ in $G'$. 
  Since $|F|\leq k$, there exists integers $i\neq j$ such that $u$ and $v$ along with $\{x_i,y_i,z_i\}$ and $z_j$ induce an $H$ in $G-F$, which is a contradiction. 
  Therefore, $F$ does not contain any forbidden edge in $G'$. For a contradiction, let $G'-F$ contains an $H$ induced by $U$. 
  Then by the assumption on $G'$, at least one side-edge of a diamond  in the $H$ is a forbidden edge, say $e=uv$. Since $|F|\leq k$, 
  there exists at least one integer $i$ such that $U$ along with $z_i$ induces an $H$ in $G-F$, which is a contradiction. 
  For the other direction, let $(G',k,R)$ be a yes-instance and let $F'$ be a solution of it. 
  For a contradiction, assume that $G-F'$ contains an $H$ induced by a set $U$ of vertices. 
  It is straight-forward to verify that none of the vertices in $X_e\cup Y_e\cup Z_e$ can be part of an induced $C_4$ in $G-F'$. 
  Therefore, all the vertices in both the induced $C_4$ in the $H$ must be in $G'$. Then $G'-F'$ has an induced $H'$, which is a contradiction.    
\end{proof}

\begin{figure}
  \centering
  \begin{subfigure}[b]{.5\textwidth}
    \centering
    \input{figs/gadgets/a7cred}
    \caption{The gadget used to handle $\overline{A_6}$ in Lemma~\ref{lem:a6c}}
    \label{fig:a7cred}
  \end{subfigure}%
  \begin{subfigure}[b]{.5\textwidth}
    \centering
    \input{figs/gadgets/a9cred}
    \caption{The gadget used to handle $\overline{A_8}$ in Lemma~\ref{lem:a8c}}
    \label{fig:a9cred}
  \end{subfigure}%
  \caption{}
  \label{fig:a7ca9cred}
\end{figure}
Graph $\overline{A_8}$ is handled in a similar way, by noting that $C_4$ can be obtained by removing two degree-2 vertices. For the reduction, we need to observe that an additional assumption can be made in the incompressibility proof for \pname{Restricted $C_4$-free Edge Deletion} given in \cite{CaiC15incompressibility}.
\begin{observation}
  \label{obs:cai-c4-d}
  Let $H$ be $C_4$. Then \RHED\ is incompressible (assuming \NOPH) even if the input graph does not contain any subgraph (not necessarily induced) $C_4$ such that all its edges are allowed. 
\end{observation}

\begin{lemma}
  \label{lem:a8c}
  Let $H$ be $\overline{A_8}$.  Then \HED\ and \HEE\ are incompressible, assuming \NOPH.
\end{lemma}
\begin{proof}
  We give a PPT from \RHDED\ to \HED\ where $H'$ is $C_4$. 
  Then the completion enforcer given by Lemma~\ref{lem:ce} implies the incompressibility for \HEE\ (by Proposition~\ref{pro:enforcers}). 
  
  By Observation~\ref{obs:cai-c4-d}, \RHDED\ is incompressible even if the input graph does not contain a $C_4$ 
  (not necessarily induced) having only allowed edges. 
  Let $(G',k, R)$ be an instance of \RHDED\ such that every subgraph $C_4$ in $G'$ has a forbidden edge. 
  For every forbidden edge $e=uv$ in $G'$, introduce three sets $X_e,Y_e,Z_e$ of $k+2$ independent 
  vertices each such that $u$ is adjacent to every vertex in $X_e$, and $v$ is adjacent to every vertex 
  in $Y_e\cup Z_e$. Further, for $1\leq i\leq k+2$, $x_iy_i$ and $x_iz_i$ are edges in the graph, for $x_i\in X_e$, $y_i\in Y_e$, and $z_i\in Z_e$ (assuming a labelling of vertices in $X_e$, $Y_e$, and $Z_e$). This completes the construction (see Figure~\ref{fig:a9cred}). Let the resultant graph be $G$. 
  
  We claim that $(G',k,R)$ is a yes-instance of \RHDED\ if and only if $(G,k,R)$ is a yes-instance of \HED. Let $(G,k)$ be a yes-instance of \HED. Let $F$ be a solution of it.
  If $F$ contains a forbidden edge $e=uv$ in $G'$, then there exists integers $i\neq j$ such that 
  $\{u,v,x_i,x_j,y_i,z_i\}$ ($x_i,x_j\in X_e, y_i\in Y_e, z_i\in Z_e$) induces an $H$ in $G-F$, which is a contradiction.  Therefore, $F$ can contain none of the forbidden edges in $G'$. Now let $G'-F$ contains a $C_4$ induced by a set $U$. Then, at least one of the edge in the $C_4$ must be a forbidden edge, say $e=uv$. Since $|F|\leq k$, there exists at least two vertices $x_i, x_j$ ($i\neq j$) in $X_e$ such that the $U$ along with $x_i$ and $x_j$ induces an $H$ in $G-F$, which is a contradiction. For the other direction, let $(G',k,R)$ be a yes-instance of \RHDED. Let $F'$ be a solution of it. We claim that $G-F'$ is $H$-free. For a contradiction, let there be an induced $H$ in $G-F'$. It is straight-forward to verify that all the vertices of the induced $C_4$ in the $H$ must be in $G'$. Therefore, $G'-F'$ has an induced $C_4$, which is a contradiction. 
\end{proof}

As $C_4$ can be obtained from $\overline{A_1}$ by removing a degree-3 vertex, we can reduce \pname{Restricted $C_4$-free Edge Completion} to \pname{$\overline{A_1}$-free Edge Completion} using the following observation on the proof of incompressibility of \textsc{$C_4$-free Edge Completion} in \cite{CaiC15incompressibility}.

\begin{observation}
  \label{obs:a1c-com}
  Let $H$ be $C_4$. Then \HEC\ is incompressible (assuming \NOPH) for inputs $(G,R,k)$ even if the following conditions are satisfied:
  \begin{enumerate}[(i)]
      \item For every forbidden nonedge $xy$, $x$ and $y$ have only at most two common neighbors in the graph obtained by adding all allowed nonedges to $G$;
      \item Let $S$ be a subset of the set of all allowed edges in $G$. If $G+S$ has an induced $C_4$, then the following conditions are satisfied:
      \begin{itemize}
          \item Let $e,e'$ be the two edges of an induced $P_3$ in the $C_4$. Then at least one of them is in $G$.
          \item If only at most one edge of the $C_4$ is an allowed nonedge in $G$,
          then one of the two nonedges in the $C_4$ is forbidden in $G$.
          \item If $G+S$ has an induced $C_4$ where two nonadjacent edges in the $C_4$ are allowed nonedges in $G$, 
          then there is an induced $C_4$ in $G+S$ where only at most one edge in the $C_4$ is an allowed nonedge in $G$. 
      \end{itemize} 
  \end{enumerate}
\end{observation}

\begin{lemma}
  \label{lem:a1c-com}
  Let $H$ be $\overline{A_1}$. Then \HEC\ is incompressible, assuming \NOPH.
\end{lemma}
\begin{proof}
  We will prove that \RHEC\ is incompressible (assuming \NOPH), then the statement follows from the existence of completion enforcer (see Lemma~\ref{lem:ce} and Proposition~\ref{pro:enforcers}). We give a PPT from \RHDEC\ where $H'$ is a $C_4$.
  
  Let $(G',R',k)$ be an instance of \RHDEC\ where $(G',R')$ satisfies the properties given in Observation~\ref{obs:a1c-com}. 
  Initialize $G$ to be $G'$ and $R$ to be $R'$. 
  For every $xy\in R'$ such that $x$ and $y$ are the end vertices of an induced $P_3$ in $G'$, 
  introduce a vertex $v$ in $G$ adjacent to $x,y,z$, where $z$ is a middle vertex of an induced $P_3$ in $G'$, 
  where $x$ and $y$ are the end vertices of the $P_3$. 
  We note that only one vertex $v$ is introduced for a forbidden nonedge $xy$ (where $x$ and $y$ are the end vertices of an induced $P_3$), 
  and $v$ is made adjacent only to $x,y$, and only one
  common neighbor $z$ of $x$ and $y$ (even if $x$ and $y$ have another common neighbor).
  Add all nonedges incident to $v$ to $R$. Let $I$ be the set of all newly introduced vertices. 
  Add all nonedges among vertices in $I$ to $R$. 
  We claim that $(G',R',k)$ is a yes-instance of \RHDEC\ if and only if $(G,R,k)$ is a yes-instance of \RHEC.  
  
  Let $(G',R',k)$ be a yes-instance of \RHDEC. Let $F'$ be a solution of it. 
  For a contradiction, assume that $G+F'$ has an $H$ induced by $U$. Clearly, $U$ contains at least one vertex, say $v$, in $I$. 
  Let $x,y,z$ be the neighbors of $v$, where $z$ is the middle vertex of the $P_3$ induced by $\{x,y,z\}$ in $G'$. 
  Since the neighborhood of $v$ forms an induced $P_3$ in $G'$, $v$ cannot be a degree-3 vertex adjacent to the degree-2 vertex in the $H$. 
  Assume that $v$ is a degree-2 vertex in the $H$. 
  Then the nonedge in the diamond in the $H$ must be $xy$. 
  Since there are only at most two common neigbors of $x,y$ in $G'+F'$ (see condition (i) in Observation~\ref{obs:a1c-com}), one of the degree-3 vertex nonadjacent to $v$ in the $H$ must be from $I$ (recall that one of the common neighbors of $x$ and $y$ is adjacent to $v$). 
  This is a contradiction, as there is only one vertex in $I$ adjacent to both $x$ and $y$. 
  Now, assume that $v$ is a degree-3 vertex nonadjacent to a degree-2 vertex in $H$. 
  Since there is no other vertex in $I$ (other than $v$) adjacent to both $x$ and $y$, the remaining vertices in the $H$ must be from the copy of $G'$ in $G$. 
  Therefore, $G'+F'$ has an induced $C_4$, a contradiction.
  
  For the other direction, let $(G,R,k)$ be a yes-instance and let $F$ be a solution. 
  For a contradiction, assume that $G'+F$ has a $C_4$ induced by $U$. 
  If the $C_4$ contains only at most one allowed nonedge in $G'$,
  then by condition (ii) of Observation~\ref{obs:a1c-com}, one of the nonedge $xy$ in the $C_4$ is forbidden in $G'$.
  By condition (i), $x$ and $y$ do not have any other common neighbors other than the other two vertices in the $C_4$.
  Then there is an induced $P_3$ formed by three vertices of the $C_4$ such that a vetex in $I$ is adjacent to all vertices in the $P_3$. 
  Hence there is an induced $H$ in $G+F$, a contradiction. 
  If two edges in the $C_4$ are allowed nonedges in $G'$, then they must be nonadjacent edges of the $C_4$ (condition (ii)).
  Then by condition (ii), 
  there is an induced $C_4$ in $G'+F$ where only at most one edge of the $C_4$ is allowed. Then the above arguments give a contradiction.
\end{proof}
$\overline{A_6}$ can be handled in a similar way.
\begin{lemma}
  \label{lem:a6c-com}
  Let $H$ be $\overline{A_6}$. Then \HEC\ is incompressible, assuming \NOPH.
\end{lemma}
\begin{proof}
  We will prove that \RHEC\ is incompressible (assuming \NOPH), then the statement follows from the existence of completion enforcer (see Lemma~\ref{lem:ce} and Proposition~\ref{pro:enforcers}). We give a PPT from \RHDEC\ where $H'$ is a $C_4$.
  
  Let $(G',R',k)$ be an instance of \RHDEC\ where $(G',R')$ satisfies the properties given in Observation~\ref{obs:a1c-com}. 
  Initialize $G$ to be $G'$ and $R$ to be $R'$. For every $xy\in R'$ such that $x$ and $y$ are the end vertices of an induced $P_3$ in $G'$, 
  introduce two adjacent vertices $u,v$ in $G$ such that $v$ is adjacent to $x,y,z$, where $z$ is the middle vertex of an induced $P_3$ in $G'$, where $x$ and $y$ are the end vertices of the $P_3$. 
  We note that only one pair of vertices $u,v$ is introduced for a forbidden nonedge $xy$ (where $x$ and $y$ are the end vertices of an induced $P_3$), 
  and $v$ is made adjacent only to $x,y$, and only one
  common neighbor $z$ of $x$ and $y$ (even if $x$ and $y$ have another common neighbor).
  Further, $u$ is adjacent to $y$. Add all nonedges incident to $u$ and $v$ to $R$. Let $I$ be the set of all newly introduced vertices. 
  Add all nonedges among vertices in $I$ to $R$. We claim that $(G',R',k)$ is a yes-instance of \RHDEC\ if and only if $(G,R,k)$ is a yes-instance of \RHEC.  
  
  Let $(G',R',k)$ be a yes-instance of \RHDEC. Let $F'$ be a solution of it. 
  For a contradiction, assume that $G+F'$ has an $H$ induced by $U$. 
  Every vertex in $H$, other than the degree-2 vertex whose neighborhood induces a $K_2$, is part of an induced $C_4$ in $H$.
  Therefore, at least one such vertex in the induced $H$ must be from $I$.
  Clearly, a degree-2 vertex in $I$ can act as only a degree-2 vertex in the $H$ whose neighborhood induces a $K_2$.
  Therefore, a degree-4 vertex $v$ in $I$ must act as a vertex part of an induced $C_4$ in the $H$. 
  Let $x,y,z$ be the neighbors of $v$ in $G'$, where $z$ is the middle vertex of the $P_3$ induced by $\{x,y,z\}$ in $G'$. 
  Since the neighborhood of $v$ in $G'$ forms an induced $P_3$ in $G'$, 
  $v$ cannot be a degree-3 or degree-4 vertex adjacent to the degree-2 vertex (whose neighborhood induces a $2K_1$) 
  in the $H$. 
  Assume that $v$ is a degree-2 vertex, whose neighborhood induces a $2K_1$, in the $H$. 
  Then the nonedge in the diamond in the $H$ formed by deleting the degree-2 vertices must be $xy$. 
  Since there are only at most two common neighbors of $x,y$ in $G'+F'$ (see condition (i) in Observation~\ref{obs:a1c-com}), 
  one of the degree-3 or degree-4 vertex nonadjacent to $v$ in the $H$ must be from $I$ 
  (recall that one of the common neighbors of $x$ and $y$ is adjacent to $v$). 
  This is a contradiction, as there is only one vertex in $I$ adjacent to both $x$ and $y$.
  Now, assume that $v$ is a degree-3 or degree-4 vertex nonadjacent to a degree-2 vertex (whose neighborhood induces a $2K_1$) in $H$. 
  Since there is no other vertex in $I$ (other than $v$) adjacent to both $x$ and $y$, $G'+F'$ has an induced $C_4$, a contradiction.
  
  For the other direction, let $(G,R,k)$ be a yes-instance and let $F$ be a solution. 
  For a contradiction, assume that $G'+F$ has a $C_4$ induced by $U$. 
  If the $C_4$ contains only at most one allowed nonedge in $G'$,
  then by condition (ii) of Observation~\ref{obs:a1c-com}, one of the nonedge $xy$ in the $C_4$ is forbidden in $G'$.
  By condition (i), $x$ and $y$ do not have any other common neighbors other than the other two vertices in the $C_4$.
  Then there is an induced $P_3$ formed by three vertices of the $C_4$ such that a vetex $v$ in $I$ is adjacent to all vertices in the $P_3$ and 
  a vertex $u\in I$ is adjacent to $v$ and $y$ (one of the end-vertices of the forbidden edge $xy$). 
  Hence there is an induced $H$ in $G+F$, a contradiction. 
  If two edges in the $C_4$ are allowed nonedges in $G'$, then they must be nonadjacent edges of the $C_4$ (condition (ii)).
  Then by condition (ii), 
  there is an induced $C_4$ in $G'+F$ where only at most one edge of the $C_4$ is allowed. Then the above arguments give a contradiction.
\end{proof}
Now, Theorem~\ref{thm:ab}(\ref{thm:ab:editing}) follows from Lemma~\ref{lem:a-part1-del}, 
\ref{lem:a6c}, \ref{lem:a7c}, \ref{lem:a8c}, \ref{lem:a9c}, 
and Proposition~\ref{pro:folklore}. 
Theorem~\ref{thm:ab}(\ref{thm:ab:deletion}) follows from Lemma~\ref{lem:a-part1-del}, 
\ref{lem:a6c}, \ref{lem:a7c}, \ref{lem:a8c}, \ref{lem:a9c}, 
\ref{lem:a-part1-com}, \ref{lem:a1c-com}, \ref{lem:a6c-com}, and Proposition~\ref{pro:folklore}. 
Theorem~\ref{thm:ab}(\ref{thm:ab:completion}) follows from Theorem~\ref{thm:ab}(\ref{thm:ab:deletion}) and Proposition~\ref{pro:folklore}. 
Theorem~\ref{thm:main-editing} follows from Lemma~\ref{lem:editing}, \ref{lem:smaller-gang}, Theorem~\ref{thm:ab}(\ref{thm:ab:editing}), and Proposition~\ref{pro:birdeye}. Similarly,  Theorem~\ref{thm:main-deletion} follows from Lemma~\ref{lem:deletion}, \ref{lem:smaller-gang}, Theorem~\ref{thm:ab}(\ref{thm:ab:deletion}), and Proposition~\ref{pro:birdeye}.



%% file: tables/gadgets.tex
\begin{figure}
\begin{center}
\scalebox{0.95}{
\begin{tabular}{|c|c|c|c|c|c|c|c|}\hline
\multicolumn{2}{|c|}{\multirow{2}{*}{Graph}} & \multicolumn{3}{|c|}{\textsc{Deletion}} & \multicolumn{3}{|c|}{\textsc{Completion}}\\\cline{3-8}
\multicolumn{2}{|l|}{} & $S_D(x,y,z)$ & Basic unit & Enforcer & $S_C(x,y,z)$ & Basic unit & Enforcer \\\hline
$\overline{A_1}$ & \input{figs/gang/c532} & \input{figs/gadgets/c532-s} & \input{figs/gadgets/c532-bu} & \input{figs/gadgets/c532-de} & &  & \input{figs/gadgets/c532-ce}\\\hline
$\overline{A_2}$ & \input{figs/gang/c661} & \input{figs/gadgets/c661-s} & \input{figs/gadgets/c661-bu} & \input{figs/gadgets/c661-de} & \input{figs/gadgets/c661-cs} & \input{figs/gadgets/c661-cbu} & \input{figs/gadgets/c661-ce}\\\hline
${A_3}$ & \input{figs/gang/673} & \input{figs/gadgets/673-s} & \input{figs/gadgets/673-bu} & \input{figs/gadgets/673-de} &  &  & \input{figs/gadgets/673-ce}\\\hline
$\overline{A_3}$ & \input{figs/gang/c673} & \input{figs/gadgets/c673-s} & \input{figs/gadgets/c673-bu} & \input{figs/gadgets/c673-de} &  &  & \\\hline
${A_4}$ & \input{figs/gang/8141} & \input{figs/gadgets/8141-s} & \input{figs/gadgets/8141-bu} & \input{figs/gadgets/8141-de} & & & \input{figs/gadgets/8141-ce}\\\hline
${A_5}$ & \input{figs/gang/9181} & \input{figs/gadgets/9181-s} & \input{figs/gadgets/9181-bu} & \input{figs/gadgets/9181-de} & & & \input{figs/gadgets/9181-ce}\\\hline
$\overline{A_{6}}$ & \input{figs/gang/c662} & & & & & & \input{figs/gadgets/c662-ce}\\\hline
$\overline{A_{7}}$ & \input{figs/gang/c663} & \input{figs/gadgets/c663-s} & \input{figs/gadgets/c663-bu} & & \input{figs/gadgets/c663-cs} & \input{figs/gadgets/c663-cbu} & \input{figs/gadgets/c663-ce}\\\hline
$\overline{A_8}$ & \input{figs/gang/c672} &  & & & \input{figs/gadgets/c672-cs} & \input{figs/gadgets/c672-cbu} & \input{figs/gadgets/c672-ce}\\\hline
$\overline{A_{9}}$ & \input{figs/gang/c791} & \input{figs/gadgets/c791-s} & \input{figs/gadgets/c791-bu} & & \input{figs/gadgets/c791-cs} & \input{figs/gadgets/c791-cbu} & \input{figs/gadgets/c791-ce}\\\hline
$\overline{B_1}$ & \input{figs/gang/clawuk2c} &  &  & & \input{figs/gadgets/clawuk2c-cs} & \input{figs/gadgets/clawuk2c-cbu} & \input{figs/gadgets/clawuk2c-ce}\\\hline
$\overline{B_2}$ & \input{figs/gang/c782} &  &  & & \input{figs/gadgets/782c-cs} & \input{figs/gadgets/782c-cbu} & \input{figs/gadgets/782c-ce}\\\hline
$\overline{B_3}$ & \input{figs/gang/z8135c} &  &  & & \input{figs/gadgets/z8135c-cs} & \input{figs/gadgets/z8135c-cbu} & \input{figs/gadgets/z8135c-ce}\\\hline
\end{tabular}}
\end{center}
\caption{Various gadgets used in the proofs of this section. In a satisfaction-testing component $S_D(x,y,z)\ (S_C(x,y,z))$, $x$ is the darkened (non)edge added(deleted) in $H$ to obtain the gadget, and $y$ and $z$ are the other two darkened (non)edges. Both the allowed (non)edges in basic units are darkened. The distinguished edge in a deletion enforcer and the distinguished nonedge in a completion enforcer are darkened.}
\label{table:gadgets}
\end{figure}

%% file: conclusion.tex
\section{Concluding Remarks}
We obtained a set $\mathcal{H}^E$ of \nog\ 5-vertex graphs such that proving the incompressibility of \HEE\ for every $H\in \mathcal{H}^E$ will lead to a complete dichotomy of the incompressibility for \HEE\ for graphs $H$ with at least five vertices. We obtained similar sets $\mathcal{H}^D$ and $\mathcal{H}^C$ ($=\overline{\mathcal{H}^D}$) of \nogd\ graphs each for \HED\ and \HEC\ respectively. 
Thus we have the following future problems.

\begin{itemize}
    \item Prove incompressibility or obtain polynomial kernel for \HEE\ for every graph $H\in \mathcal{H}^E$. 
    \item Prove incompressibility or obtain polynomial kernel for \HED\ for every graph $H\in \mathcal{H}^D$. 
\end{itemize}

As remarked in the introduction, these sets $\mathcal{H}^E$ and  $\mathcal{H}^D$ give the frontier where the possibility of existence of polynomial kernels is the highest. For some graph $H$ in these sets, if the problem admits polynomial kernel, then one needs to include $H$ in $\mathcal{Y}$ and has to analyze the few extra cases arising out of it to obtain a possibly larger \baseset\ of graphs.

There is a curious case still unresolved when $H$ has at most four vertices---the claw. It is known that \pname{Claw-free Edge Deletion} admits a polynomial kernel when the input graphs does not contain a clique of size $t$, for any fixed positive integer $t$ \cite{DBLP:journals/algorithmica/AravindSS17}. It is also known that \pname{$\{\text{claw}, \text{diamond}\}$-free Edge Deletion} admits a polynomial kernel \cite{CyganPPLW17}.

\begin{itemize}
    \item Does claw-free edge modification problems admit polynomial kernels? 
\end{itemize}

All these efforts can be seen as steps toward two larger goals: for finite sets $\mathcal{H}$ of graphs
\begin{itemize}
    \item Obtain a dichotomy on polynomial-time solvable and NP-hard cases for $\mathcal{H}$-free edge modification problems.
    \item Obtain a dichotomy on the incompressibility of $\mathcal{H}$-free edge modification problems.
\end{itemize}

As a next step towards these larger goals one may look at the case when $\mathcal{H}$ contains exactly two graphs. We hope that the reductions we introduced in this paper can be of help to obtain various hardness results in this and related settings.